\documentclass[11pt]{article}

\usepackage[utf8]{inputenc}
\usepackage[T1]{fontenc}
\usepackage{lmodern}
\usepackage[DIV=11]{typearea} 

\usepackage{microtype}

\usepackage{amssymb}
\usepackage{amsmath}
\usepackage{amsthm}
\usepackage{thmtools}

\usepackage[procnumbered,ruled,vlined,linesnumbered]{algorithm2e}

\usepackage{gnuplot-lua-tikz}

\usepackage{xcolor}
\usepackage{xspace}
\usepackage{enumitem}
\usepackage{floatpag}

\usepackage{comment}

\usepackage[
	style=alphabetic,
	backref=true,
	doi=false,
	url=false,
	maxcitenames=3,
	mincitenames=3,
	maxbibnames=10,
	minbibnames=10,
	backend=bibtex8,
	sortlocale=en_US
]{biblatex}

\usepackage[ocgcolorlinks]{hyperref} 
\usepackage{cleveref}

\makeatletter
\g@addto@macro\bfseries{\boldmath}
\makeatother

\makeatletter
\g@addto@macro\mdseries{\unboldmath}
\g@addto@macro\normalfont{\unboldmath}
\g@addto@macro\rmfamily{\unboldmath}
\g@addto@macro\upshape{\unboldmath}
\makeatother

\DeclareCiteCommand{\citem}
    {}
    {\mkbibbrackets{\bibhyperref{\usebibmacro{postnote}}}}
    {\multicitedelim}
    {}

\renewcommand*{\multicitedelim}{\addcomma\space}

\AtEveryCitekey{\clearfield{note}}

\newcommand{\myhref}[1]{%
  \iffieldundef{doi}
    {\iffieldundef{url}
       {#1}
       {\href{\strfield{url}}{#1}}}
    {\href{http://dx.doi.org/\strfield{doi}}{#1}}%
}

\DeclareFieldFormat{title}{\myhref{\mkbibemph{#1}}}
\DeclareFieldFormat
  [article,inbook,incollection,inproceedings,patent,thesis,unpublished]
  {title}{\myhref{\mkbibquote{#1\isdot}}}

\makeatletter
\AtBeginDocument{%
    \newlength{\temp@x}%
    \newlength{\temp@y}%
    \newlength{\temp@w}%
    \newlength{\temp@h}%
    \def\my@coords#1#2#3#4{%
      \setlength{\temp@x}{#1}%
      \setlength{\temp@y}{#2}%
      \setlength{\temp@w}{#3}%
      \setlength{\temp@h}{#4}%
      \adjustlengths{}%
      \my@pdfliteral{\strip@pt\temp@x\space\strip@pt\temp@y\space\strip@pt\temp@w\space\strip@pt\temp@h\space re}}%
    \ifpdf
      \typeout{In PDF mode}%
      \def\my@pdfliteral#1{\pdfliteral page{#1}}
      \def\adjustlengths{}%
    \fi
    \ifxetex
      \def\my@pdfliteral #1{}
      \def\adjustlengths{\setlength{\temp@h}{-\temp@h}\addtolength{\temp@y}{1in}\addtolength{\temp@x}{-1in}}%
    \fi%
    \def\Hy@colorlink#1{%
      \begingroup
        \ifHy@ocgcolorlinks
          \def\Hy@ocgcolor{#1}%
          \my@pdfliteral{q}%
          \my@pdfliteral{7 Tr}
        \else
          \HyColor@UseColor#1%
        \fi
    }%
    \def\Hy@endcolorlink{%
      \ifHy@ocgcolorlinks%
        \my@pdfliteral{/OC/OCPrint BDC}%
        \my@coords{0pt}{0pt}{\pdfpagewidth}{\pdfpageheight}%
        \my@pdfliteral{F}
        \my@pdfliteral{EMC/OC/OCView BDC}%
        \begingroup%
          \expandafter\HyColor@UseColor\Hy@ocgcolor%
          \my@coords{0pt}{0pt}{\pdfpagewidth}{\pdfpageheight}%
          \my@pdfliteral{F}
        \endgroup%
        \my@pdfliteral{EMC}%
        \my@pdfliteral{0 Tr}
        \my@pdfliteral{Q}%
      \fi
      \endgroup
    }%
}
\makeatother

\colorlet{DarkRed}{red!50!black}
\colorlet{DarkGreen}{green!50!black}
\colorlet{DarkBlue}{blue!50!black}

\hypersetup{
	linkcolor = DarkRed,
	citecolor = DarkGreen,
	urlcolor = DarkBlue,
	bookmarks = true,
	bookmarksnumbered = true,
	linktocpage = true
}

\allowdisplaybreaks[1]

\declaretheorem[numberwithin=section]{theorem}
\declaretheorem[numberlike=theorem]{lemma}
\declaretheorem[numberlike=theorem]{proposition}
\declaretheorem[numberlike=theorem]{corollary}
\declaretheorem[numberlike=theorem]{definition}

\DontPrintSemicolon
\SetProcNameSty{textsc}
\SetFuncSty{textsc}

\SetKw{KwAnd}{and}
\SetKw{KwBreak}{break}

\SetKwProg{Procedure}{Procedure}{}{}

\SetKwFunction{Update}{Update}
\SetKwFunction{Delete}{Delete}
\SetKwFunction{Query}{Query}
\SetKwFunction{Distance}{Distance}
\SetKwFunction{Initialize}{Initialize}
\SetKwFunction{Increase}{Increase}
\SetKwFunction{Refresh}{Refresh}
\SetKwFunction{ComputePathUnion}{ComputePathUnion}
\SetKwFunction{RemoveChildren}{RemoveChildren}
\SetKwFunction{UpdateHubLinks}{UpdateHubLinks}
\SetKwFunction{ApproximatePathUnion}{ApproximatePathUnion}

\newcommand{\dist}{\ensuremath{d}}
\newcommand{\distest}{\delta}

\newcommand{\scc}{\textsf{SCC}\xspace}
\newcommand{\sssp}{\textsf{SSSP}\xspace}
\newcommand{\sssssp}{\stsp}
\newcommand{\stsp}{\textsf{stSP}\xspace}
\newcommand{\ssr}{\textsf{SSR}\xspace}
\newcommand{\ssssr}{\str}
\newcommand{\str}{\textsf{stR}\xspace}

\newcommand{\cg}{\ensuremath{{\mathcal C}}}
\newcommand{\cP}{\mathcal{P}}
\newcommand{\cQ}{Q}
\newcommand{\jset}{\ensuremath{B}}
\newcommand{\jdist}{\ensuremath{\dist_{\jset}}}
\newcommand{\jsize}{\ensuremath{b}}
\newcommand{\pudist}{\ensuremath{\dist_{G [\cQ (s_i, t_i)]}}}

\ifdefined\ShowComment
\def\danupon#1{\marginpar{$\leftarrow$\fbox{D}}\footnote{$\Rightarrow$~{\sf #1 --Danupon}}}
\def\sebastian#1{\marginpar{$\leftarrow$\fbox{S}}\footnote{$\Rightarrow$~{\sf #1 --Sebastian}}}
\def\monika#1{\marginpar{$\leftarrow$\fbox{M}}\footnote{$\Rightarrow$~{\sf #1 --Monika}}}

\else
\def\danupon#1{}
\def\sebastian#1{}
\def\monika#1{}
\fi

\title{Sublinear-Time Decremental Algorithms for Single-Source Reachability and Shortest Paths on Directed Graphs\thanks{Preliminary versions of this paper were presented at the 46th ACM Symposium on Theory of Computing (STOC 2014)~\cite{HenzingerKNSTOC14} and the 42nd International Colloquium on Automata, Languages, and Programming (ICALP 2015)~\cite{HenzingerKNICALP15}.}}

\date{}

\author{
Monika Henzinger\thanks{University of Vienna, Faculty of Computer Science, Austria. Supported by the Austrian Science Fund (FWF): P23499-N23. The research leading to these results has received funding from the European Research Council under the European Union's Seventh Framework Programme (FP/2007-2013) / ERC Grant Agreement no. 340506 and from the European Union's Seventh Framework Programme (FP7/2007-2013) under grant agreement no.~317532.}
\and Sebastian Krinninger\thanks{Work done while at University of Vienna, Faculty of Computer Science, Austria, and while supported by M.~Henzinger's aforementioned grants and additionally the University of Vienna (IK \mbox{I049-N}). Current affiliation: Max-Planck-Institute for Informatics, Germany.}
\and Danupon Nanongkai\thanks{KTH Royal Institute of Technology, Sweden.  Work partially done while at Nanyang Technological University, Singapore, ICERM, Brown University, USA, and University of Vienna, Austria, and while supported in part by the following research grants: Nanyang Technological University grant M58110000, Singapore Ministry of Education (MOE) Academic Research Fund (AcRF) Tier 2 grant MOE2010-T2-2-082, and Singapore MOE AcRF Tier 1 grant MOE2012-T1-001-094.}
}

\hypersetup{
	pdftitle = {Sublinear-Time Decremental Algorithms for Single-Source Reachability and Shortest Paths on Directed Graphs},
	pdfauthor = {Monika Henzinger, Sebastian Krinninger, Danupon Nanongkai}
}

\begin{document}
\maketitle
\begin{abstract}
We consider dynamic algorithms for maintaining Single-Source Reachability (\ssr) and approximate Single-Source Shortest Paths (\sssp) on $n$-node $m$-edge {\em directed} graphs under edge deletions ({\em decremental algorithms}).
The previous fastest algorithm for \ssr and \sssp goes back three decades to Even and Shiloach \citem[JACM 1981]{EvenS81}; it has $O(1)$ query time and $O(mn)$ total update time (i.e., linear amortized update time if all edges are deleted).
This algorithm serves as a building block for several other dynamic algorithms. The question whether its total update time can be improved is a major, long standing, open problem.

In this paper, we answer this question affirmatively. We obtain a randomized algorithm with a total update time of $ O(\min (m^{7/6} n^{2/3 + o(1)}, m^{3/4} n^{5/4 + o(1)}) ) = O (m n^{9/10 + o(1)}) $ for \ssr and $(1+\epsilon)$-approximate \sssp if the edge weights are integers from $ 1 $ to $ W \leq 2^{\log^c{n}} $ and $ \epsilon \geq 1 / \log^c{n} $ for some constant $ c $.
We also extend our algorithm to achieve roughly the same running time for Strongly Connected Components (\scc), improving the algorithm of Roditty and Zwick \citem[FOCS 2002]{RodittyZ08}.
Our algorithm is most efficient for sparse and dense graphs.
When $ m = \Theta(n) $ its running time is $ O (n^{1 + 5/6 + o(1)}) $ and when $ m = \Theta(n^2) $ its running time is $ O (n^{2 + 3/4 + o(1)}) $.
For \ssr we also obtain an algorithm that is faster for dense graphs and has a total update time of $ O ( m^{2/3} n^{4/3 + o(1)} + m^{3/7} n^{12/7 + o(1)}) $ which is $ O (n^{2 + 2/3}) $ when $ m = \Theta(n^2) $.
All our algorithms have constant query time in the worst case and are correct with high probability against an oblivious adversary.

\end{abstract}
\newpage

\tableofcontents
\newpage

\section{Introduction}\label{sec:stoc2014:intro}

Dynamic graph algorithms are data structures that maintain a property of a dynamically changing graph, supporting both update and query operations on the graph. In {\em undirected} graphs fundamental properties such as the connected, 2-edge connected, and 2-vertex connected components as well as a minimum spanning forest can be maintained {\em very quickly}, i.e., in polylogarithmic time per operation \cite{HenzingerK99,HolmLT01,Thorup00,KapronKM13}, where an operation is either an edge insertion, an edge deletion, or a query. 
Some of these properties, such as connectivity, can even be maintained in polylogarithmic {\em worst-case} time.
More general problems, such as maintaining distances, also admit sublinear amortized time per operation as long as only edge deletions are allowed~\cite{HenzingerKNSODA14,HenzingerKNFOCS14}.

These problems when considered on {\em directed} graphs, however, become much harder. Consider, for example, a counterpart of the connectivity problem where we want to know whether there is a directed path from a node $u$ to a node $v$, i.e., whether $u$ can {\em reach} $v$. In fact, consider a very special case where we want to maintain whether a {\em fixed} node $s$ can reach any node $v$ under {\em edge deletions only}. This problem is called {\em single-source reachability} (\ssr) in the {\em decremental} setting.
It is one of the simplest, oldest, yet most useful dynamic graph problems.
It was used as a subroutine for solving other dynamic graph problems such as transitive closure~\cite{HenzingerK95,King99,BaswanaHS07,RodittyZ08} and strongly connected components (\scc)~\cite{RodittyZ08}.
Furthermore it is a special case of decremental (approximate) single-source shortest paths (\sssp)~\cite{EvenS81,BernsteinR11,HenzingerKNFOCS14}, which in turn was used as a subroutine for solving dynamic all-pairs shortest paths~\cite{King99,BaswanaHS07,BaswanaHS03,RodittyZ08,RodittyZ11,Bernstein13,HenzingerKNFOCS13,AbrahamCT14}.
%
%
Yet, no algorithm with sublinear update time was known for this problem.
This was also the case for decremental {\em $s$-$t$ reachability} (\str), where we want to maintain whether the node~$s$ can reach the node~$t$, and {$s$-$t$ shortest path} (\stsp), where we want to maintain the distance from $s$ to $t$.

\paragraph*{Related Work.}
To the best of our knowledge, the first non-trivial algorithm was published in 1971 as part of Dinitz's celebrated maximum flow algorithm; this algorithm takes $O(mn)$ total update time to solve decremental \stsp and \str (see \cite{Dinitz06} for details). In other words, it takes linear time ($O(n)$ time) per update if we delete all $m$ edges;  here, $n$ and $m$ are the number of nodes and edges, respectively.
Independently from Dinitz, Even and Shiloach \cite{EvenS81} developed an algorithm with the same total update time for \ssr and \sssp\footnote{It was actually published for undirected graphs and it was observed by Henzinger and King~\cite{HenzingerK95} that it can be easily adapted to work for directed graphs.}. This algorithm remains the fastest prior to our work.
For directed \emph{acyclic} graphs, Italiano~\cite{Italiano88} gave a decremental algorithm with a total update time of $ O (m) $.
For the \emph{incremental} version of the problem, where we only allow insertions of edges, a total update time of $ O (m) $ is sufficient in general directed graphs~\cite{Italiano86}.
For the \emph{fully dynamic} version of the problem, where both insertions and deletions of edges are allowed, Sankowski obtained an algorithm with a worst-case running time of $ O (n^{1.575}) $ \emph{per update}, resulting in a total update time of $ O (\Delta n^{1.575}) $, where $ \Delta $ is the number of updates.\footnote{Sankowski's worst-case update time for the fully dynamic single-source single-sink reachability (\ssssr) problem is $ O (n^{1.495}) $.}

King~\cite{King99} showed how to extend the algorithm of Even and Shiloach to {\em weighted} graphs, giving the first decremental single-source shortest path algorithm with total update time $O(mnW)$, where $W$ is the maximum edge weight (and all edge weights are positive integers)\footnote{The total update time is actually $O(md)$, where $d$ is the maximum distance, which could be $\Theta(nW)$.}. Using weight rounding~\cite{RaghavanT87,Cohen98,Zwick02,Bernstein09,Madry10,Bernstein13} this can be turned into a $(1+\epsilon)$-approximate single-source shortest paths algorithm with total update time $\tilde O(mn\log W)$, where the $\tilde O(\cdot)$ notation hides factors polylogarithmic in $n$.\danupon{Should we mention \cite{RodittyZ11} which shows that exact SSSP cannot be done better than $O(mn)$ (with some assumptions)? Also Patrascu's lower bound?}
The situation is similar for decremental strongly connected components: The fastest decremental \scc algorithms take total update time $O(mn)$ (\cite{RodittyZ08,Lacki13,Roditty13}). Thus many researchers in the field have asked whether the $O(mn)$ total update time for the decremental setting can be improved upon for these problems while keeping the query time constant or polylogarithmic~\cite{King08,Lacki13,RodittyZ08}.\danupon{Should we also cite \cite{RodittyZ11}? They stated that ``An interesting issue to explore is whether faster partially dynamic SSSP algorithms may be obtained if approximate answers are allowed.''}

\paragraph*{Our Results.}
We improve the previous $O(mn)$-time algorithms for decremental \ssr, approximate \sssp, and \scc in directed graphs.
We also give algorithms for \str and approximate \stsp.
We summarize our results in the following theorem.
In \Cref{fig:comparison_running_time} we compare the running times of our new algorithms with the previous solution by Even and Shiloach for different densities of the graph and \Cref{table:results summary} summarizes our results.

\begin{theorem}
There exist decremental algorithms for reachability and shortest path problems in directed graphs with the following total update times:
\begin{itemize}
\item \ssr and \scc:
$ \tilde O (\min( m^{7/6} n^{2/3}, m^{3/4} n^{5/4 + o(1)}, m^{2/3} n^{4/3 + o(1)} + m^{3/7} n^{12/7 + o(1)})) $

$ = O (m n^{9/10 + o(1)}) $.

\item \str:
$ \tilde O (\min (m^{5/4} n^{1/2}, m^{2/3} n^{4/3 + o(1)}) = O (m n^{6/7 + o(1)}) $.

\item $(1+\epsilon)$-approximate \sssp:
$ O (\min (m^{7/6} n^{2/3 + o(1)}, m^{3/4} n^{5/4 + o(1)}) ) = O (m n^{9/10 + o(1)}) $.

\item $(1+\epsilon)$-approximate \stsp:
$ O (\min (m^{5/4} n^{1/2 + o(1)}, m^{2/3} n^{4/3 + o(1)}) = O (m n^{6/7 + o(1)}) $.
\end{itemize}
The algorithms are correct with high probability against an oblivious adversary and have constant query time.
The total update time of the \scc algorithms is only in expectation.
Our algorithms can maintain $ (1+\epsilon) $-approximate shortest paths in graphs with positive integer edge weights from $ 1 $ to $ W $ if $ W \leq 2^{\log^c{n}} $ and $ \epsilon \geq 1 / \log^c{n} $ for any constant $ c $.\danupon{NEXT REVISION: Mention the dependency on $\epsilon$ or replace $1+\epsilon$ by $1+o(1)$.}
\end{theorem}

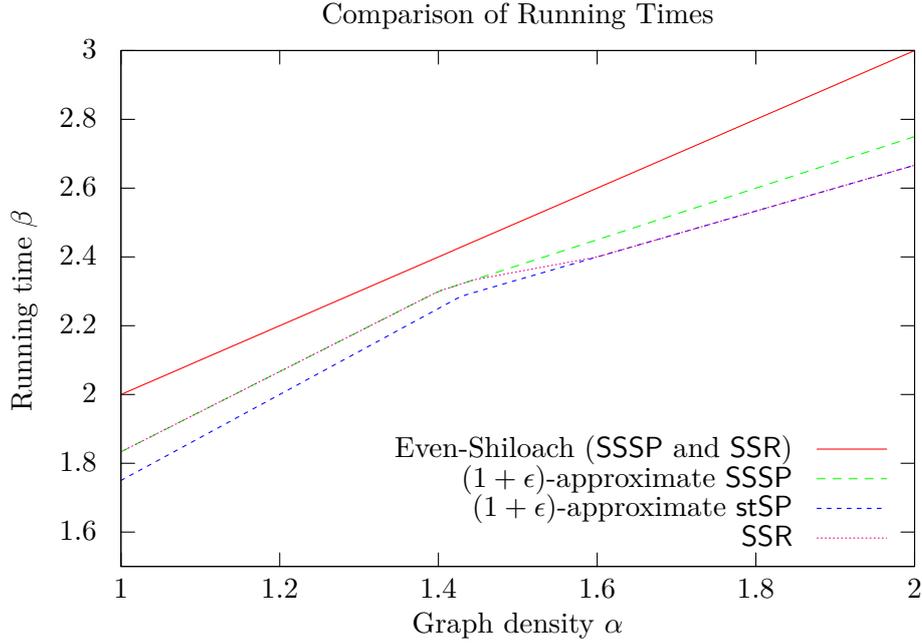
\begin{figure}[htbp]
\centering
\begin{tikzpicture}[gnuplot]
\path (0.000,0.000) rectangle (12.500,8.750);
\gpcolor{color=gp lt color border}
\gpsetlinetype{gp lt border}
\gpsetlinewidth{1.00}
\draw[gp path] (1.504,1.441)--(1.684,1.441);
\draw[gp path] (11.947,1.441)--(11.767,1.441);
\node[gp node right] at (1.320,1.441) { 1.6};
\draw[gp path] (1.504,2.353)--(1.684,2.353);
\draw[gp path] (11.947,2.353)--(11.767,2.353);
\node[gp node right] at (1.320,2.353) { 1.8};
\draw[gp path] (1.504,3.265)--(1.684,3.265);
\draw[gp path] (11.947,3.265)--(11.767,3.265);
\node[gp node right] at (1.320,3.265) { 2};
\draw[gp path] (1.504,4.177)--(1.684,4.177);
\draw[gp path] (11.947,4.177)--(11.767,4.177);
\node[gp node right] at (1.320,4.177) { 2.2};
\draw[gp path] (1.504,5.089)--(1.684,5.089);
\draw[gp path] (11.947,5.089)--(11.767,5.089);
\node[gp node right] at (1.320,5.089) { 2.4};
\draw[gp path] (1.504,6.001)--(1.684,6.001);
\draw[gp path] (11.947,6.001)--(11.767,6.001);
\node[gp node right] at (1.320,6.001) { 2.6};
\draw[gp path] (1.504,6.913)--(1.684,6.913);
\draw[gp path] (11.947,6.913)--(11.767,6.913);
\node[gp node right] at (1.320,6.913) { 2.8};
\draw[gp path] (1.504,7.825)--(1.684,7.825);
\draw[gp path] (11.947,7.825)--(11.767,7.825);
\node[gp node right] at (1.320,7.825) { 3};
\draw[gp path] (1.504,0.985)--(1.504,1.165);
\draw[gp path] (1.504,7.825)--(1.504,7.645);
\node[gp node center] at (1.504,0.677) { 1};
\draw[gp path] (3.593,0.985)--(3.593,1.165);
\draw[gp path] (3.593,7.825)--(3.593,7.645);
\node[gp node center] at (3.593,0.677) { 1.2};
\draw[gp path] (5.681,0.985)--(5.681,1.165);
\draw[gp path] (5.681,7.825)--(5.681,7.645);
\node[gp node center] at (5.681,0.677) { 1.4};
\draw[gp path] (7.770,0.985)--(7.770,1.165);
\draw[gp path] (7.770,7.825)--(7.770,7.645);
\node[gp node center] at (7.770,0.677) { 1.6};
\draw[gp path] (9.858,0.985)--(9.858,1.165);
\draw[gp path] (9.858,7.825)--(9.858,7.645);
\node[gp node center] at (9.858,0.677) { 1.8};
\draw[gp path] (11.947,0.985)--(11.947,1.165);
\draw[gp path] (11.947,7.825)--(11.947,7.645);
\node[gp node center] at (11.947,0.677) { 2};
\draw[gp path] (1.504,7.825)--(1.504,0.985)--(11.947,0.985)--(11.947,7.825)--cycle;
\node[gp node center,rotate=-270] at (0.246,4.405) {Running time $\beta$};
\node[gp node center] at (6.725,0.215) {Graph density $\alpha$};
\node[gp node center] at (6.725,8.287) {Comparison of Running Times};
\node[gp node right] at (10.479,2.541) {Even-Shiloach (\sssp and \ssr)};
\gpcolor{color=gp lt color 0}
\gpsetlinetype{gp lt plot 0}
\draw[gp path] (10.663,2.541)--(11.579,2.541);
\draw[gp path] (1.504,3.265)--(1.609,3.311)--(1.715,3.357)--(1.820,3.403)--(1.926,3.449)%
  --(2.031,3.495)--(2.137,3.541)--(2.242,3.587)--(2.348,3.633)--(2.453,3.680)--(2.559,3.726)%
  --(2.664,3.772)--(2.770,3.818)--(2.875,3.864)--(2.981,3.910)--(3.086,3.956)--(3.192,4.002)%
  --(3.297,4.048)--(3.403,4.094)--(3.508,4.140)--(3.614,4.186)--(3.719,4.232)--(3.825,4.278)%
  --(3.930,4.324)--(4.036,4.370)--(4.141,4.417)--(4.247,4.463)--(4.352,4.509)--(4.458,4.555)%
  --(4.563,4.601)--(4.669,4.647)--(4.774,4.693)--(4.880,4.739)--(4.985,4.785)--(5.090,4.831)%
  --(5.196,4.877)--(5.301,4.923)--(5.407,4.969)--(5.512,5.015)--(5.618,5.061)--(5.723,5.107)%
  --(5.829,5.153)--(5.934,5.200)--(6.040,5.246)--(6.145,5.292)--(6.251,5.338)--(6.356,5.384)%
  --(6.462,5.430)--(6.567,5.476)--(6.673,5.522)--(6.778,5.568)--(6.884,5.614)--(6.989,5.660)%
  --(7.095,5.706)--(7.200,5.752)--(7.306,5.798)--(7.411,5.844)--(7.517,5.890)--(7.622,5.937)%
  --(7.728,5.983)--(7.833,6.029)--(7.939,6.075)--(8.044,6.121)--(8.150,6.167)--(8.255,6.213)%
  --(8.361,6.259)--(8.466,6.305)--(8.571,6.351)--(8.677,6.397)--(8.782,6.443)--(8.888,6.489)%
  --(8.993,6.535)--(9.099,6.581)--(9.204,6.627)--(9.310,6.673)--(9.415,6.720)--(9.521,6.766)%
  --(9.626,6.812)--(9.732,6.858)--(9.837,6.904)--(9.943,6.950)--(10.048,6.996)--(10.154,7.042)%
  --(10.259,7.088)--(10.365,7.134)--(10.470,7.180)--(10.576,7.226)--(10.681,7.272)--(10.787,7.318)%
  --(10.892,7.364)--(10.998,7.410)--(11.103,7.457)--(11.209,7.503)--(11.314,7.549)--(11.420,7.595)%
  --(11.525,7.641)--(11.631,7.687)--(11.736,7.733)--(11.842,7.779)--(11.947,7.825);
\gpcolor{color=gp lt color border}
\node[gp node right] at (10.479,2.148) {$(1+\epsilon)$-approximate \sssp};
\gpcolor{color=gp lt color 1}
\gpsetlinetype{gp lt plot 1}
\draw[gp path] (10.663,2.148)--(11.579,2.148);
\draw[gp path] (1.504,2.505)--(1.609,2.559)--(1.715,2.612)--(1.820,2.666)--(1.926,2.720)%
  --(2.031,2.774)--(2.137,2.827)--(2.242,2.881)--(2.348,2.935)--(2.453,2.989)--(2.559,3.042)%
  --(2.664,3.096)--(2.770,3.150)--(2.875,3.204)--(2.981,3.257)--(3.086,3.311)--(3.192,3.365)%
  --(3.297,3.419)--(3.403,3.472)--(3.508,3.526)--(3.614,3.580)--(3.719,3.633)--(3.825,3.687)%
  --(3.930,3.741)--(4.036,3.795)--(4.141,3.848)--(4.247,3.902)--(4.352,3.956)--(4.458,4.010)%
  --(4.563,4.063)--(4.669,4.117)--(4.774,4.171)--(4.880,4.225)--(4.985,4.278)--(5.090,4.332)%
  --(5.196,4.386)--(5.301,4.440)--(5.407,4.493)--(5.512,4.547)--(5.618,4.601)--(5.723,4.647)%
  --(5.829,4.681)--(5.934,4.716)--(6.040,4.750)--(6.145,4.785)--(6.251,4.820)--(6.356,4.854)%
  --(6.462,4.889)--(6.567,4.923)--(6.673,4.958)--(6.778,4.992)--(6.884,5.027)--(6.989,5.061)%
  --(7.095,5.096)--(7.200,5.130)--(7.306,5.165)--(7.411,5.200)--(7.517,5.234)--(7.622,5.269)%
  --(7.728,5.303)--(7.833,5.338)--(7.939,5.372)--(8.044,5.407)--(8.150,5.441)--(8.255,5.476)%
  --(8.361,5.510)--(8.466,5.545)--(8.571,5.580)--(8.677,5.614)--(8.782,5.649)--(8.888,5.683)%
  --(8.993,5.718)--(9.099,5.752)--(9.204,5.787)--(9.310,5.821)--(9.415,5.856)--(9.521,5.890)%
  --(9.626,5.925)--(9.732,5.960)--(9.837,5.994)--(9.943,6.029)--(10.048,6.063)--(10.154,6.098)%
  --(10.259,6.132)--(10.365,6.167)--(10.470,6.201)--(10.576,6.236)--(10.681,6.270)--(10.787,6.305)%
  --(10.892,6.340)--(10.998,6.374)--(11.103,6.409)--(11.209,6.443)--(11.314,6.478)--(11.420,6.512)%
  --(11.525,6.547)--(11.631,6.581)--(11.736,6.616)--(11.842,6.650)--(11.947,6.685);
\gpcolor{color=gp lt color border}
\node[gp node right] at (10.479,1.755) {$(1+\epsilon)$-approximate \stsp};
\gpcolor{color=gp lt color 2}
\gpsetlinetype{gp lt plot 2}
\draw[gp path] (10.663,1.755)--(11.579,1.755);
\draw[gp path] (1.504,2.125)--(1.609,2.183)--(1.715,2.240)--(1.820,2.298)--(1.926,2.355)%
  --(2.031,2.413)--(2.137,2.470)--(2.242,2.528)--(2.348,2.586)--(2.453,2.643)--(2.559,2.701)%
  --(2.664,2.758)--(2.770,2.816)--(2.875,2.873)--(2.981,2.931)--(3.086,2.989)--(3.192,3.046)%
  --(3.297,3.104)--(3.403,3.161)--(3.508,3.219)--(3.614,3.277)--(3.719,3.334)--(3.825,3.392)%
  --(3.930,3.449)--(4.036,3.507)--(4.141,3.564)--(4.247,3.622)--(4.352,3.680)--(4.458,3.737)%
  --(4.563,3.795)--(4.669,3.852)--(4.774,3.910)--(4.880,3.967)--(4.985,4.025)--(5.090,4.083)%
  --(5.196,4.140)--(5.301,4.198)--(5.407,4.255)--(5.512,4.313)--(5.618,4.370)--(5.723,4.428)%
  --(5.829,4.486)--(5.934,4.543)--(6.040,4.585)--(6.145,4.616)--(6.251,4.647)--(6.356,4.678)%
  --(6.462,4.708)--(6.567,4.739)--(6.673,4.770)--(6.778,4.800)--(6.884,4.831)--(6.989,4.862)%
  --(7.095,4.892)--(7.200,4.923)--(7.306,4.954)--(7.411,4.985)--(7.517,5.015)--(7.622,5.046)%
  --(7.728,5.077)--(7.833,5.107)--(7.939,5.138)--(8.044,5.169)--(8.150,5.200)--(8.255,5.230)%
  --(8.361,5.261)--(8.466,5.292)--(8.571,5.322)--(8.677,5.353)--(8.782,5.384)--(8.888,5.414)%
  --(8.993,5.445)--(9.099,5.476)--(9.204,5.507)--(9.310,5.537)--(9.415,5.568)--(9.521,5.599)%
  --(9.626,5.629)--(9.732,5.660)--(9.837,5.691)--(9.943,5.722)--(10.048,5.752)--(10.154,5.783)%
  --(10.259,5.814)--(10.365,5.844)--(10.470,5.875)--(10.576,5.906)--(10.681,5.937)--(10.787,5.967)%
  --(10.892,5.998)--(10.998,6.029)--(11.103,6.059)--(11.209,6.090)--(11.314,6.121)--(11.420,6.151)%
  --(11.525,6.182)--(11.631,6.213)--(11.736,6.244)--(11.842,6.274)--(11.947,6.305);
\gpcolor{color=gp lt color border}
\node[gp node right] at (10.479,1.362) {\ssr};
\gpcolor{color=gp lt color 3}
\gpsetlinetype{gp lt plot 3}
\draw[gp path] (10.663,1.362)--(11.579,1.362);
\draw[gp path] (1.504,2.505)--(1.609,2.559)--(1.715,2.612)--(1.820,2.666)--(1.926,2.720)%
  --(2.031,2.774)--(2.137,2.827)--(2.242,2.881)--(2.348,2.935)--(2.453,2.989)--(2.559,3.042)%
  --(2.664,3.096)--(2.770,3.150)--(2.875,3.204)--(2.981,3.257)--(3.086,3.311)--(3.192,3.365)%
  --(3.297,3.419)--(3.403,3.472)--(3.508,3.526)--(3.614,3.580)--(3.719,3.633)--(3.825,3.687)%
  --(3.930,3.741)--(4.036,3.795)--(4.141,3.848)--(4.247,3.902)--(4.352,3.956)--(4.458,4.010)%
  --(4.563,4.063)--(4.669,4.117)--(4.774,4.171)--(4.880,4.225)--(4.985,4.278)--(5.090,4.332)%
  --(5.196,4.386)--(5.301,4.440)--(5.407,4.493)--(5.512,4.547)--(5.618,4.601)--(5.723,4.647)%
  --(5.829,4.681)--(5.934,4.716)--(6.040,4.750)--(6.145,4.785)--(6.251,4.805)--(6.356,4.824)%
  --(6.462,4.844)--(6.567,4.864)--(6.673,4.884)--(6.778,4.903)--(6.884,4.923)--(6.989,4.943)%
  --(7.095,4.963)--(7.200,4.982)--(7.306,5.002)--(7.411,5.022)--(7.517,5.042)--(7.622,5.061)%
  --(7.728,5.081)--(7.833,5.107)--(7.939,5.138)--(8.044,5.169)--(8.150,5.200)--(8.255,5.230)%
  --(8.361,5.261)--(8.466,5.292)--(8.571,5.322)--(8.677,5.353)--(8.782,5.384)--(8.888,5.414)%
  --(8.993,5.445)--(9.099,5.476)--(9.204,5.507)--(9.310,5.537)--(9.415,5.568)--(9.521,5.599)%
  --(9.626,5.629)--(9.732,5.660)--(9.837,5.691)--(9.943,5.722)--(10.048,5.752)--(10.154,5.783)%
  --(10.259,5.814)--(10.365,5.844)--(10.470,5.875)--(10.576,5.906)--(10.681,5.937)--(10.787,5.967)%
  --(10.892,5.998)--(10.998,6.029)--(11.103,6.059)--(11.209,6.090)--(11.314,6.121)--(11.420,6.151)%
  --(11.525,6.182)--(11.631,6.213)--(11.736,6.244)--(11.842,6.274)--(11.947,6.305);
\gpcolor{color=gp lt color border}
\gpsetlinetype{gp lt border}
\draw[gp path] (1.504,7.825)--(1.504,0.985)--(11.947,0.985)--(11.947,7.825)--cycle;
\gpdefrectangularnode{gp plot 1}{\pgfpoint{1.504cm}{0.985cm}}{\pgfpoint{11.947cm}{7.825cm}}
\end{tikzpicture}
\caption{Running times of our decremental reachability and approximate shortest paths algorithms as a function of the density of the initial graph in comparison to the Even-Shiloach algorithm with total update time $ O (m n) $. A point $ (\alpha, \beta) $ in this diagram means that for a graph with $ m = \Theta (n^\alpha) $ the algorithm has a running time of $ O (n^{\beta + o(1)}) $.}\label{fig:comparison_running_time}
\end{figure}

\begin{table}
\centering
\begin{tabular}{ |l|l|l|l| }
\hline
Problems & Running Time & $ o(mn) $ when & Fastest when \\ 
\hline
\ssr & $ \tilde O (m^{7/6} n^{2/3}) $ & $ m = o(n^2) $ & $ n \leq m \leq n^{7/5} $ \\
 & $ O (m^{3/4} n^{5/4 + o(1)}) $ & $ m = \omega(n) $ & $ n^{7/5} \leq m \leq n^{13/9} $ \\
 & $ O (m^{2/3} n^{4/3 + o(1)} + m^{3/7} n^{12/7 + o(1)}) $ & $ m = \omega(n^{5/4}) $ & $ n^{13/9} \leq m \leq n^2 $ \\
\hline
\sssssp & $ \tilde O (m^{5/4} n^{1/2}) $ & $ m = o(n^2) $ & $ m \leq n^{10/7} $ \\
 & $ O (m^{2/3} n^{4/3 + o(1)}) $ & $ m = \omega(n) $ & $ m \geq n^{10/7} $ \\
\hline
$ (1+\epsilon) $-approximate \sssp & $ O (m^{7/6} n^{2/3 + o(1)}) $ & $ m = o(n^2) $ & $ m \leq n^{7/5} $ \\
 & $ O (m^{3/4} n^{5/4 + o(1)}) $ & $ m = \omega(n) $ & $ m \geq n^{7/5} $ \\
\hline
$ (1+\epsilon) $-approximate \sssssp & $ O (m^{5/4} n^{1/2 + o(1)}) $ & $ m = o(n^2) $ & $ m \leq n^{10/7} $ \\
 & $ O (m^{2/3} n^{4/3 + o(1)}) $ & $ m = \omega(n) $ & $ m \geq n^{10/7} $ \\
\hline
\end{tabular}
\caption{Summary of our algorithms. The running time in the second column is the total update time.
For $ (1+\epsilon) $-approximate \sssp and $ (1+\epsilon) $-approximate \sssssp we assume that, for some constant $ c $, $ \epsilon \geq 1 / \log^c{n} $ and the largest edge weight is at most $ 2^{\log^c{n}} $. The third column shows roughly where these algorithms give an improvement over the previous $O(mn)$-time algorithms. The fourth column indicates the density of the graph for which the algorithm gives the fastest running time.} \label{table:results summary}
\end{table}

\paragraph*{Discussion.}
Our main result are dynamic algorithms with constant query time and sublinear update time for single source reachability and $ (1 + \epsilon) $-approximate single source shortest paths in directed graphs undergoing edge deletions.
There is some evidence that it is hard to generalize our results in the following ways.
\begin{itemize}
\item \emph{(All pairs vs.~single source)}
The naive algorithm for computing all-pairs reachability (also called transitive closure) in a directed graph takes time $ O (m n) $ even in the {\em static setting}.
To date no combinatorial algorithm (not relying on fast matrix multiplication) is known that gives a polynomial improvement over this running time.
Thus, unless there is a major breakthrough for static transitive closure, we cannot hope for a combinatorial algorithm for decremental all-pairs reachability with a total update time of $ o (m n) $ and small query time.
Further, it was recently shown \cite{HenzingerKNS15} that no algorithm (not even a non-combinatorial one) can solve decremental all-pairs reachability faster, based on an assumption called Online Matrix-Vector Multiplication ({\sf OMv}).

\item \emph{(Worst-case update time)} The total update time of our single-source reachability algorithms is $ o (m n) $ which gives an amortized update time of $ o (n) $ over a sequence of $ \Omega (m) $ deletions.
It recently has been shown by Abboud and Vassilevska Williams~\cite{AbboudW14} that a combinatorial algorithm with \emph{worst-case} update time and query time of $ o (n^2) $ per deletion implies a faster combinatorial algorithm for Boolean matrix multiplication and, as has been shown by Vassilevska Williams and Williams~\cite{WilliamsW10}, for other problems as well.
Furthermore, for the more general problem of maintaining the number of reachable nodes from a source under deletions (which our algorithms can do) a worst-case update and query time of both $ o (m) $ falsifies the strong exponential time hypothesis.
It might therefore not be possible to deamortize our algorithms.

\item \emph{(Approximate vs. exact)} Our \sssp algorithms only provide approximate solutions.
It has been observed by Roditty and Zwick~\cite{RodittyZ11} that any exact combinatorial \sssp algorithm handling edge deletions that has a total update time of $ o(m n) $ and small query time implies a faster combinatorial for Boolean matrix multiplication.
After the preliminary version~\cite{HenzingerKNSTOC14} of this work appeared, Henzinger et al.~\cite{HenzingerKNS15} showed that $ O(mn) $ is essentially the best possible total update time for maintaining exact distances under the assumption that there is no ``truly subcubic'' algorithm for a problem called online Boolean matrix-vector multiplication.
Both of these hardness results apply even for unweighted undirected graphs.
Thus, approximation might be necessary to break the $ O (m n) $ barrier.
\end{itemize}

\paragraph*{Organization.}
We first introduce the notation and basic concepts shared by all our algorithms in \Cref{sec:stoc2014:prelim}.
To provide some intuition for our approach we sketch two basic $s$-$t$ reachability algorithms in \Cref{sec: overview s-t reachability}.
We give the reductions for extending our results to \sssp (and \ssr) and \scc in \Cref{sec:sssp,sec:strongly connected component}, respectively.
In \Cref{sec:sparse} we give a hierarchical $s$-$t$ reachability algorithm, which in turn gives algorithms for \ssr and \scc.
In \Cref{sec:approx_shortest_paths} we extend this idea to $ (1 + \epsilon) $-approximate \stsp and \sssp.
Finally, in \Cref{sec:SSR_dense} we give an \ssr algorithm with improved total update time for dense graphs.
We conclude the paper by discussing open problems.

\section{Preliminaries}\label{sec:stoc2014:prelim}

\subsection{Notation}

We are given a directed graph $ G $ that might be weighted or unweighted.
The graph undergoes a sequence of \emph{updates}, which might be edge deletions or, for weighted graphs, edge weight increases.
This is called the \emph{decremental setting}.
We denote by $ \Delta (G) $ the number of updates in $ G $, i.e., the number of edge deletions and edge weight increases.
We say that a node $ v $ is \emph{reachable} from $ u $ (or $ u $ can reach $ v $) if there is a path from $ u $ to $ v $ in $ G $.
The \emph{distance} $ \dist_G (x, y) $ of a node $ x $ to a node $ y $ in $ G $ is the weight of the shortest path, i.e., the minimum-weight path, from $ x $ to $ y $ in $ G $.
If there is no path from $ x $ to $ y $ in~$ G $ we set $ \dist_G (x, y) = \infty $.

Let $ G = (V, E) $ be a weighted directed graph, where $ V $ is the set of nodes of $ G $ and $ E $ is the set of edges of $ G $.
We denote by $ n $ the number of nodes of $ G $ and by $ m $ the number of edges of $ G $ \emph{before the first update}.
We denote the \emph{weight} of an edge $ (u, v) $ in~$ G $ by $ w_G (u, v) $.
In the weighted case we consider positive integer edge weights and denote the maximum allowed edge weight by $ W $.
We assume that $ W \leq 2^{\log^c{n}} $ for some constant $ c $.
In unweighted graphs we think of every edge weight as equal to $ 1 $.
Note that for unweighted graphs we have $ \Delta (G) \leq m $ and for weighted graphs we have $ \Delta (G) \leq m W $.

For every path $ \pi = \langle v_0, v_2, \ldots, v_k \rangle $ we denote its weight in $ G $ by $ w (\pi, G) = \sum_{0 \leq i \leq k-1} w_G (v_i, v_{i+1}) $ and its size in terms of number of edges (also called \emph{hops}) by $ |\pi| = k $.
We say that the \emph{length} of a path is its weight, but, to avoid ambiguity, we reserve this notion for unweighted graphs where the length of a path is equal to its number of edges. 
For every integer $ h \geq 1 $ and all nodes $ x $ and $ y $ in a directed graph~$ G $, the \emph{$ h $-hops distance} $ \dist_G^h (x, y) $ is the minimum weight of all paths from $ x $ to $ y $ in~$ G $ consisting of at most $ h $ edges.
Note that $ \dist_G (u, v) = \dist_G^n (u, v) $.

For every graph $ G = (V, E) $ and every subset of nodes $ U \subseteq V $, we define $ E [U] = E \cap U^2 $ and denote by $ G[U] = (U, E[U]) $ the \emph{subgraph of $ G $ induced by $ U $}.
We set the weight of every edge in~$ G [U] $ equal to its weight in~$ G $.
For sets of nodes $ U \subseteq V $ and $ U' \subseteq V $ we define $ E [U, U'] = E \cap (U \times U') $, i.e., $ E [U, U'] $ is the set of edges $ (u, v) \in E $ such that $ u \in U $ and $ v \in U' $.

Our algorithms allow problem-specific update and query operations.
The algorithms we design will have constant query time and we will compare them by their \emph{total update time} over all deletions and, in case of weighted graphs, edge weight increases.
They are randomized and correct with high probability against an oblivious adversary~\cite{Ben-DavidBKTW94} who chooses its sequence of updates and queries before the algorithm starts.\footnote{In particular, the oblivious adversary does not know the random choices of our algorithm an cannot infer them by choosing its updates and queries adaptively.}
We say that an event happens {\em with high probability (whp)} if it happens with probability at least $ 1 - 1/n^c $, for some constant $ c $.
We assume that arithmetic operations on integers can be performed in constant time.
Note that the total update time of decremental algorithms for weighted graphs will be $ \Omega (\Delta (G)) $, the number of edge deletions and edge weight increases in $ G $.
This is unavoidable as the algorithm has to read every update in the graph.
In unweighted graphs the dependence on $ \Delta (G) $ does not have to be stated explicitly because there we have $ \Delta (G) \leq m $ and the total update time is $ \Omega (m) $ for reading the initial graph anyway.
We use $ \tilde O $-notation to hide logarithmic factors, i.e., we write $ \tilde O (t(m, n, W)) $ as an abbreviation for $ O (t(m, n, W) \log^c{n}) $ when $ c $ is a constant.
Similarly, we use the notation $ \hat{O} (t (m, n)) $ as an abbreviation for $ O (t (m, n) \cdot n^{o(1)}) $.
Whenever one of our algorithms returns an $\alpha$-approximate result, the approximation factor will be of the form $ \alpha = 1 + \epsilon $ where $ 0 < \epsilon \leq 1 $.
In this paper we assume that $ \epsilon \geq 1 / \log^c{n} $ for some constant~$ c $.

\subsection{Problem Descriptions}

Our goal is to design efficient dynamic algorithms for the following problems.

\begin{definition}[\ssssr]
A \emph{decremental single-source single-sink reachability (\ssssr) algorithm} for a directed graph $ G $ undergoing edge deletions, a source node $ s $, and a sink node $ t $, maintains the information whether $ s $ can reach $ t $.
It supports the following operations:
\begin{itemize}
\item \Delete{$u$, $v$}: Delete the edge $ (u, v) $ from $ G $.
\item \Query{}: Return `yes' if $ s $ can reach $ t $ and `no' otherwise.
\end{itemize}
\end{definition}

\begin{definition}[\ssr]
A \emph{decremental single-source reachability (\ssr) algorithm} for a directed graph $ G $ undergoing edge deletions and a source node $ s $, maintains the set of nodes reachable from $ s $ in $ G $.
It supports the following operations:
\begin{itemize}
\item \Delete{$u$, $v$}: Delete the edge $ (u, v) $ from $ G $.
\item \Query{$v$}: Return `yes' if $ s $ can reach $ v $ and `no' otherwise.
\end{itemize}
\end{definition}

\begin{definition}[\scc]\label{def:strongly connected component}
A \emph{decremental strongly connected components (\scc) algorithm} for a directed graph $ G $ undergoing edge deletions maintains a set of IDs of the strongly connected components of $ G $ and, for every node $ v $, the ID of the strongly connected component that contains $ v $.
It supports the following operations:
\begin{itemize}
\item \Delete{$u$, $v$}: Delete the edge $ (u, v) $ from $ G $.
\item \Query{$v$}: Return the ID of the strongly connected component that contains $ v $.
\end{itemize}
\end{definition}

\begin{definition}[$ \alpha $-approximate \sssssp]
A \emph{decremental $ \alpha $-approximate single-source single-sink shortest path (\sssssp) algorithm} for a weighted directed graph $ G $ undergoing edge deletions and edge weight increases, a source node $ s $, and a sink node $ t $, maintains a distance estimate $ \distest (s, t) $ such that $ \dist_G (s, t) \leq \distest (s, t) \leq \alpha \dist_G (s, t) $.
It supports the following operations:
\begin{itemize}
\item \Delete{$ u $, $ v $}: Delete the edge $ (u, v) $ from $ G $.
\item \Increase($ u $, $ v $, $ x $): Increase the weight of the edge $ (u, v) $ to $ x $.
\item \Query{}: Return the $ \alpha $-approximate distance estimate $ \distest (s, t) $.
\end{itemize}
\end{definition}

\begin{definition}[$ \alpha $-approximate \sssp]
A \emph{decremental $ \alpha $-approximate single-source shortest paths (\sssp) algorithm} for a weighted directed graph $ G $ undergoing edge deletions and edge weight increases and a source node $ s $, maintains, for every node $ v $, a distance estimate $ \distest (s, v) $ such that $ \dist_G (s, v) \leq \distest (s, v) \leq \alpha \dist_G (s, v) $.
It supports the following operations:
\begin{itemize}
\item \Delete{$ u $, $ v $}: Delete the edge $ (u, v) $ from $ G $.
\item \Increase($ u $, $ v $, $ x $): Increase the weight of the edge $ (u, v) $ to $ x $.
\item \Query{$ v $}: Return the $ \alpha $-approximate distance estimate $ \distest (s, v) $.
\end{itemize}
\end{definition}

\subsection{Definitions and Basic Properties}\label{sec:definitions and properties}

In the following we introduce the basic concepts shared by all our algorithms.
The first new concept used by our algorithms is the path union of a pair of nodes.

\begin{definition}[Path Union]\label{def:path_union}
For every directed graph $ G $, every $ D \geq 1 $, and all pairs of nodes $ x $ and $ y $ of $ G $, the path union $ \cP (x, y, D, G) \subseteq V $ is the set containing all nodes that lie on some path $ \pi $ from $ x $ to $ y $ in $ G $ of weight $ w_G (\pi) \leq D $.
\end{definition}

Note that if the shortest path from $ x $ to $ y $ in $ G $ has weight at most $ D $, then the subgraph of $ G $ induced by the path union $ \cP (x, y, D, G) $ contains this shortest path.
Thus, instead of finding this shortest path in $ G $ directly we can also find it in this potentially smaller subgraph of $ G $.
In our algorithms we will be able to bound the sizes of the path union subgraphs, which makes them very useful.
Observe that for a fixed value of $ D $, the path union can be computed in nearly linear time:

\begin{lemma}[Linear-Time Path Union]\label{pro:path_union_characterization}\label{lem:path_union_computation}
For every directed graph $ G $, every $ D \geq 1 $ and all pairs of nodes $ x $ and $ y $ of $ G $, we have $ \cP (x, y, D, G) = \{ v \in V \mid \dist_G (x, v) + \dist_G (v, y) \leq D \} $.
We can compute this set in time $ \tilde O (m) $ in weighted graphs and $ O (m) $ in unweighted graphs, respectively.
\end{lemma}

\begin{proof}
Clearly, if $ \dist_G (x, v) + \dist_G (v, y) \leq D $, then the concatenation of the shortest path from $ x $ to $ v $ and the shortest path from $ v $ to $ y $ is a path from $ x $ to $ y $ of weight at most $ D $ and thus $ v \in \cP (x, y, D, G) $.
Conversely, if $ v \in \cP (x, y, D, G) $, then there is a shortest path $ \pi $ from $ x $ to $ y $ containing $ v $ of weight at most $ D $.
Let $ \pi_1 $ and $ \pi_2 $ be the subpaths of $ \pi $ from $ x $ to $ v $ and from $ v $ to $ y $, respectively.
Then $ \dist_G (x, v) + \dist_G (v, y) \leq w_G (\pi_1) + w_G (\pi_2) = w_G (\pi) \leq D $.

Using Dijkstra's algorithm we compute $ \dist_G (x, v) $ and $ \dist_G (v, y) $ for every node $ {v \in V} $ in time $ \tilde O (m) $.
Afterwards, we iterate over all nodes and check for every node $ v $ whether $ \dist_G (x, v) + \dist_G (v, y) \leq D $ in total time $ O (n) $.
In unweighted graphs we can compute $ \dist_G (x, v) $ and $ \dist_G (v, y) $ for every node $ v \in V $ in time $ O (m) $ by performing breadth-first search (BFS).
\end{proof}

We further observe in the following that path unions are monotone in the decremental setting, i.e., no nodes will ever be added to any path union for a fixed value of $ D $ while the graph undergoes deletions and weight increases
This means that once we have computed the path union we can also use it for future versions of the graph, as long as we do not want to consider larger weights of the paths.

\begin{lemma}\label{pro:path_union_property_under_deletions}
Let $ G $ be a directed graph and let $ G' $ be the result of deleting some edge from $ G $ or increasing the weight of some edge from $ G $.
Then for every $ D \geq 1 $ and all pairs of nodes $ x $ and $ y $ of $ G $, $ \cP (x, y, h, G') \subseteq \cP (x, y, h, G) $.
\end{lemma}

\begin{proof}
Let $ v \in \cP (x, y, h, G') $, i.e., there is a path $ \pi $ from $ x $ to $ y $ in $ G' $ of weight at most $ D $.
All edges of $ \pi $ are also contained in $ G $ and thus there is a path from $ x $ to $ y $ in $ G' $ of weight $ w_G (\pi) = w_{G'} (\pi) \leq D $.
Thus, every node of $ \pi $ is contained in $ \cP (x, y, h, G) $ which implies that all edges of $ \pi $ are contained in $ G [\cP (x, y, h, G)] $.
Therefore $ w_{G [\cP (x, y, h, G)]} (\pi) = w_G (\pi) \leq D $ and thus every node on $ \pi $, and in particular $ v $, is contained in $ \cP (x, y, h, G) $.
\end{proof}

The last property of path unions we will use repeatedly is that we can update them by ``computing the path union of the path union''.
This rebuilding of path unions will be useful later to restrict the overlap of path unions of different pairs of nodes.

\begin{lemma}\label{lem:path_union_in_supergraph}
For every directed graph $ G $, every $ D \geq 1 $, every pair of nodes $ x $ and $ y $, and every set of nodes $ \cQ $ such that $ \cP (x, y, D, G) \subseteq \cQ \subseteq V $ we have $ \cP (x, y, D, G [\cQ]) = \cP (x, y, D, G) $.
\end{lemma}

\begin{proof}
Let $ v \in \cP (x, y, D, G [\cQ]) $, which means that $ v $ lies on a path $ \pi $ from $ x $ to $ y $ of weight at most $ D $.
As $ G [\cQ] $ is a subgraph of $ G $, this path is also contained in $ G $ (with the same weight) and thus $ v \in \cP (x, y, D, G) $.

Now let $ v \in \cP (x, y, D, G) $, which means that $ v $ lies on a path $ \pi $ from $ x $ to $ y $ of weight at most $ D $.
By the assumption $ \cP (x, y, D, G) \subseteq \cQ $ every node $ v' $ of $ \pi $ is contained in $ \cQ $ and thus $ \pi $ is contained in $ G [\cQ] $ (and has the same weight as in $ G $).
As $ v $ lies on a path from $ x $ to $ y $ of weight at most $ D $ in $ G $, we have $ v \in \cP (x, y, D, G [\cQ]) $.
\end{proof}

We now introduce two concepts that are used frequently, also in other dynamic algorithms.
The first such concept is the following algorithmic primitive for maintaining shortest paths trees up to bounded depth.

\begin{theorem}[Even-Shiloach tree~\cite{EvenS81,HenzingerK95,King99}]\label{lem:ES-tree}
There is a decremental algorithm, called \emph{Even-Shiloach tree} (short: ES-tree), that, given a weighted directed graph $ G $  undergoing edge deletions and edge weight increases with positive integer edge weights, a source node $ s $, and a parameter $ D \geq 1 $, maintains a shortest paths tree from $ s $ and the corresponding distances up to depth $ D $ with total update time $ O (m D + \Delta (G)) $, i.e., the algorithm maintains $ \dist_G (s, v) $ and the parent of $ v $ in the shortest paths tree for every node $ v $ such that $ \dist_G (s, v) \leq D $.
By reversing the edges of $ G $ it can also maintain the distance from $ v $ to $ s $ for every node $ v $ in the same time.
\end{theorem}

Note that if we want to use the ES-tree to maintain a full shortest paths tree (i.e., containing all shortest paths from the source), then we have to set $ D $ equal to the maximum (finite) distance from the source.
If the graph is unweighted, then it is sufficient to set $ D = n-1 $ and we can thus maintain a full shortest paths tree with total update time $ O (m n) $.
If the graph is weighted and $ W $ is the maximum edge weight, then the maximum distance might be as large as $ (n-1) W $ and the algorithm then takes time $ O (m n W) $, which is less efficient than recomputation from scratch even if $ W = n $.
Similarly, if we want to maintain a shortest paths tree up to $ h \leq n-1 $ hops (containing all shortest paths with at most $ h $ edges), the algorithm above takes time $ O (m h) $ in unweighted graphs and $ O (m h W) $ in weighted graphs.
Using a scaling technique~\cite{Bernstein09,Madry10,Bernstein13}, the running time in weighted graphs can be reduced if we allow approximation:
we can maintain a $ (1 + \epsilon) $-approximate shortest paths tree containing all shortest paths up to $ h $ hops in total time $ \tilde O (m n \log{W} / \epsilon) $.

The second concept used repeatedly in our algorithms is sampling nodes or edges at random.
All our randomized algorithms will be correct whp against an oblivious adversary who fixes its sequence of updates and queries before the algorithm is initialized, revealing its choices to the algorithm one after the other.
It is well-known, and exploited by many other algorithms for dynamic shortest paths and reachability, that by sampling a set of nodes with a sufficiently large probability we can guarantee that certain sets of nodes contain at least one of the sampled nodes.
To the best of our knowledge, the first use of this technique in graph algorithms goes back to Ullman and Yannakakis~\cite{UllmanY91}.

\begin{lemma}\label{lem:center_on_shortest_path}\label{lem:hitting_set_argument}
Let $ a \geq 1 $, let $ T $ be a set of size $ t $ and let $ S_1, S_2, \ldots, S_k $ be subsets of $ T $ of size at least $ q $.
Let~$ U $ be a subset of $ T $ that was obtained by choosing each element of $ T $ independently with probability $ p = \min (x / q, 1) $ where $ x = a \ln{(k t)} + 1 $.
Then, with high probability (whp), i.e., probability at least $ 1 - 1/t^a $, the following two properties hold:
\begin{enumerate}
\item For every $ 1 \leq i \leq k $, the set $ S_i $ contains a node of $ U $, i.e., $ S_i \cap U \neq \emptyset $.
\item $ |U| \leq 3 x t / q = O (a t \ln{(k t)} / q) $.
\end{enumerate}
\end{lemma}

\begin{proof}
For every $ 1 \leq i \leq k $ let $ E_i $ be the event that $ S_i \cap U = \emptyset $, i.e., that $ S_i $ contains no node of~$ U $.
Furthermore, let $ E $ be the event that there is a set $ S_i $ that contains no node of $ U $.
Note that $ E = \bigcup_{1 \leq i \leq k} E_i $.

We first bound the probability of the event $ E_{i} $ for $ 1 \leq i \leq k $.
The size of $ S_i \cap U $ is determined by a Bernoulli trial with success probability $ p $.
The probability that $ | S_i \cap U | = 0 $ is therefore given by
\begin{equation*}
\Pr (E_i) = \left( 1 - p \right)^{|S_i|} \leq \left( 1 - p \right)^{q} = \left( 1 - \frac{a \ln{(k t)} + 1}{q} \right)^q \leq \frac{1}{e^{a \ln{(k t)} + 1}} = \frac{1}{e k^a t^a} \leq \frac{1}{2 k^a t^a} \, .
\end{equation*}
Here we use the well-known inequality $ (1 - 1/y)^y \leq 1 / e $ that holds for every $ y > 0 $.
Now we simply apply the union bound twice and get
\begin{equation*}
\Pr (E) = \Pr \left( \bigcup_{1 \leq i \leq k} E_i \right) \leq \sum_{1 \leq i \leq k} \Pr (E_i) \leq \sum_{1 \leq i \leq k} \frac{1}{2 k^a t^a} = k \frac{1}{2 k^a t^a} \leq \frac{1}{2 t^a} \, .
\end{equation*}

To bound the size of $ U $, let, for every $ 1 \leq i \leq t $, $ X_i $ be the random variable that is $ 1 $ if the $i$-th element is sampled and $ 0 $ otherwise, i.e., $ X $ is $ 1 $ with probability $ p $ and $ 0 $ with probability $ 1 - p $.
The total size of $ U $ is given by the random variable $ X = \sum_{1 \leq i \leq t} X_i $.
Let $ \mu $ be the expectation of $ X $.
By a standard Chernoff bound we have
\begin{equation*}
\Pr [X > (1 + \delta) \mu] \leq \frac{1}{e^{\frac{\delta^2}{2 + \delta} \mu}} \, .
\end{equation*}
With $ \delta = 2 $ we get that the probability that $ U $ has size more than $ 3 \mu = 3 x t / q $ is at most
\begin{equation*}
\frac{1}{e^{\mu}} = \frac{1}{e^{x t / q}} \leq \frac{1}{e^{x}} = \frac{1}{e^{a \ln{kt} + 1}} \leq \frac{1}{2 k^a t^a} \leq \frac{1}{2 t^a} \, .
\end{equation*}

By a union bound, the probability that one of property~1 ($ U $ is a hitting set) or property~2 ($ U $ is small enough) fails is at most $ 1 / t^a $ and thus the claim follows.
\end{proof}

We now sketch one example of how we intend to use \Cref{lem:hitting_set_argument} for dynamic graphs.
Consider an unweighted graph $ G $ undergoing edge deletions.
Suppose that we want the following condition to hold for every pair of nodes $ x $ and $ y $  with probability at least $ 1 - 1/n $ in \emph{all} versions of $ G $ (i.e., initially and after each deletion): if $ \dist_G (x, y) \geq q $, then there is a shortest path from $ x $ to $ y $ that contains a node in $ U $.
We apply \Cref{lem:hitting_set_argument} as follows.
The set $ T $ is the set of nodes of $ G $ and has size $ n $.
The sets $ S_1, S_2, \ldots, S_k $ are obtained as follows: for every version of $ G $ (i.e., after each deletion) and every pair of nodes $ x $ and $ y $ such that $ \dist_G (x, y) \geq q $, we define a set $ S_i $ that contains the nodes on the first shortest path from $ x $ to $ y $ (for an arbitrary, but fixed order on the paths).
As the graph undergoes edge deletions, there are at most $ m \leq n^2 $ versions of the graph.
Furthermore, there are $ n^2 $ pairs of nodes.
Therefore we have $ k \leq n^4 $ such sets.
Thus, for the property above to hold with probability $ 1 - 1/n $, we simply have to sample each node with probability $ (\ln{k t}) / q = (\ln{n^5}) / q = (5 \ln{n}) / q $ by \Cref{lem:hitting_set_argument}.

Finally, we explain how to decrease the number of edge weight increases to be considered.
Remember that in general the number of edge deletions and edge weight increases can only be bounded by $ \Delta (G) = O (m W) $.
However, by the reduction below, we may assume without loss of generality that $ \Delta (G') = O (m \log_{1+\epsilon} W) $ as we are only interested in $ (1 + \epsilon) $-approximate shortest paths.
This for example helps us in reducing the sampling probability for applying \Cref{lem:hitting_set_argument} as the number of versions of $ G $ is reduced significantly.

\begin{lemma}\label{lem:few_updates}
Given a weighted directed graph $ G $ undergoing edge deletions and edge weight increases, we can, in constant time per update in $ G $, maintain a weighted directed graph $ G' $ undergoing edge deletions and edge weight increases such that
\begin{itemize}
\item the number of updates in $ G' $ is at most $ \lceil \log_{1+\epsilon} W \rceil $ per edge (which implies $ \Delta (G') = O (m \log_{1+\epsilon} W) $) and
\item $ \dist_G (x, y) \leq \dist_{G'} (x, y) \leq (1+\epsilon) \dist_G (x, y) $ for all nodes $ x $ and $ y $.
\end{itemize}
\end{lemma}

\begin{proof}
We define $ G' $ to be the graph where each edge weight of $ G $ is rounded up to the nearest power of $ (1 + \epsilon) $.
For every edge $ (u, v) $, we have $ (1 + \epsilon)^i \leq w (u, v) < (1 + \epsilon)^{i+1} $ for some $ i $.
By setting $ w'(u, v) = \lfloor (1 + \epsilon)^{i+1} \rfloor $ we get $ w (u, v) \leq w'(u, v) \leq (1 + \epsilon) w(u, v) $.
This means that the weight of every path $ \pi $ in $ G $ is $ (1 + \epsilon) $-approximated by the weight of~$ \pi $ in~$ G' $.
As this is also true for a shortest path, we get $ \dist_{G} (x, y) \leq \dist_{G'} (x, y) \leq (1 + \epsilon) \dist_G (x, y) $ for all nodes $ x $ and~$ y $.
With every edge weight increase of an edge $ (u, v) $ in $ G $ we can in constant time compute the weight $ w'(u, v) $ and check whether it has increased.
The number of different edge weights in $ G $ is $ \lceil \log_{1+\epsilon} W \rceil $ and therefore the total number of edge deletions and edge weight increases in $ G' $ is $ O (m \log_{1+\epsilon} W) $.
Thus, $ G' $ is the desired graph.
\end{proof}

\section{Algorithm Overview for $s$-$t$ Reachability}\label{sec: overview s-t reachability}

In this section, we illustrate our main ideas by giving simple algorithms for the $s$-$t$ reachability problem (\ssssr).
At the heart of all our algorithms is a new way of maintaining reachability or distances between some {\em nearby} pairs of nodes using small path unions (defined in \Cref{sec:definitions and properties}). 

\subsection{Decremental Bounded-Hop Multi-Pair Reachability} \label{sec:bounded-hop multi-pair}

We first give an algorithm for solving the following ``restricted'' reachability problem:
We are given $k$ pairs of sources and sinks $(s_1, t_1)$, $\ldots$, $(s_k, t_k)$ and a parameter $h$. We want to maintain, for each $1 \leq i \leq k$, whether $\dist_G(s_i, t_i)\leq h$.

In this algorithm we use the following set $ B \subseteq V $ of \emph{hubs} of size $ \tilde O (b) $ (for some parameter~$ b $).
Let $ U \subseteq V $ be a set of nodes obtained by sampling each node independently with probability $ a b \ln{n} / n $ and let $ F \subseteq E $ be a set of edges obtained by sampling each edge independently with probability $ a b \ln{n} / m $ from the initial graph (for a large enough constant~$ a $).
We let $ B $ be the set containing $ U $ and the endpoints of every edge in $ F $.
For every pair of nodes $ u $ and $ v $ we define the \emph{hub-distance from $ u $ to $ v $} as $\dist_{\jset}(u, v) =\min_{x\in \jset} \dist_G(u, x)+\dist_G(x, v)$.

The hubs are useful in combination with the path unions in the following way.
Intuitively, instead of directly maintaining whether $\dist_G(s_i, t_i)\leq h$, we can maintain whether $\dist_{G [\cP (s_i, t_i, h, G)]} (s_i, t_i)\leq h$ since $G [\cP(s_i, t_i, h, G)]$ (the subgraph of $ G $ induced by $ \cP(s_i, t_i, h, G) $) contains all $s_i$-$t_i$ paths of length at most $h$ in $G$. This could be helpful when $\cP(s_i, t_i, h, G)$ is much smaller than~$V$. Our first key idea is the observation that all $(s_i, t_i)$ pairs for which $\cP(s_i, t_i, h, G)$ is large can use their paths through a small number of hubs to check whether $\dist_{G [\cP(s_i, t_i, h, G)]} (s_i, t_i)\leq h$ (thus the name ``hub'').
In particular, consider the following algorithm. 
We maintain the distance from each $s_i$ to each $t_i$ in two ways. 
\begin{enumerate}
\item Maintain $\dist_\jset(s_i, t_i)$, for each $1\leq i\leq k$, as long as $\dist_\jset(s_i, t_i)\leq h$. 
\item Once $\dist_\jset(s_i, t_i)>h$, construct $\cQ (s_i, t_i) = \cP(s_i, t_i, h, G)$ and maintain $\pudist(s_i, t_i)$ up to value $h$.
\end{enumerate}
Our algorithm outputs $\dist_G(s_i, t_i)\leq h$ if and only if either $\dist_\jset(s_i, t_i)\leq h$ or $\pudist(s_i, t_i)\leq h$. The correctness is obvious: either $\dist_\jset(s_i, t_i)\leq h$ which already implies that $\dist_G(s_i, t_i)\leq h$ or otherwise we maintain  $\pudist(s_i, t_i)$ which captures all $h$-hop $s_i$-$t_i$ paths. 
The more important point is the efficiency of maintaining both distances. Our analysis mainly uses the following lemma.

\begin{lemma}[Either hub-distance or path-union graph is small]\label{thm:using hubs} With high probability, for each $i$, either $\dist_\jset(s_i, t_i)\leq h$, or $G [\cP(s_i, t_i, h, G)]$ has at most $\min (m/\jsize, n^2/\jsize^2) $ edges.  
\end{lemma}
\begin{proof}[Proof Sketch]
By the random sampling of the hubs, if $G [\cP(s_i, t_i, h, G)]$ has more than $n/\jsize$ nodes, then one of these nodes, say $x$, will have been sampled by the algorithm and thus contained in $\jset$ whp (\Cref{lem:hitting_set_argument}). By definition, $x$ lies on some $s_i$-$t_i$ path of length at most~$h$. Thus, $\dist_\jset(s_i, t_i)\leq \dist_G(s_i, x)+\dist_G(x, t_i) \leq h$.
Similarly, if $G [\cP(s_i, t_i, h, G)]$ has more than $m/\jsize$ edges, then one of these edges, say $(x, y)$, will have been sampled by the algorithm whp. This means that $x$ will be contained in $\jset$, and thus $\dist_\jset(s_i, t_i)\leq \dist_G(s_i, x)+\dist_G(x, t_i) \leq h$.
As the number of edges of $G [\cP(s_i, t_i, h, G)]$ can be at most $ |\cP(s_i, t_i, h, G)|^2 $, the stated bound follows.
\end{proof}

\Cref{thm:using hubs} guarantees that when we construct $\cQ(s_i, t_i)$, its induced subgraph is much smaller than $G$ whp, and so it is beneficial to maintain the distance in $G [\cQ(s_i, t_i)]$ instead of $G$.
For each $ 1 \leq i \leq k $, we can maintain $\pudist(s_i, t_i)$ by running an ES-tree rooted at $ s_i $ up to distance $ h $ in $ G [\cQ(s_i, t_i)] $, which takes time $\tilde O(| E [\cQ(s_i, t_i)] | \cdot h) $ (where $E [\cQ(s_i, t_i)] $ denotes the set of edges of $G [\cQ(s_i, t_i)]$).
As $ | E [\cQ(s_i, t_i)] | \leq \min (m/\jsize, n^2/\jsize^2) $, this takes time $ O ( k h \min (m/\jsize, n^2/\jsize^2)) $.
By \Cref{lem:path_union_computation} we can construct each $\cQ(s_i,t _i)$ in $ O(m)$ time for each $ 1 \leq i \leq k $, resulting in a total cost of $ O (k m) $ for computing the path unions.
Finally, the time needed for maintaining the hub distances can be analyzed as follows.

\begin{lemma}[Maintaining all $\dist_\jset(s_i, t_i)$]\label{thm:time for hub distance}
We can maintain whether $\dist_\jset(s_i, t_i)\leq h$, for all $ 1 \leq i \leq k $, in total time $\tilde O(\jsize m h+k\jsize h)$.
\end{lemma}
\begin{proof}[Proof Sketch] For each hub $v\in \jset$, we maintain the values of $\dist_G(s_i, v)$ and $\dist_G(v, t_i)$ up to~$h$ using ES-trees up to distance $h$ rooted at each hub. This takes $\tilde O(\jsize m h)$ total update time. 
Every time $\dist_G(s_i, v)$ or $\dist_G(v, t_i)$ changes, for some hub $v\in \jset$ and some $ 1 \leq i \leq k $, we update the value of $ \dist_G (s_i, t_i) = \min_{v \in B} (\dist_G(s_i, v)+\dist_G(v, t_i))$, incurring $ \tilde O(b)$ time. Since the value of each of $\dist_G(s_i, v)$ and $\dist_G(v, t_i)$ can change at most $h$ times, we need $\tilde O(k\jsize h)$ time in total. 
\end{proof}

It follows that we can maintain, for all $1\leq i\leq k$, whether $\dist_G(s_i, t_i)\leq h$ with a total update time of
\begin{equation}
\tilde O \Big(
\underbrace{\jsize m h + k \jsize h}_{\substack{\text{\scriptsize maintain $\dist_\jset (s_i, t_i)$}\\\text{\scriptsize (\Cref{thm:time for hub distance})}}}
+
\underbrace{k m}_{\text{\scriptsize construct $ \cQ (s_i, t_i) $}}
+
~\underbrace{k h \cdot \min ( m / \jsize, n^2 / \jsize^2 )}_{\text{maintain $ \pudist (s_i, t_i) $}}
 \Big) \, . \label{eq:bounded_hop_multi_pair_running_time}
\end{equation}
Note that, previously, the fastest way of maintaining, for all $ 1 \leq i \leq k $, whether $\dist_G(s_i, t_i)\leq h$ was to maintain an ES-tree separately for each pair, which takes $ O(mhk)$ time.

\subsection{Decremental $s$-$t$ Reachability in Dense Graphs}\label{sec: overview: dense graph}

In the following we are given a source node $ s $ and a sink node $ t $ and want to maintain whether $ s $ can reach $ t $.
In our algorithm we use a set of \emph{centers} $ C \subseteq V $ of size $ \tilde O (c) $ (for some parameter $ c $) obtained by sampling each node independently with probability $ a c \ln{n} / n $ (for some large enough constant $ a $).
Using these centers, we define the {\em center graph}, denoted by $\cg$, as follows. The nodes of $\cg$ are the centers in $C$ and for every pair of centers $u$ and $v$ in $C$, there is a directed edge $(u, v)$ in $\cg$ if and only if $\dist_G (u, v)\leq n/c$.

\begin{lemma}[$\cg$ preserves $s$-$t$ reachability]\label{thm:property of center graph} Whp, $s$ can reach $t$ in $G$ if and only if $s$ can reach $t$ in $\cg$. 
\end{lemma}
\begin{proof}[Proof Sketch] Since we can convert any path in $\cg$ to a path in $G$, the ``if'' part is clear. To prove the ``only if'' part, let $\pi$ be an $s$-$t$ path in $G$.
By the random sampling of centers, there is a set of centers $c_1, c_2, \ldots, c_k $ on $ \pi $ such that $c_1=s$, $c_k=t$, and $\dist_G(c_i, c_{i+1})\leq n/c$ for all $i$ whp (\Cref{lem:hitting_set_argument}). The last property implies that the edge $(c_i, c_{i+1})$ is contained in $\cg$ for all $i$.  Thus $s$ can reach $t$ in $\cg$. 
\end{proof}

\Cref{thm:property of center graph} implies that to maintain $s$-$t$ reachability in $\cg$ it is sufficient to maintain the edges of $\cg$ and to maintain $s$-$t$ reachability in $\cg$. To maintain the edges of $\cg$, we simply have to maintain the distances between all pairs of centers up to $n/c$. This can be done using the bounded-hop multi-pair reachability algorithm from before with $k=c^2$ and $h=n/c$, where we make each pair of centers a source-sink pair.
By plugging in these values in \Cref{eq:bounded_hop_multi_pair_running_time}, we obtain a total update time of $ \tilde O (b m n / c + b c n + c^2 m + c n^3 / b^2) $ for this bounded-hop multi-pair reachability.
To maintain $s$-$t$ reachability in $\cg$, note that $\cg$ is a dynamic graph undergoing only deletions.
Thus, we can simply maintain an ES-tree rooted at $ s $ in $\cg$.
As $ \cg $ has $ \tilde O (c) $ nodes and $ \tilde O (c^2) $ edges this takes time $ \tilde O(c^3) $.
It follows that the total update time of this $s$-$t$ reachability algorithm is
\begin{equation}
\tilde O \Big(~
\underbrace{\jsize m n/c+\jsize n c}_{\substack{\text{\scriptsize maintain $\dist_\jset (\cdot, \cdot)$}\\\text{\scriptsize (\Cref{thm:time for hub distance})}}}
~+\underbrace{c^2 m}_{\text{\scriptsize construct $ \cQ (\cdot, \cdot) $}}
+\underbrace{c n^3/\jsize^2}_{\text{maintain $ \pudist (\cdot, \cdot) $}}
+\underbrace{c^3}_{\text{\scriptsize maintain $ \dist_{\cg} (s, t) $}} 
~\Big) \, . \label{eq:time for dense algo}
\end{equation}
By setting\footnote{Note that we have to make sure that $1\leq \jsize\leq n$ and $1\leq c\leq n$. It is easy to check that this is the case using the fact that $m\leq n^2$.}  
$\jsize = n^{8/7}/m^{3/7}$ and $c=(\jsize n)^{1/3}=n^{5/7}/m^{1/7}$, we get a total update time of\footnote{
%
%
\textit{Detailed calculation:} First note that $c\leq n^{5/7} \leq m$.  So, the term $c^3$ is dominated by the term $mc^2$. 
Observe that $\jsize= (n^3/mc)^{1/2}$, thus $n^3c/\jsize^2=mc^2$. So, the fourth term is the same as the third term ($mc^2$). 
Using $c=(\jsize n)^{1/3}$, we have that the first and third terms are the same. 
Now, the third term is $mc^2=m(n^{5/7}/m^{1/7})^2=m^{1-2/7}n^{10/7}=m^{5/7}n^{10/7}$. For the second term, using $\jsize n = c^3$, we have $\jsize n c = c^4 = n^{20/7}/m^{4/7}$.
}
%
%
$ \tilde O \left( m^{5/7} n^{10/7}+n^{20/7}/m^{4/7} \right) $.
This is $ o (m n) $ if $m=\omega(n^{3/2})$.

\subsection{Decremental $s$-$t$ reachability in Sparse Graphs}\label{sec: overview: sparse graph}

Maintaining $s$-$t$ reachability in sparse graphs, especially when $m=\Theta(n)$, needs a slightly different approach. 
Carefully examining the running time of the previous algorithm in \Cref{eq:time for dense algo} reveals that we {\em cannot} maintain $\dist_{G [\cQ (u, v)]} (u, v)$ for {\em all} pairs of centers $ u $ and $ v $ at all times: this would cost $\tilde O (\min (c m n / \jsize, c n^3 / \jsize^2) ) $---the third term of \Cref{eq:bounded_hop_multi_pair_running_time}---but we always need $c=\omega(\jsize)$ to keep the first term of \Cref{eq:bounded_hop_multi_pair_running_time} to $\jsize m n/c=o(mn) $ and, since $ b \leq n $, $ c m n / \jsize = \omega (m n) $ and $ c n^3 / \jsize^2 = \omega (n^2) = \omega (m n) $.
The new strategy is to maintain $\dist_{G [\cQ (u, v)]} (u, v)$ only for {\em some} pairs of centers at each time step.

\paragraph*{Algorithm.}
As before, we sample $ \tilde O (\jsize)$ hubs $ \tilde O (c)$ centers, and maintain $\jdist(\cdot, \cdot)$ between all centers which takes $\tilde O(m\jsize n/c+\jsize n c)$ time (\Cref{thm:time for hub distance}). 
The algorithm runs in phases. At the beginning of each phase $i$, the algorithm does the following. Compute a BFS-tree $T$ on the outgoing edges of every node rooted at the source~$s$ in~$G$. If $T$ does not contain $t$, then we know that $s$ cannot reach $t$ anymore and there is nothing to do. Otherwise, let $L=(c_1, c_2, \ldots, c_k)$ be the list of centers on the shortest path from $s$ to $t$ in $T$ ordered increasingly by their distances to $s$. For simplicity, we let $c_0=s$ and $c_{k+1}=t$. Note that we can assume that 
\begin{equation}
\forall i,\ \dist_G(c_{i}, c_{i+1})\leq n/c \text{~~and~} \dist_G(c_{i},c_{i+2})>n/c \, .\label{eq:sparse algo - disjointness guarantee}
\end{equation}
The first inequality holds because the centers are obtained from random sampling (\Cref{lem:hitting_set_argument}) and the second one holds because, otherwise, we can remove $c_{i+1}$ from the list $L$ without breaking the first inequality. 
Observe that $L$ would induce an $s$-$t$ path in the ``center graph''. Our intention in this phase is to maintain whether $s$ can still reach $t$ using this path; in other words, whether $\dist_G c_{i}, c_{i+1})\leq n/c$ for all~$i$. (We start a new phase if this is not the case.)
We do this using the framework of \Cref{sec:bounded-hop multi-pair}: For each pair $(c_{i}, c_{i+1})$, we know that $\dist_G (c_{i}, c_{i+1})\leq n/c$ when $\dist_\jset(c_{i}, c_{i+1})\leq n/c$. Once $\dist_\jset(c_{i}, c_{i+1})> n/c$, we construct $\cQ(c_{i}, c_{i+1}) = \cP(c_{i}, c_{i+1}, n/c, G)$.
After each deletion of an edge $(u, v)$ in this phase, we find every index $ i $ such that $ u $ and $ v $ are contained in $\cQ(c_{i}, c_{i+1})$ and, using the static BFS algorithm, check whether $\dist_{G [\cQ(c_{i}, c_{i+1})]} (c_i, c_{i+1})\leq n/c$.
We start a new phase when $\dist_G (c_i, c_{i+1})$ for some pair of centers $(c_i, c_{i+1})$ changes from $\dist_G (c_i, c_{i+1})\leq n/c$ to $\dist_G (c_i, c_{i+1})>n/c$.

\paragraph*{Running Time Analysis.}
First, let us bound the number of phases.
As for each pair $ (c_i, c_{i+1}) $ the distance $\dist_G (c_i, c_{i+1})$ can only become larger than $ n/c $ at most once, there are at most $ \tilde O (c^2) $ phases. 
At the beginning of each phase, we have to construct a BFS-tree in $G$, taking $O(m)$ time and contributing $ \tilde O(c^2 m)$ to the total time over all phases. During the phase, we have to construct $\cQ(c_i, c_{i+1})$ for at most $c$ pairs of center, taking $\tilde O(c m)$ time (by \Cref{lem:path_union_computation}) and contributing $\tilde O(c^3 m)$ total time.
Moreover, after deleting an edge $(u, v)$, we have to update $\dist_{G [\cQ (c_i, c_{i+1})]} (c_i, c_{i+1})$ for every $\cQ(c_i, c_{i+1})$ containing both $ u $ and $ v $, by running a BFS algorithm. This takes $O(| E [\cQ(c_i, c_{i+1}] |)$ time for each $\cQ(c_i, c_{i+1})$, which is $\tilde O(m/\jsize)$ whp, by \Cref{thm:using hubs}.
The following lemma implies that every node will be contained in only a constant number of such sets $\cQ(c_i, c_{i+1})$; 
so, we need $\tilde O(m^2/\jsize)$ time to update $\dist_{G [\cQ (c_i, c_{i+1})]}(c_i, c_{i+1})$ over all $m$ deletions.  

\begin{lemma}\label{thm: disjoint path unions}
For any $i$ and $j\geq i+3$, $ \cQ(c_i, c_{i+1}) $ and $ \cQ(c_{j}, c_{j+1}) $ are disjoint. 
\end{lemma}
\begin{proof}
First, we claim that $\dist_G(c_{i}, c_{j+1}) > 2n/c$. To see this, let $G'$ be the version of the graph at the beginning of the current phase. We know that 
$\dist_{G'} (c_{i}, c_{i+4}) = \dist_{G'} (c_{i}, c_{i+2}) +\dist_{G'} (c_{i+2}, c_{i+4}) > 2n/c,$
where the equality holds because at the beginning of the current phase (i.e., in $G'$), every $c_i$ lies on the shortest $s$-$t$ path, and the inequality holds because of \Cref{eq:sparse algo - disjointness guarantee}. 
Since $j+1\geq i+4$ and both $c_{i+4}$ and $c_{j+1}$ lie on the shortest $s$-$t$ path in $G'$,  $\dist_{G'} (c_{i}, c_{j+1})\geq  \dist_{G'} (c_{i}, c_{i+4}) > 2n/c$. The claim follows since the distance between two nodes never decreases after edge deletions.

Now, suppose for the sake of contradiction that there is some node $ v $ that is contained in both $ \cQ(c_i, c_{i+1}) = \cP(c_i, c_{i+1}, n/c, G) $ and $ \cQ(c_{j}, c_{j+1}) = \cP(c_{j}, c_{j+1}, n/c, G) $. This means that $ v $ lies in some $c_i$-$c_{i+1}$ path and some $c_{j}$-$c_{j+1}$ path, each of length at most $n/c$. This means that $\dist_G(c_i, v)\leq n/c$ and $\dist_G(v, c_{j+1})\leq n/c$ and therefore we get $\dist_G(c_i, c_{j+1}) \leq \dist_G(c_i, v)+\dist_G(v, c_{j+1})\leq 2n/c$. This contradicts the lower bound above. 
\end{proof}

\noindent Thus, the total update time is
\begin{equation*}
\tilde O \Big(~
\underbrace{m\jsize n/c + \jsize n c}_{\substack{\text{\scriptsize maintain $\dist_\jset (\cdot, \cdot)$}\\\text{\scriptsize (\Cref{thm:time for hub distance})}}}
~+\underbrace{mc^2}_{\substack{\text{\scriptsize construct BFS}\\\text{\scriptsize tree in every phase}}}
+\underbrace{mc^3}_{\substack{\text{\scriptsize construct $\cQ(c_i, c_{i+1})$}\\\text{\scriptsize in every phase}}}
+\underbrace{m^2/\jsize}_{\substack{\text{\scriptsize update $\dist_{G [\cQ (c_i, c_{i+1})]} (c_i, c_{i+1})$}}}
~\Big) \, .
\end{equation*}
By setting $c=(mn)^{1/7}$ and  $\jsize = c^4/n = m^{4/7}/n^{3/7}$, we get a running time of $ \tilde O \left( m^{10/7}n^{3/7} \right) $.\footnote{
\textit{Detailed calculation:} First note that the third term ($mc^2$) is dominated by the fourth term ($mc^3$). Using $\jsize=c^4n$, the first term is the same as the fourth term ($mc^3$). Now, the fourth term is $mc^3=m(mn)^{3/7}=m^{10/7}n^{3/7}$. 
For the second term, using $\jsize n = c^4$, we have $\jsize n c = c^5 = (mn)^{5/7}$ which is at most $m^{10/7}n^{3/7}$ (using $n\leq m$). 
For the last term, $m^2/\jsize =  m^2/(m^{4/7}/n^{3/7}) = m^{(14-4)/7}n^{3/7}$.
} 
This is $ o (mn)$ if $m=o(n^{4/3})$. 

\section{General Reductions}

\subsection{From \sssp to \stsp}\label{sec:sssp}

In the following we show a reduction from decremental approximate \sssp to decremental approximate \stsp.
The naive way of doing this would be to use $ n $ instances of the approximate \stsp path algorithm, one for every node.
We can use much fewer instances by randomly sampling the nodes for which we maintain approximate\stsp and by using Bernstein's shortcut edges technique~\cite{Bernstein13}.

\begin{theorem}\label{thm: s-t to single source}
Assume we already have the following decremental algorithm that, given a weighted directed graph $ G $ undergoing edge deletions and edge weight increases, a source node $ s $, and a set of sinks $ T $ of size $ k $, maintains, for every sink $ t \in T $, a distance estimate $ \delta (s, t) $ such that $ \dist_G (s, t) \leq \delta (s, t) \leq \alpha \dist_G (s, t)) $ for some $ \alpha \geq 1 $ with constant query time and a total update time of $ T (k, m, n) $.
Then, for any $ k \leq n $, there exists a decremental $ (1 + \epsilon) \alpha $-approximate \sssp algorithm with constant query time and total update time $ O (T( O(k \log{(n \Delta)}), m, n) + m n \Delta / k $, where $ \Delta = \lceil \log_{1 + \epsilon} (nW) \rceil $, that is correct with high probability against an oblivious adversary.
\end{theorem}

\begin{proof}
At the initialization we randomly sample each node of $ G $ with probability $ a k \ln{(n \Delta)} / n $ (for a large enough constant $ a $).
We call the sampled nodes sinks.
We use the decremental algorithm for a set of sinks from the assumption to maintain a distance estimate $ \distest (s, t) $ for every sink $ t $ such that $ \dist_G (s, t) \leq \distest (s, t) \leq \alpha \dist_G (s, t) $.
Additionally, we maintain a graph $ G' $ which consists of the same nodes as $ G $ and the following edges:
(1) For every edge $ (u, v) $ of $ G $ we add an edge $ (u, v) $ of weight $ w_{G'} (u, v) = (1 + \epsilon)^{\lceil \log_{1+\epsilon} w_G (u ,v) \rceil} $.
(2) For every sink $ t $ we add an edge $ (s, t) $ of weight $ w_{G'} (s, t) = (1 + \epsilon)^{\lceil \log_{1+\epsilon} \distest (s, t) \rceil} $ (these edges are called \emph{shortcut edges}).
We maintain $ G' $ by updating the weights of these edges every time there is an edge weight increase or deletions in $ G $ or the distance estimate $ \distest (s, t) $ of some sink $ t $ changes its value.
On $ G' $ we use Bernstein's decremental $ (1 + \epsilon) $-approximate \sssp algorithm~\cite{Bernstein13} with source $ s $ and hop count $ h = n / k $.
This algorithm maintains a distance estimate $ \distest' (s, v) $ for every node $ v $ such that $ \dist_{G'} (s, v) \leq \distest' (s, v) \leq (1 + \epsilon) \dist_{G'}^h (s, v) $.

First, observe that $ \dist_G (s, v) \leq \dist_{G'} (s, v) \leq (1 + \epsilon) \dist_G (s, v) $ as the new shortcut edges never under-estimate the true distances in $ G $.
We now claim that $ \dist_{G'}^h (s, v) \leq (1 + \epsilon) \alpha \dist_G (s, v) $.
If $ \dist_G (s, v) = \infty $, then the claim is trivially true.
Otherwise let $ \pi $ be the shortest path from $ s $ to $ v $ in $ G $.
All edges of this path are also contained in $ G' $.
Thus, if $ \pi $ has at most $ h $ edges, then $ \dist_{G'}^h (s, v) = \dist_G (s, v) $.
If $ \pi $ has more than $ h $ edges, then we know that the set of nodes consisting of the last $ h $ nodes of $ \pi $ contains a sink $ t $ whp by \Cref{lem:hitting_set_argument} as the sinks are obtained by sampling from the nodes with probability $ a k \ln{(n \Delta)} / n = a \ln{(n \Delta)} / h $.
Thus, the graph $ G' $ contains a shortcut edge $ (s, t) $ of weight $ w_{G'} (s, t) \leq (1 + \epsilon) \distest (s, t) \leq (1 + \epsilon) \alpha \dist_G (s, t) $.
Now let $ \pi' $ be the path from $ s $ to $ v $ that starts with this edge $ (s, t) $ and then follows the path $ \pi $ from $ t $ to $ v $.
Clearly, the path $ \pi' $ has at most $ h $ edges and is contained in $ G' $.
As the weight of $ \pi' $ is at most $ w_{G'} (s, t) + (1 + \epsilon) \dist_G (t, v) $ we get
\begin{align*}
\dist_{G'}^h (s, v) \leq w_{G'} (s, v) + (1 + \epsilon) \dist_G (t, v) \leq (1 + \epsilon) \alpha \dist_G (s, t) + \dist_G (t, v) \leq (1 + \epsilon) \alpha \dist_G (s, v) \, .
\end{align*}
Putting everything together, we get that the distance estimate $ \distest' (s, v) $ fulfills
\begin{equation*}
\dist_G (s, v) \leq \dist_{G'} (s, v) \leq \distest' (s, v) \leq (1 + \epsilon) \dist_{G'}^h (s, v) \leq (1 + \epsilon)^2 \alpha \dist_G (s, v) \leq (1 + 3 \epsilon) \alpha \dist_{G'}^h (s, v) \, .
\end{equation*}
By running the whole algorithm with $ \epsilon' = \epsilon / 3 $ we obtain a $ (1 + \epsilon) \alpha $-approximation instead of a $ (1 + 3 \epsilon) \alpha $-approximation.

Finally, we argue about the running time.
Our running time has two parts.
(1) Bernstein's decremental $ (1 + \epsilon) $-approximate \sssp algorithm~\cite{Bernstein13} has constant query time and a total update time of $ O (m h \Delta) $ (here we also use the fact that $ w_{G'} (u, t) $ increases at most $ \Delta $ times for every sink $ t $).
Thus, our decremental algorithm also has constant query time and by our choice of $ h = n/k $, our total update time contains the term $ m n \Delta / k) $.
(2) Furthermore, we have $ O (k \log{(n \Delta)}) $ sinks whp by \Cref{lem:hitting_set_argument} (and can make the algorithm fail with small probability if not).
Thus, the total update time of the decremental algorithm for a set of sinks from the assumption is $ T (O (k \log{(n \Delta)}), m, n) $.
\end{proof}

\subsection{From \scc to \ssr}\label{sec:strongly connected component}

In the following we reduce decremental strongly connected components to decremental single-source reachability.
Our reduction is almost identical to the one of Roditty and Zwick~\cite{RodittyZ08}, but in order to work in our setting we have to generalize their running time analysis at the cost of losing a factor of $ \log n $.
They show that an $ O(m n) $ algorithm for single-source reachability implies an $ O(m n) $ algorithm for strongly connected components.
We show that in fact $ o(m n) $ time for single-source reachability implies $ o(m n) $ time for strongly connected components.
In the following we will often just write ``component'' instead of ``strongly connected component''.

In contrast to the rest of this paper, we will here impose the following technical condition on the decremental single-source reachability algorithm:
when we update the algorithm after the deletion of an edge, the update procedure will return all nodes that were reachable before the deletion, but are not reachable anymore after this deletion.
Note that all the reachability algorithms we present in this paper fulfill this condition.

\danupon{This sounds a bit sloppy. Maybe we can just say a bit more like this: This is because our single-source reachability algorithms are basically run an ES-tree on the original graph with some edge added, as in \Cref{thm: s-t to single source}. And ES-tree fulfills this requirement.}

\paragraph*{Algorithm.}
The algorithm works as follows.
For every component we, uniformly at random, choose among its nodes one \emph{representative}.
In an array, we store for every node a pointer to the representative of its component.
Queries that ask for the component of a node $ v $ are answered in constant time by returning (the ID of) the representative of $ v $'s component.
Using the decremental \ssr algorithm, we maintain, for every representative~$ w $ of a component $ C $, the sets $ I (w) $ and $ O (w) $ containing all nodes that reach $ w $ and that can be reached by $ w $, respectively.
Note that, for every node $ v $, we have $ v \in C $ if and only if $ v \in I (w) $ and $ v \in O (w) $.
After the deletion of an edge $ (u, v) $ such that $ u $ and $ v $ are contained in the same component $ C $ we check whether $ C $ decomposes.
This is the case only if, after the deletion, $ u \notin I (w) $ or $ v \notin O (w) $ (which can be checked with the \ssr algorithm of $ w $).

We now explain the behavior of the algorithm when a component $ C $ decomposes into the new components $ C_1, \ldots, C_k $.
The algorithm chooses a new random representative $ w_i $ for every component $ C_i $ and starts maintaining the sets $ I (w_j) $ and $ O (w_j) $ using two new decremental \ssr algorithms.
There is one notable exception: If the representative $ w $ of $ C $ is still contained in one of the components $ C_j $, then for this component we do \emph{not} choose a new representative.
Instead, $ C_j $ reuses $ w $ and its \ssr algorithms without any re-initialization.
The key to the efficiency of the algorithm is that a large component $ C_i $ has a high probability of inheriting the representative from $ C $.

Before choosing the new representatives we actually have to determine the new components $ C_1, \ldots, C_k $.
We slightly deviate from the original algorithm of Roditty and Zwick to make this step more efficient.
If $ w \in C_j $, then it is not necessary to explicitly compute $ C_j $ as all nodes in $ C_j $ keep their representative $ w $.
We only have to explicitly compute $ C_1, \ldots, C_{j-1}, C_{j+1}, \ldots C_k $.
This can be done as follows:
Let $ A $ denote the set of nodes that were contained in $ I (w) $ before the deletion of $ (u, v) $ and are not contained in $ I (w) $ anymore after this deletion.
Similarly, let $ B $ denote the set of nodes that were contained in $ O (w) $ before the deletion and are not contained in $ O (w) $ anymore afterwards.
The nodes in $ A \cup B $ are exactly those nodes of $ C $ that are not contained in~$ C_j $.
Let $ G' $ denote the subgraph of $ G $ induced by $ A \cup B $.
Then the components of $ G' $ are exactly the desired components $ C_1, \ldots, C_{j-1}, C_{j+1}, \ldots C_k $.
Note that the sets $ A $ and $ B $ are returned by the update-procedure of the \ssr algorithms of~$ w $, which allows us to compute $ A \cup B $.
The graph $ G' $ can be constructed by iterating over all outgoing edges of $ A \cup B $ and the components of $ G' $ can be found using a static \scc algorithm.

\paragraph*{Analysis.}
The correctness of the algorithm explained above is immediate.
For the running time we will argue that, up to a $\log{n}$-factor, it is the same as the running time of the \ssr algorithm.
When the running time of \ssr is of the form $ \tilde O (m^\alpha n^\beta) $, our argument works when $ \alpha \geq 1 $ or $ \beta \geq 1 $.
To understand the basic idea (for $ \beta \geq 1 $), consider the case that the graph decomposes into only two components $ C_1 $ and $ C_2 $ (with $ n_1, n_2 \leq n $ and $ m_1, m_2 \leq m $ being the corresponding number of nodes and edges).
We know that one of the two components still contains the representative $ w $.
For this component we do not have to spawn a new decremental \ssr algorithm.
This is an advantage for large components as they have a high probability of containing the representative.
The probability of $ w $ being contained in $ C_1 $ is $ n_1/n $ and it is $ n_2/n $ for being contained in $ C_2 $.
Thus, the expected cost of the decomposition is $ O (m_2^\alpha n_2^\beta n_1 / n + m_1^\alpha n_1^\beta  n_2 / n) $.
We charge this cost to the smaller component, say $ C_1 $.
As $ C_1 $ has $ n_1 $ nodes, the average cost we charge to every node in~$ C_1 $ is $ O(m_2^\alpha n_2^\beta / n + m_1^\alpha n_1^{\beta-1} n_2 / n) $.
This amounts to an average cost of $ O(m^\alpha n^{\beta-1}) $ per node.
As we will charge each node only when the size of its component has halved, the total update time is $ O(m^\alpha n^\beta \log{n})$.

\danupon{Be careful with edges outside both components.}

\begin{theorem}[From \ssr to \scc]\label{thm:strongly connected component} 
If there is a decremental \ssr algorithm with constant query time and a total update time of $ O (n t_1 (m, n) + m t_2 (m, n)) $ such that $ t_1 (m, n) \geq 1 $, $ t_2 (m, n) \geq 1 $, and $ t_1 (m, n) $ and $ t_2 (m, n) $ are non-decreasing\footnote{The technical assumption that $ t_ 1(m, n) $ and $ t_2 (m, n) $ are non-decreasing in $ m $ and $ n $ is natural as usually the running time of an algorithm does not improve with increasing problem size.} in $ m $ and $ n $, then there exists a decremental \scc algorithm with constant query time and an expected total update time of $ O ((n t_1 (m, n) + m t_2 (m, n)) \log{n}) $.
\end{theorem}

\begin{proof}
Let us first analyze the costs related to the decomposition of a component.
Assume that the component $ C_0 $ decomposes into $ C_1, \ldots, C_k $.
Let $ n_i $ and $ m_i $ denote the number of nodes and edges in component $ i $, respectively.
Let $ m_i' $ denote the sum of the out-degrees of the nodes in $ C_i $ in the initial graph (i.e., before the first deletion).
Note that $ m_i' $ is an upper bound on $ m_i $.
Furthermore we have $ \sum_{i=1}^k n_i = n_0 $, $ \sum_{i=1}^k m_i \leq m_0 $, and $ \sum_{i=1}^k m_i' = m_0' $.

Assume that the representative $ w $ of $ C_0 $ is contained in $ C_i $ after the decomposition.
First of all, for every $ j \neq i $ we have to pay a cost of $ O (n_j t_1 (m_j, n_j) + m_j t_2 (m_j, n_j)) $ for initializing and updating the decremental \ssr algorithm of the new representative of~$ C_j $.
Second, we have to pay for computing the new components.
This consists of three steps: (a) computing $ A \cup B $, (b) computing $ G' $, and (c) computing the components of $ G' $.
Remember that $ A \cup B $ is the union of $ A $, the set of nodes that cannot reach $ w $ anymore after deleting $ (u, v) $, and $ B $, the set of nodes that $ w $ cannot reach anymore after deleting $ (u, v) $.
After deleting $ (u, v) $ the incoming \ssr algorithm of $ w $ outputs $ A $ and the outgoing \ssr algorithm of $ w $ outputs $ B $.
Thus, the cost of computing $ A \cup B $ can be charged to the reachability algorithms of $ w $ (which have to output $ A $ and $ B $ anyway).
The graph $ G' $ is the subgraph of $ G $ induced by the nodes in $ A \cup B $.
We construct $ G' $ by checking, for every node in $ A \cup B $, which of its outgoing edges stay in $ A \cup B $.
This takes time $ O (\sum_{j \neq i} m_j') $.
Using Tarjan's linear time algorithm~\cite{Tarjan72}, we can compute the strongly connected components of $ G' $ in the same running time.
As $ m_j' \geq m_j $ the total cost of $ C_j $ is $ O (n_j t_1 (m_j, n_j) + m_j' t_2 (m_j, n_j)) = O (n_j t_1 (m, n) + m_j' t_2 (m, n))$

By the random choice of the representatives, the probability that $ w $ is contained in $ C_i $ is $ n_i / n_0 $.
Thus, the expected cost of the decomposition of $ C_0 $ is proportional to
\begin{equation}\label{eq:two_parts_of_cost}
\sum_{i=1}^k \frac{n_i}{n_0} \sum_{j \neq i} (n_j t_1 (m, n) + m_j' t_2 (m, n)) = \sum_{i=1}^k \frac{n_i}{n_0} \sum_{j \neq i} n_j t_1 (m, n) + \sum_{i=1}^k \frac{n_i}{n_0} \sum_{j \neq i} m_j' t_2 (m, n) \, .
\end{equation}
We analyze each of these terms individually.

Consider first the cost of $ O (\sum_{i=1}^k n_i / n_0 \sum_{j \neq i} n_j t_1 (m, n)) $.
For every pair $ i, j $ such that $ i \neq j $ we have to pay a cost of $ O(n_i n_j (t_1 (m, n) + n) / n_0) $.
If $ n_i \leq n_j $ we charge this cost to the component $ C_i $, otherwise we charge it to $ C_j $ (i.e., we always charge the cost to the smaller component).
Note that the component to which we charge the cost has at most $ n_0 / 2 $ nodes (otherwise it would not be the smaller one).
For a fixed component $ i $, the total charge is proportional to
\begin{equation*}
\sum_{j \neq i} \frac{n_i n_j (t_1 (m, n) + n)}{n_0} = n_i (t_1 (m, n) + n) \cdot \frac{\sum_{j \neq i} n_j}{n_0} \leq n_i (t_1 (m, n) + n) \, .
\end{equation*}
We share this cost equally among the nodes in $ C_i $ and thus charge $ {O (t (m, n) + n)} $ to every node in $ C_i $.
Every time we charge a node, the size of its component halves.
Thus, every node is charged at most $ \log{n} $ times and the total update time for the first term in Equation~\eqref{eq:two_parts_of_cost} is $ O(n t(m, n) \log{n}) $.

Consider now the cost of $ O (\sum_{i=1}^k n_i / n_0 \sum_{j \neq i} m_j' t_2 (m, n)) $.
Note that
\begin{equation*}
\sum_{i=1}^k \frac{n_i}{n_0} \sum_{j \neq i} m_j' = \sum_{i=1}^k m_i' \sum_{j \neq i} \frac{n_j}{n_0} = \sum_{i=1}^k m_i' \frac{n_0 - n_i}{n_0} \, .
\end{equation*}
We now charge $ m_i' (n_0 - n_i) / n_0 $ to every component $ C_i $ ($ 1 \leq i \leq k $).
In particular we charge $ (n_0 - n_i) / n_0 $ to every edge $ (u, v) $ of the initial graph such that $ u \in C_i $.
We now argue that in this way every edge is charged only $ O (\log{n}) $ times, which will imply that the total update time for the second term in Equation~\eqref{eq:two_parts_of_cost} is $ O(m t_2 (m, n) \log{n}) $.
Consider an edge $ (u, v) $ and the component containing $ u $.
We only charge the edge $ (u, v) $ when the component containing $ u $ decomposes.
Let $ a_0 $ denote the initial number of nodes of this component and let $ a_p $ its number of nodes after the $p$-th decomposition.
As argued above, we charge $ (a_{p-1} - a_p) / a_{p-1} $ to $ (u, v) $ for the $p$-th decomposition.
Thus, for $ q $ decompositions we charge $ \sum_{1 \leq p \leq q} (a_{p-1} - a_p) / a_{p-1} $.
Now observe that
\begin{multline*}
\sum_{1 \leq p \leq q} \frac{a_{p-1} - a_p}{a_{p-1}} \leq 
 \sum_{1 \leq p \leq q} \sum_{i=0}^{a_{p-1} - a_p - 1} \frac{1}{a_{p-1}} 
 \leq \sum_{1 \leq p \leq q}  \sum_{i=0}^{a_{p-1} - a_p - 1} \frac{1}{a_{p-1} - i} \\
 = \sum_{1 \leq p \leq q} \sum_{i=a_p + 1}^{a_{p-1}} \frac{1}{i} = \sum_{i=a_q + 1}^{a_0} \frac{1}{i} \, .
\end{multline*}
Since $ a_0 \leq n $, this harmonic series is bounded by $ O (\log{n}) $.

Finally, we bound the initialization cost.
Let $ C_1, \ldots, C_k $ denote the initial components and let $ n_i $ and $ m_i $ denote the number of nodes and edges of component $ C_i $, respectively.
The initial components can be computed in time $ O(m) $ with Tarjan's algorithm~\cite{Tarjan72}.
Furthermore, each component starts two decremental \ssr algorithms and we have to pay for the total update time of these algorithms.
This time is proportional to 
\begin{equation*}
\sum_{i=1}^k (n_i t_1 (m_i, n_i) + t_2 (m_i, n_i)) \leq t_1 (m, n) \sum_{i=1}^k n_i + t_2 (m, n) \sum_{i=1}^k m_i \leq n t_1 (m, n) + m t_2 (m, n) \, .
\end{equation*}
\end{proof}


\section{Single-Source Single-Sink Reachability}\label{sec:sparse}
In this section we give an algorithm for maintaining a path from a source node~$ s $ to a sink node~$ t $ in a directed graph undergoing edge deletions, i.e., we solve the decremental \ssssr problem.
Using the reduction of \Cref{sec:sssp}, this implies an algorithm for the decremental single-source reachability (\ssr) problem.

\subsection{Algorithm Description}\label{sec:st_reach_algo_description}

Our $s$-$t$ reachability algorithm has a parameter $ k \geq 1 $ and for each $ 1 \leq i \leq k $ parameters $ b_i \leq n $ and $ c_i \leq n $.
We determine suitable choices of these parameters in \Cref{sec:running_time_sparse}.
For each $ 1 \leq i \leq k-1 $, our choice will satisfy $ b_i \geq b_{i+1} $ and $ c_i \geq 2 c_{i+1} $.
We also set $ b_{k+1} = 1 $, $ c_0 = n $, $ c_{k+1} = 1 $, and $ h_i = n / c_i $ for all $ 0 \leq i \leq k + 1 $.
Note that this implies $ h_{i+1} \geq 2 h_i $ for all $ 1 \leq i \leq k $.
Intuitively, $ b_i $ and $ c_i $ are roughly the number of $i$-hubs and $i$-centers used by our algorithm and $ h_i $ is the hop range of the $i$-centers.
In the algorithm we will often consider ordered pairs of centers of the form $ (x, y) $.

\paragraph*{Initialization.}
At the initialization (i.e., before the first deletion), our algorithm determines sets of nodes $ B_1 \supseteq B_2 \supseteq \dots \supseteq B_k $ and $ C_0 \supseteq C_1 \supseteq \dots \supseteq C_{k+1} $ as follows.
For each $ 1 \leq i \leq k $, we sample each node of the graph with probability $ a b_i \ln{(\Delta (G))} / n $ and each edge with probability $ a b_i \ln{(\Delta (G))} / m $ (for a large enough constant $ a $).
The set $ B_i $ then consists of the sampled nodes and the endpoints of the sampled edges.
We set $ C_0 = V $ and $ C_{k+1} = \{s, t\} $.
For each $ 1 \leq i \leq k $, we sample each node of the graph with probability $ a c_i \ln{(\Delta (G))} / n $ (for a large enough constant $ a $).
The set $ C_i $ then consists of the sampled nodes together with the nodes in $ C_{i+1} $.
For every $ 1 \leq i \leq k $ we call the nodes in $ B_i $ $i$-hubs and for every $ 0 \leq i \leq k+1 $ we call the nodes in $ C_i $ $i$-centers.
Note that, as $ \Delta (G) \leq m $ in unweighted graphs, the number of $i$-hubs is $ \tilde O (b_i) $ whp and the number of $i$-centers is $ \tilde O (c_i) $ by \Cref{lem:hitting_set_argument}.
We can make the algorithm fail with small probability if these sets are larger.

\paragraph*{Data Structures.}
Our algorithm uses the following data structures:
\begin{itemize}
\item For every $i$-hub $ z $ (with $ 1 \leq i \leq k $) an ES-tree up to depth $ 2 h_i $ in $ G $ (\emph{outgoing} tree of $ z $) and an ES-tree up to depth $ 2 h_i $ in the reverse graph of $ G $ (\emph{incoming} tree of~$ z $)
\item For every pair of $i$-centers $ (x, y) $ (with $ 1 \leq i \leq k $) the set of $i$-hubs from $ B_i $ linking $ x $ to $ y $. 
\item For every pair of $i$-centers $ (x, y) $ (with $ 1 \leq i \leq k $) a set of nodes $ \cQ (x, y, i) $, which is initially empty.
\item For every pair of $i$-centers $ (x, y) $ (with $ 0 \leq i \leq k $) a list of pairs of $(i+1)$-centers called $(i+1)$-parents of $ (x, y) $. 
\item For every pair of $i$-centers $ (x, y) $ (with $ 1 \leq i \leq k+1 $) a list of pairs of $(i-1)$-centers called $(i-1)$-children of $ (x, y) $.
\end{itemize}

\paragraph*{Maintaining Hub Links.}
For every $ 1 \leq i \leq k $, we say that an $i$-hub $ z $ \emph{links} an $i$-center $ x $ to an $i$-center $ y $ if $ \dist_G (x, z) \leq 2 h_i $ and $ \dist_G (z, y) \leq 2 h_i $.
The hub links of every pair of $i$-centers $ (x, y) $, for all $ 1 \leq i \leq k $, can be maintained as follows.
After every edge deletion in the graph, we report the deletion to the ES-trees maintained by the hubs.
Initially, and after each deletion, the level of a node $ v $ in the incoming (outgoing) tree of an $i$-hub $ z $ is at most $ 2 h_i $ if and only if $ \dist_G (v, z) \leq 2 h_i $ ($ \dist_G (z, v) \geq 2 h_i $).
Thus, we can check whether an $i$-hub $ z $ links an $i$-center $ x $ to an $i$-center~$ y $ by summing up the levels of $ x $ and $ y $ in the incoming and outgoing ES-tree of $ z $, respectively.
Thus, for every priority $ i $ and every pair of $i$-centers $ (x, y) $ we can initialize the set of $i$-hubs linking $ x $ to $ y $ by iterating over all hubs.
We can maintain these sets under the edge deletions in $ G $ as follows:
Every time the level of some $i$-center $ x $ in the incoming ES-tree of some $i$-hub $ z $ exceeds $ 2 h_i $, we iterate over all $i$-centers $ y $ and remove $ z $ from the set of hubs linking $ x $ to $ y $.
We proceed similarly if the level of some $i$-center $ y $ in the outgoing ES-tree of some $i$ hub $ z $ exceeds $ 2 h_i $.
In this way, we can also generate, after every edge deletion, a list of pairs of $i$-centers $ (x, y) $ such that $ x $ is not linked to $ y $ by an $i$-hub anymore (but was linked to an $i$-hub before the deletion).

\paragraph*{Main Algorithm.}
Initially, the algorithm computes the shortest path $ \pi $ from $ s $ to $ t $ in $ G $.
It then determines a sequence of $k$-centers $ v_1, \ldots, v_l $ on $ \pi $.
For each $ 1 \leq j \leq l-1 $, the algorithm now tries to maintain a path from $ v_j $ to $ v_{j+1} $ under the edge deletions in $ G $.
If it fails to do so, it recomputes the shortest path from $ s $ to $ t $ in $ G $.
We call this a \emph{refresh} operation.

For each $ 1 \leq j \leq l-1 $, the path from $ v_j $ to $ v_{j+1} $ is maintained as follows.
The first case is that there is a $k$-hub linking $ v_j $ to $ v_{j+1} $.
As long as such a hub exists, we know that there is a path from $ v_j $ to $ v_{j+1} $.
If there is no $k$-hub linking $ v_j $ to $ v_{j+1} $ anymore, the algorithm computes the path union $ \cP (v_j, v_{j+1}, 2 h_k, G) $.
It then recursively maintains the path from $ v_j $ to $ v_{j+1} $ in $ G [\cP (v_j, v_{j+1}, 2 h_k, G)] $ using $(k-1)$-centers and $(k-1)$-hubs.
To keep track of this recursive hierarchy, we make, for every $ 1 \leq j \leq l-1 $, each pair $ (v_j, v_{j+1}) $ an $(i-1)$-child of $ (x, y) $ and $ (x, y) $ an $i$-parent of $ (v_j, v_{j+1}) $.
The end of the recursion is reached in layer $ 0 $, where the set of $0$-centers consists of all nodes.
The only paths between $0$-centers the algorithm considers are the edges between them.

Algorithm~\ref{alg:sparse} shows the pseudocode for these procedures.
A crucial ingredient of our algorithm is the computation of path unions for pairs of centers.
For every $ 1 \leq i \leq k $ and all $i$-centers $ (x, y) $ the path union from the last recomputation is stored in $ \cQ (x, y, i) $.
To simplify the formulation of the algorithm we set $ \cQ (s, t, k+1) = V $ (which is consistent with our choice of $ h_{k+1} = n $).
Thus, the procedure \Refresh{$s$, $t$, $k+1$} simply checks whether there exists a path from $ s $ to $ t $ in $ G $.

\begin{algorithm}
\thisfloatpagestyle{empty} 
\caption{Algorithm for decremental $s$-$t$ reachability}
\label{alg:sparse}

\Procedure{\Refresh{$x$, $y$, $i$}}{
	Compute shortest path $ \pi $ from $ x $ to $ y $ in $ G [\cQ (x, y, i)] $\;
	\eIf{length of $ \pi > h_i $}{
		\If{$ i = k+1 $}{
			Stop and output ``$ s $ cannot reach $ t $''\;
		}
		\ForEach{$(i+1)$-parent $ (x', y') $ of $ (x, y) $}{
			\Refresh{$x'$, $y'$, $i+1$}\;\label{line:causing_refresh}
		}
	}{
		\RemoveChildren{$x$, $y$, $i$}\;
		Determine $ (i-1) $-centers $ v_1, \ldots, v_l $ on $ \pi $ in order of appearance on $ \pi $ such that $ v_1 = x $, $ v_l = y $ and, for each $ 1 \leq j \leq l-2 $, $ v_{j+1} $ is the first $(i-1)$-center following $ v_j $ on $ \pi $ at distance at least $ h_{i-1} /2 $ from $ v_j $.\;\label{lin:determine_children}
		\tcp{If $ i = 1 $ then the $0$-centers are all nodes on $ \pi $ because $ h_0 = 1 $.} 
		\ForEach{$ 1 \leq j \leq l-1 $}{
			Make $ (v_j, v_{j+1}) $ an $(i-1)$-child of $ (x, y) $ and $ (x, y) $ an $i$-parent of $ (v_j, v_{j+1}) $\;
			\If{$ v_{j+1} $ not linked to $ v_{j} $ by an $i$-hub}{
				\ComputePathUnion{$v_j$, $v_{j+1}$, $i-1$}\;
				\Refresh{$v_j$, $v_{j+1}$, $i-1$}\;
			}
		}
	}
}

\Procedure{\ComputePathUnion{$x$, $y$, $i$}}{
	\eIf{$ \cQ (x, y, i) = \emptyset $}{
		Let $ (x', y') $ be any $(i+1)$-parent of $ (x, y) $\;
		$ \cQ (x, y, i) \gets \cP (x, y, 2 h_i, G [\cQ (x', y', i+1)]) $\;\label{line:path_union_initialization}
	}{
		$ \cQ (x, y, i) \gets \cP (x, y, 2 h_i, G [\cQ (x, y, i)]) $\;\label{line:path_union_update}
	}
}

\Procedure{\RemoveChildren{$x$, $y$, $i$}}{
	\ForEach{$(i-1)$-child $ (x', y') $ of $ (x, y) $}{
		Remove $ (x', y') $ from $(i-1)$-children of $ (x, y) $ and remove $ (x, y) $ from $i$-parents of $ (x', y') $\;
		\If{$ (x', y') $ has no $i$-parents anymore and $ i \geq 1 $}{
			\RemoveChildren{$x'$, $y'$, $i-1$}\;
		}
	}
}

\Procedure{\Initialize{}}{
	\lForEach{$ 1 \leq i \leq k $ and all $i$-centers $ (x, y) $}{
		$ \cQ (x, y, i) \gets \emptyset $
	}
	$ \cQ (s, t, k+1) \gets V $\;
	\Refresh{$s$, $t$, $k+1$}\;
}

\Procedure{\Delete{$u$, $v$}}{
	\ForEach{$1$-parent $ (x, y) $ of $ (u, v) $}{
		\Refresh{$x$, $y$, $1$}\;
	}
	For every $ 1 \leq i \leq k $ and every pair of $i$-centers $ (x, y) $ that is not linked by an $i$-hub anymore: \Refresh{$x$, $y$, $i$}
}
\end{algorithm}

\subsection{Correctness}

It is obvious that the algorithm correctly maintains, for all pairs of $i$-centers $ (x, y) $, the set of $i$-hubs linking $ x $ to $ y $.
Furthermore, the algorithm only stops if it has correctly detected that there is no path from $ s $ to $ t $ in the current graph $ G $.
Thus, it remains to show that, as long as the algorithm does not stop, there there is a path from $ s $ to $ t $.
We say that a pair of $i$-centers $ (x, y) $ is \emph{active} if it has at least one $ (i+1) $-parent.
We further define the pair $ (s, t) $, the only pair of $(k+1)$-centers, to be active.

\begin{lemma}
After finishing the delete-operation, for every $ 0 \leq i \leq k+1 $ and every active pair of $i$-centers $ (x, y) $, there is a path from $ x $ to $ y $ in the current graph $ G $.
\end{lemma}

\begin{proof}
The proof is by induction on $ i $.
If $ y $ is linked to $ x $ by some $i$-hub $ z $, we know that in $ G $ there is a path from $ x $ to $ z $ as well as a path from $ z $ to $ y $.
Their concatenation is a path from $ x $ to $ y $ in $ G $.

Consider now the case that there is no $i$-hub linking $ x $ to $ y $.
If $ i = 0 $, we know that there is an edge from $ x $ to $ y $ in $ G $ as otherwise the pair $ (x, y) $ would not be active. 
If $ i \geq 1 $, let $ (v_1, v_2), (v_2, v_3), \ldots, (v_{l-1}, v_l) $ with $ v_1 = x $ and $ v_l = y $ denote the $(i+1)$-children of $ (x, y) $ found at the time of the last refresh of $ (x, y) $.
Note that all these children are active since $ (x, y) $ is their $i$-parent.
Thus, by the induction hypothesis, we know that there is a path from $ v_j $ to $ v_{j+1} $ in $ G $ for all $ 1 \leq j \leq l-1 $.
The concatenation of these paths gives a path from $ x $ to $ y $ in $ G $, as desired.
\end{proof}

\subsection{Running Time Analysis}\label{sec:running_time_sparse}

\paragraph*{Important Properties.}
In the running time analysis we will need some properties of the path unions computed by our algorithm.

\begin{lemma}\label{lem:path_union_contained_in_higher_priority}
Let $ (x', y') $ be pair of nodes such that $ \dist_G (x', y') \leq h_{i+1} $ and let $ (x, y) $ be a pair of nodes such that $ x $ and $ y $ lie on a shortest path from $ x $ to $ y $ in $ G $ (and $ x $ appears before $ v $ on this shortest path).
Then $ \cP (x, y, 2 h_i, G) \subseteq \cP (x', y', 2 h_{i+1}, G) $.
\end{lemma}

\begin{proof}
Let $ v \in \cP (x, y, 2 h_i, G) $.
We first apply the triangle inequality:
\begin{align*}
\dist_G (x', v) + \dist_G (v, y') \leq \dist_G (x', x) + \dist_G (x, v) + \dist_G (v, y) + \dist_G (y, y') \, .
\end{align*}
Since $ x $ and $ y $ (in this order) lie on the shortest path from $ x' $ to $ y' $, we have $ \dist_G (x', x) + \dist_G (y, y') \leq \dist_G (x', y') $ and by our assumption we have $ \dist_G (x', y') \leq h_{i+1} $.
Since  $ v \in \cP (x, y, 2 h_i, G) $, we have $ \dist_G (x, v) + \dist_G (v, y) \leq 2 h_i $
As $ 2 h_i \leq h_{i+1} $ by our assumptions on the parameters we get
\begin{equation*}
\dist_G (x', v) + \dist_G (v, y') \leq h_{i+1} + 2 h_i \leq h_{i+1} + h_{i+1} = 2 h_{i+1}
\end{equation*}
which implies that $ v \in \cP (x', y', 2 h_{i+1}, G) $.
\end{proof}

\begin{lemma}\label{lem:supergraph_of_path_union}
For every $ 1 \leq i \leq k $ and every pair of $i$-centers $ (x, y) $, if $ \cQ (x, y, i) \neq \emptyset $, then $ \cP (x, y, 2 h_i, G) \subseteq \cQ (x, y, i) $.
\end{lemma}

\begin{proof}
The proof is by induction on $ i $.
We first argue that it is sufficient to show that after the first initialization of $ \cQ (x, y, i) $ in Line~\ref{line:path_union_initialization} we have $ \cQ (x, y, i) = \cP (x, y, 2 h_i, G) $.
After the initialization, and before $ \cQ (x, y, i) $ is recomputed, $ \cQ (x, y, i) $ might be ``outdated'' due to deletions in $ G $, i.e., not be equal to $ \cP (x, y, 2 h_i, G) $ anymore.
However, since $ \cP (x, y, 2 h_i, G) $ only ``loses'' nodes through the deletions in $ G $, it is still the case that $ \cP (x, y, 2 h_i, G) \subseteq \cQ (x, y, i) \subseteq V $.
Thus, this is also true directly before $ \cQ (x, y, i) $ is recomputed for the first time in Line~\ref{line:path_union_update}.
There $ \cQ (x, y, i) $ is updated to $ \cP (x, y, 2 h_i, G [\cQ (x, y, i)]) $.
By \Cref{lem:path_union_in_supergraph} we know that $ \cP (x, y, 2 h_i, G [\cQ (x, y, i)]) = \cP (x, y, 2 h_i, G) $, i.e., after the recomputation, $ \cQ (x, y, i) $ is equal to $ \cP (x, y, 2 h_i, G) $ again.
By repeating this argument, we get that $ \cP (x, y, 2 h_i, G) \subseteq \cQ (x, y, i) \subseteq V $.

We now show that after the first initialization of $ \cQ (x, y, i) $ in Line~\ref{line:path_union_initialization} we have $ \cQ (x, y, i) = \cP (x, y, 2 h_i, G) $.
Let $ (x', y') $ be the $(i+1)$-parent of $ (x, y) $ used in the initialization.
We will show that $ \cP (x, y, 2 h_i, G [\cQ (x', y', i+1)]) = \cP (x, y, 2 h_i, G) $.
For $ i = k $ the claim is trivially true because $ (s, t) $ is the only $(k+1)$-parent of $ (x, y) $ and $ \cQ (s, t, k+1) = V $.
For $ i \leq k-1 $, we know by the induction hypothesis that $ \cP (x', y', 2 h_{i+1}, G) \subseteq \cQ (x', y', i+1) $.
We also know that $ x $ and $ y $ lie on the shortest path from $ x' $ to $ y' $ in $ G [\cQ (x', y', i+1)] $ such that $ x $ precedes $ y $ on this shortest path and the path has length at most $ h_i $.
As $ G [\cQ (x', y', i+1)] $ is a subgraph of $ G $, this path is also contained in $ G $.
Thus, $ \dist_G (x', x) + \dist_G (y, y') \leq h_i $.
By \Cref{lem:path_union_contained_in_higher_priority} this implies that $ \cP (x, y, 2 h_i G) \subseteq \cP (x', y', 2 h_{i+1}, G) $.
It follows that $ \cP (x, y, 2 h_i, G) \subseteq \cQ (x', y', i+1) $.
Thus, by \Cref{lem:path_union_in_supergraph}, we have $ \cP (x, y, 2 h_i, G [\cQ (x', y', i+1)]) = \cP (x, y, 2 h_i, G) $ as desired.
\end{proof}

\begin{lemma}\label{lem:subgraph_property_of_child}
Let $ (x, y) $ be an active pair of $i$-centers and let $ (x', y') $ be an $(i+1)$-parent of $ (x, y) $.
Then $ \cQ (x, y, i) \subseteq \cQ (x', y', i+1) $.
\end{lemma}

\begin{proof}
Consider the point in time when $ (x', y') $ becomes an $(i+1)$-parent of $ (x, y) $.
This can only happen in the else-branch of \Refresh{$x$, $y$, $i+1$}.
Then we know that $ x $ and $ y $ lie on a path from $ x' $ to $ y' $ of length at most $ h_{i+1} $ in $ G [\cQ (x', y', i+1)] $ such that $ x $ precedes $ y $ on this path and thus $ \dist_G (x', x) + \dist_G (y, y') \leq h_{i+1} $.
From \Cref{lem:path_union_contained_in_higher_priority} it follows that $ \cP (x, y, 2 h_i, G) \subseteq \cP (x', y', 2 h_{i+1}, G) $.
Furthermore, $ \cP (x', y', 2 h_{i+1}, G) \subseteq \cQ (x', y', i+1) $ by \Cref{lem:supergraph_of_path_union}.
As the path-union of $ (x, y) $ is recomputed when $ (x, y) $ becomes an $i$-child of $ (x', y') $, we also have $ \cQ (x, y, i) = \cP (x, y, 2 h_i, G) $. 
Therefore we get that $ \cQ (x, y, i) \subseteq \cQ (x', y', i+1) $.
Furthermore, every time the algorithm recomputes $ \cQ (x', y', i+1) $, $ (x', y') $ will stop being an $(i+1)$-parent of $ (x, y) $.
It might immediately become in $(i+1)$-parent again and in this case our argument above applies again.
\end{proof}

\begin{lemma}\label{lem:distance_increase_path_union}
For every $ 1 \leq i \leq k+1 $ and every pair of $i$-centers $ (x, y) $, we have that if $ \dist_{G [\cQ (x, y, i)]} (x, y) > h_i $, then $ \dist_G (x, y) > h_i $.
\end{lemma}

\begin{proof}
We show that $ \dist_G (x, y) \leq h_i $ implies $ \dist_{G [\cQ (x, y, i)]} (x, y) \leq h_i $.
If $ \dist_G (x, y) \leq h_i $, then also $ \dist_{G [\cP (x, y, 2 h_i, G)]} (x, y) \leq h_i $ as $ \cP (x, y, 2 h_i, G) $ contains all paths from $ x $ to $ y $ in $ G $ of length at most $ 2 h_i $.
Since $ \cP (x, y, 2 h_i, G) \subseteq \cQ (x, y, i) $ by \Cref{lem:supergraph_of_path_union}, we have $ \dist_{G [\cQ (x, y, i)]} (x, y)  \leq \dist_{G [\cP (x, y, 2 h_i G)]} (x, y) $.
It follows that $ \dist_{G [\cQ (x, y, i)]} (x, y) \leq h_i $.
\end{proof}

\begin{lemma}\label{lem:size_of_path_union}
For every $ 1 \leq i \leq k+1 $ and every pair of $i$-centers $ (x, y) $ the graph $ G [\cQ (x, y, i)] $ has at most $ \min(m / b_i, n^2 / b_i^2) $ edges whp.
\end{lemma}

\begin{proof}
As no edges are ever added to $ G $, we only have to argue that the claim is true at the first initialization of $ \cQ (x, y, i) $ in \Cref{line:path_union_initialization}, where $ \cQ (x, y, i) $ is equal to $ \cP (x, y, 2 h_i, G) $  (see proof of \Cref{lem:subgraph_property_of_child}). 
We argue that at that time $ G [\cP (x, y, 2 h_i, G)] $ as at most $ \min(m / b_i, n^2 / b_i^2) $ edges.
Note to this end it is sufficient to show that $ G [\cP (x, y, 2 h_i, G)] $ has at most $ n / b_i $ nodes and $ m / b_i $ edges.

Suppose that $ \cP (x, y, 2 h_i, G) $ contains more than $ n / b_i $ nodes.
Then, by the random sampling of $i$-hubs, one of these nodes, say $ v $, would have been sampled whp while determining the set $ B_i $ at the initialization by \Cref{lem:hitting_set_argument}, making $ v $ an $i$-hub.
If $ G [\cP (x, y, 2 h_i, G)] $ contains more than $ m / b_i $ edges, then one of these edges, say $ (u, v) $, would have been sampled whp while determining the set $ B_i $ at the initialization by \Cref{lem:hitting_set_argument}, making both $ u$ and $ v $ $i$-hubs.

In both cases, $ \cP (x, y, 2 h_i, G) $ contains some $i$-hub $ v $, for which $ \dist_G (x, v) + \dist_G (v, y) \leq 2 h_i $ by \Cref{pro:path_union_characterization}.
But this means that $ x $ is linked to $ y $ by the $i$-hub $ v $ and the algorithm would not have executed \Cref{line:path_union_initialization}, which contradicts our assumption.
\end{proof}

\begin{lemma}\label{lem:number_of_children_containing_node}
Let $ 1 \leq i \leq k $ and consider a pair of active $(i+1)$-centers $ (x', y') $ and their $i$-children $ (x_j, y_j)_{1 \leq j \leq l} $ (which are active $i$-centers).
Then, for every node $ v $, there are at most $ q = 8 $ pairs of $i$-children $ (x_j, y_j) $ of $ (x', y') $ such that $ v \in \cQ (x_j, y_j, i) $.
\end{lemma}

\begin{proof}
We show that at the last time the algorithm has called \Refresh{$x'$, $y'$, $i+1$} (where it determined the current $i$-children of $ (x', y') $) there are at most $ q = 8 $ pairs of $i$-children $ (x_j, y_j) $ of $ (x', y') $ such that $ v \in \cP (x_j, y_j, 2 h_i, G) $ as, for each $ 1 \leq i \leq l $, $ \cQ (x_j, y_j, i) $ (even if initialized later) will always be a subset of this set.

Suppose that $ v $ is contained in $ q > 8 $ path unions $ \cP (x_j, y_j, 2 h_i, G) $ of $i$-children $ (x_j, y_j) $ of $ (x', y') $.
Let $ j_1, \ldots, j_q $ be the corresponding indices and assume without loss of generality that $ j_1 < j_2 \ldots < j_q $.
The children of $ (x', y') $ all lie on the shortest path from $ x' $ to $ y' $ and by the way we have selected them we even more have $ \dist_G (x_j, y_j) \geq h_i / 2 $ for all $ 1 \leq j \leq l $ and thus
\begin{equation}
\dist_G (x_{j_1}, y_{j_q}) \geq \sum_{j \in \{ j_1, \ldots, j_q \} } \dist_G (x_j, y_j) \geq \sum_{j \in \{ j_1, \ldots, j_q \} } h_i / 2 = q h_i / 2 \, . \label{eq:lower_bound}
\end{equation}

Furthermore, since $ v \in \cP (x_{j_1}, y_{j_1}, 2 h_i, G) $, $ v $ lies on a shortest path from $ x_{j_1} $ to $ y_{j_1} $ of length at most $ h_i $ in $ G $ and thus $ \dist_G (x_{j_1}, v) \leq 2 h_i $.
By the same argument we have $ \dist_G (v, y_{j_q}) \leq 2 h_i $.
We now have
\begin{equation}
\dist_G (x_{j_1}, y_{j_q}) \leq \dist_G (x_{j_1}, v) + \dist_G (v, y_{j_q}) \leq 4 h_i \label{eq:upper_bound}
\end{equation}
by the triangle inequality.
Observe that the upper bound~\ref{eq:upper_bound} and the lower bound~\ref{eq:lower_bound} contradict each other for $ q > 8 $ and thus $ q \leq 8 $.
\end{proof}

\begin{lemma}\label{lem:number_of_path_unions_containing_node}
For every $ 1 \leq i \leq k+1 $ and every node $ v $, there are at most $ q^{k-i+1} $ pairs of active $i$-centers $ (x, y) $ such that $ v \in \cQ (x, y, i) $.
\end{lemma}

\begin{proof}
The proof is by induction on $ i $.
The base case is $ i = k+1 $ and is trivially true because the only pair of $(k+1)$-centers is the pair $ (s, t) $.
Consider now  the case $ 1 \leq i \leq k $ and fix some node $ v $.
Consider a pair $ (x, y) $ of active $i$-centers such that $ v \in \cQ (x, y, i) $.
Let $ (x', y') $ be a pair of $(i+1)$-centers that is an $(i+1)$-parent of $ (x, y) $ (such a parent must exist because otherwise $ (x, y) $ would not be active).
Since $ \cQ (x, y, i) \subseteq \cQ (x', y', i+1) $ by \Cref{lem:subgraph_property_of_child}, $ v \in \cQ (x', y', i+1) $.
Thus, $ v \in \cQ (x, y, i) $ only if there is an $(i+1)$-parent $ (x', y') $ such that $ v \in \cQ (x', y', i+1) $.

By the induction hypothesis, the number of pairs of active $(i+1)$-centers $ (x', y') $ such that $ v \in \cQ (x', y', i+1) $ is at most $ q^{k-(i+1)+1} = q^{k-i} $.
Let $ (x', y') $ be such a pair of $(i+1)$-centers.
By \Cref{lem:number_of_children_containing_node}, the number of $i$-children $ (x, y) $ of $ (x', y') $ such that $ v \in \cQ (x, y, i) $ is at most $ q $.
Therefore the total number of active pairs of $i$-centers $ (x, y) $ such that $ v \in \cQ (x, y, i) $ is at most $ q \cdot q^{k-i} = q^{k-i+1} $ as desired.
\end{proof}

\begin{lemma}\label{lem:bound_on_number_of_parents}
For every $ 0 \leq i \leq k $, each active pair of $i$-centers has at most $ q^{k-i} \leq q^k $ $(i+1)$-parents.
\end{lemma}

\begin{proof}
Consider any active pair of $i$-centers $ (x, y) $ and fix some node $ v \in \cQ (x, y, i) $.
Let the pair of $(i+1)$-centers $ (x', y') $ be an $(i+1)$-parent of $ (x, y) $.
By \Cref{lem:subgraph_property_of_child} $ \cQ (x, y, i) \subseteq \cQ (x', y', i+1) $.
Thus, $ v \in \cQ (x', y', i+1) $ for every $(i+1)$-parent of $ (x, y) $.
If $ (x, y) $ had more than $ q^{k-i} $ $(i+1)$-parents, then $ v $ would be contained in more than $ q^{k-i} = q^{k-(i+1)+1} $ path unions $ \cQ (x', y', i+1) $ of $ (i+1)$-parents $ (x', y') $, contradicting \Cref{lem:number_of_path_unions_containing_node}.
\end{proof}

\paragraph*{Maintaining Hub Links.}
For every $i$-hub, we maintain an incoming and an outgoing ES-tree of depth $ 2 h_i = 2 n / c_i $, which takes time $ O (m n / c_i) $.
As there are $ \tilde O (b_i) $ hubs of priority $ i $, the total time needed for maintaining all these ES-trees is $ \tilde O (\sum_{1 \leq i \leq k} b_i m n /c_i) $.
For every pair of $i$-centers $ (x, y) $, the list of hubs linking $ x $ to $ y $ is initialized by iterating over all $i$-hubs.
As there are $ \tilde O(c_i) $ $i$-centers and $ \tilde O(b_i) $ $i$-hubs, this takes time $ \tilde O (\sum_{1 \leq i \leq k} b_i c_i^2) $.
Every time the level of an $i$-center $ x $ in the ES-tree of an $i$-hub $ z $ exceeds $ 2 h_i $, we have to remove $ z $ from the set of hub links for every possible partner $ y $ of $ x $.
As this event can occur only once for every $i$-center $ x $ and every $i$-hub $ z $ over the course of the algorithm, maintaining the sets of linking hubs takes time $ \tilde O (\sum_{1 \leq i \leq k} b_i c_i^2) $.

\paragraph*{Computing Path Unions.}
We now argue about the time needed for maintaining the sets $ \cQ (x, y, i) $ for all pairs of $i$-centers $ (x, y) $ and each $ 1 \leq i \leq k $.
We will show that we can pay for this cost by charging $ O (\min(m / b_{i+1}, n^2 / b_{i+1}^2)) $ to every pair of $i$-centers and $ O (q^k \min(m / b_{i+1}, n^2 / b_{i+1}^2)) $ to every refresh operation of the form \Refresh{$ x' $, $ y' $, $i+1$} for some pair of $(i+1)$-centers $ (x', y') $.
Note that when $ i = k $ then $ \min(m / b_{k+1}, n^2 / b_{k+1}^2) = m $ since we have set $ b_{k+1} = 1 $.

Fixing some pair of $i$-centers $ (x, y) $, we first bound the cost for the first initialization of $ \cQ (x, y, i) $, as performed in Line~\ref{line:path_union_initialization} of \Cref{alg:sparse}.
There we have to compute $ \cP (x, y, 2 h_i, G [\cQ (x', y', i+1)]) $, where $ (x', y') $ is an $(i+1)$-parent of $ (x, y) $.
By \Cref{lem:path_union_computation} this takes time proportional to the number of edges in $ G [\cQ (x', y', i+1)] $.
If $ i = k $, then $ \cQ (x', y', i+1) = V $ (and actually $ x' = s $ and $ y' = t) $ and thus computing $ \cP (x, y, 2 h_i, G) $ for all pairs of $k$-centers $ (x, y) $ takes time $ \tilde O (c_k^2 m) $.
If $ 1 \leq i \leq k-1 $, then $ G [\cQ (x', y', i+1)] $ has $ O (\min(m / b_{i+1}, n^2 / b_{i+1}^2)) $ edges by \Cref{lem:size_of_path_union}.
Thus, the first computation of $ \cQ (x, y, i) $ for all pairs of $i$-centers $ (x, y) $ takes time $ \tilde O (c_i^2 \min(m / b_{i+1}, n^2 / b_{i+1}^2)) $.

Now consider the cost of computing $ \cQ (x, y, i) $ after it has already been initialized for the first time.
Let $ \cQ' (x, y, i) = \cP (x, y, 2 h_i, G [\cQ (x, y, i)]) $ denote the updated path union as computed in Line~\ref{line:path_union_update} of \Cref{alg:sparse}.
The cost of computing $ \cQ' (x, y, i) $ is proportional to $ | E [\cQ (x, y, i)] | $, the number of edges in $ G [cQ (x, y, i)] $ before the recomputation.
Since $ \cQ' (x, y, i) \subseteq \cQ (x, y, i) $, we have $ E [\cQ (x, y, i)] = E [\cQ' (x, y, i)] \cup E [\cQ (x, y, i)] \setminus E [\cQ' (x, y, i)] $.
Note that $ \cQ' (x, y, i) $ is equal to $ \cP (x, y, 2 h_i, G) $. 
We pay for this cost by charging $ O(| E [\cQ (x, y, i)] \setminus E [\cQ' (x, y, i)] |) $ to the pair $ (x, y) $ and $ | E [\cP (x, y, 2 h_i, G)] | $ to the refresh operation on the $(i+1)$-parent of $ (x, y) $ which causes the recomputation.

As the initial size of $ G [\cQ (x, y, i)] $ is $ O (\min(m / b_{i+1}, n^2 / b_{i+1}^2)) $ and we only charge edges to the pair $ (x, y) $ that will never be contained in $ G [\cQ (x, y, i)] $ anymore, the total time charged to $ (x, y) $ is $ O (\min(m / b_{i+1}, n^2 / b_{i+1}^2)) $, which, over all path union computations, results in a cost of $ O (\sum_{1 \leq i \leq k-1 } c_i^2 \min(m / b_{i+1}, n^2 / b_{i+1}^2) + c_k^2 m) $.
It remains to bound the total cost charged to each refresh operation \Refresh{$ x' $, $ y' $, $ i+1 $} for some pair of $ (i+1) $-centers $ (x', y') $.
Below we will then separately analyze the total cost of the refresh operations.
During the refresh operation we recompute $ \cQ (x, y, i) $ for every $i$-child $ (x, y) $ of $ (x', y') $, for which we have to pay $ O (| E [\cP (x, y, 2 h_i, G)] |) $ per child.
By \Cref{lem:subgraph_property_of_child} $ \cQ (x, y, i) \subseteq \cQ (x', y', i+1) $ for every $i$-child $ (x, y) $ of $ (x', y') $.
By \Cref{lem:number_of_path_unions_containing_node}, each node is contained in at most $ q^k $ path unions of $i$-children of $ (x', y') $.
Therefore every edge of $ G [\cQ (x', y', i+1)] $ is contained in the subgraphs of $ G $ induced by the $i$-path union of at most $ q^k $ $i$-children of $ (x', y') $.
Thus, the total cost charged to the refresh operation is $ O (q^k | E [\cQ (x', y', i+1)] |) $, which is $ O (q^k \min(m / b_{i+1}, n^2 / b_{i+1}^2)) $.

\paragraph*{Cost of Refresh.}
Excluding the recursive calls, each refresh operation of the form \Refresh{$ x $, $ y $, $ i $}, where $ 1 \leq i \leq k + 1 $ and $ (x, y) $ is a pair of $i$-centers, is dominated by two costs: (1) the time needed for computing the shortest path from $ x $ to $ y $ in $ G [\cQ (x, y, i)] $ and (2) the time needed for computing $ {\cQ (x', y', i-1)} $ for every ${(i-1)}$-child $ (x', y') $ of $ (x, y) $ if $ i \geq 2 $.
As $ G [\cQ (x, y, i)] $ has at most $ \min(m / b_{i+1}, n^2 / b_{i+1}^2) $ edges by \Cref{lem:size_of_path_union}, computing the shortest path takes time $ \tilde O(\min(m / b_{i+1}, n^2 / b_{i+1}^2)) $.
We have argued above that to pay for step (2) the time we charge to that particular refresh is $ O(q^k \min(m / b_{i+1}, n^2 / b_{i+1}^2)) $.
It remains to analyze how often the refresh operation is called.

We say that a pair of $i$-centers $ (x, y) $ (for $ 0 \leq i \leq k $) causes a refresh if the algorithm calls \Refresh ($x'$, $y'$, $i+1$) (in line~\ref{line:causing_refresh}), where $ (x', y') $ is an $(i+1)$-parent of $ (x, y) $, after detecting that the distance from $ x $ to $ y $ in $ G [\cQ (x, y, i)] $ is more than $ h_i $.
If a fixed pair of $i$-centers $ (x, y) $ causes a refresh after some deletion in the graph, it will do so for each of its $(i+1)$-parents.
By \Cref{lem:bound_on_number_of_parents} the number of $(i+1)$-parents is at most $ q^k $.
Note that each pair of $i$-centers $ (x, y) $ will only cause these $ q^k $ refreshes of its parents once.
At the time the algorithm makes $ (x, y) $ the child of some $ (i+1) $-center we have $ \dist_G (x, y) \leq h_i $ whp as by the initial random sampling of $i$-centers, every shortest path consisting of $ h_i/2 - 1 $ edges contains an $i$-center whp (\Cref{lem:hitting_set_argument}).
Furthermore, the pair $ (x, y) $ will never cause a refresh anymore in the future as the refresh implies that $ \dist_G (x, y) > h_i $ by \Cref{lem:distance_increase_path_union} and thus $ (x, y) $ will never be active anymore. 

Now whenever we refresh a pair of $(i+1)$-centers $ (x', y') $, we charge the running time of $ \tilde O (q^k \min(m / b_{i+1}, n^2 / b_{i+1}^2)) $ to the $i$-child $ (x, y) $ causing the refresh.
By the argument above, each pair of $i$-centers will be charged at most $ q^k $ times.
Thus, the total time needed for all refresh operations on pairs of $(i+1)$-centers over the course of the algorithm is $ \tilde O (q^{2 k} c_i^2 \min(m / b_{i+1}, n^2 / b_{i+1}^2)) $ if $ i \geq 1 $.
For $ i = 0 $, we can bound this by $ \tilde O (q^{2k} m \cdot \min(m / b_1, n^2 / b_1^2)) $ because every node is a $0$-center and thus a pair of $0$-centers $ (x, y) $ can only be active if the graph contains the edge $ (x, y) $.

\paragraph*{Total Running Time.}\label{sec:running_time_st_reach}
Putting everything together, the total running time of our algorithm using $ k $ layers is
\begin{multline*}
\tilde O \left( \sum_{1 \leq i \leq k} b_i c_i^2 + \sum_{1 \leq i \leq k}  \frac{b_i m n}{c_i} + q^{2k} m \cdot \min \left( \frac{m}{b_1}, \frac{n^2}{b_1^2} \right) 
\right. \\
\left.
+ \sum_{1 \leq i \leq k-1} q^{2k} c_i^2 \min \left( \frac{m}{b_{i+1}}, \frac{n^2}{b_{i+1}^2} \right) + q^{2k} c_k^2 m \right) \, .
\end{multline*}
We first balance the terms to obtain a running time of $ \tilde O (m^{5/4} n^{1/2}) $.
We achieve this by setting the parameters to $ k = \lceil \log{\log{m}} \rceil $ and, for every $ 1 \leq i \leq k $,
\begin{align*}
b_i &= \frac{m^{\frac{3 \cdot (2^k - 2^{(i-1)})}{2^{(k+2)}-3}}}{n^{\frac{2^{(k+1)} - 2^i}{2^{(k+2)}-3}}} \\
c_i &= 2^{k-i} m^{\frac{2^{(k+1)} - 3 \cdot 2^{(i-1)}}{2^{(k+2)}-3}} n^{\frac{2^i - 1}{2^{(k+2)}-3}} \, .
\end{align*}
With this choice of the parameters we get (for all $ 1 \leq i \leq k $), $ b_i \geq b_{i+1} $, $ c_i \geq 2 c_{i+1} $,
\begin{align*}
\frac{b_i m n}{c_i} &\leq \frac{m^2}{b_1} = m^{\frac{5 \cdot 2^k-3}{2^{(k+2)}-3}} n^{\frac{2^{(k+1)} - 2}{2^{(k+2)}-3}} \, , \\
\frac{c_i^2 m}{b_{i+1}} &\leq c_k^2 m \leq 2^{2k} m^{\frac{5 \cdot 2^k-3}{2^{(k+2)}-3}} n^{\frac{2^{(k+1)} - 2}{2^{(k+2)}-3}} \, \text{, and} \\
b_i c_i^2 &\leq 2^{2k} \frac{m^{\frac{7 \cdot 2^k - 9 \cdot 2^{(i-1)}}{2^{(k+2)}-3}}}{n^{\frac{2^{(k+1)} - 3 \cdot 2^i + 2}{2^{(k+2)}-3}}} \leq 2^{2k} m^{\frac{5 \cdot 2^k-3}{2^{(k+2)}-3}} n^{\frac{2^{(k+1)} - 2}{2^{(k+2)}-3}} \, ,
\end{align*}
where the last inequality holds due to $ m \leq n^2 $ and $ i \geq 0 $.
Thus, the total update time is
\begin{equation*}
\tilde O \left( k 2^{2k} q^{2k} m^{\frac{5 \cdot 2^k-3}{2^{(k+2)}-3}} n^{\frac{2^{(k+1)} - 2}{2^{(k+2)}-3}} \right) \, .
\end{equation*}
Now observe that
\begin{equation*}
m^{\frac{5 \cdot 2^k-3}{2^{(k+2)}-3}} n^{\frac{2^{(k+1)} - 2}{2^{(k+2)}-3}} = m^{\frac{5}{4} + \frac{3}{4} \cdot \frac{1}{2^{(k+2)}-3}} n^{\frac{1}{2} - \frac{1}{2} \frac{1}{2^{(k+2)}-3}} \leq m^{\frac{5}{4} + \frac{3}{4} \cdot \frac{1}{2^{(k+2)}-3}} n^{\frac{1}{2}} \, .
\end{equation*}
By our choice of $ k = \lceil \log{\log{m}} \rceil $ we have
\begin{align*}
m^{\frac{3}{4} \cdot \frac{1}{2^{(k+2)}-3}} \leq m^{\frac{1}{2^k}} \leq m^{\frac{1}{2^{\log{\log{m}}}}} = m^{\frac{1}{\log{m}}} = 2 \, .
\end{align*}
Thus, the running time of our algorithm is $ \tilde O (k q^{2k} m^{5/4} n^{1/2}) $.
With $ k = \lceil \log{\log{m}} \rceil $ and $ q = 8 $ this is $ \tilde O (m^{5/4} n^{1/2}) $.

We now balance the terms to obtain a running time of $ O (m^{4/3} n^{2/3 + o(1)}) $.
We achieve this by setting the parameters to $ k = \lfloor \sqrt{\log{n} / \log{q}} \rfloor $ and, for every $ {1 \leq i \leq k} $,
\begin{align*}
b_i &= \frac{n^{(k + 2i - 1) / (3k + 1)}}{m^{(4i - k - 3) / (6k + 2)}} \\
c_i &= 2^{k-i} n^{2i / (3k + 1)} m^{(3k + 1 - 4i) / (6k + 2)} \, .
\end{align*}
With this choice of the parameters we get (for all $ 1 \leq i \leq k $), $ b_i \geq b_{i+1} $, $ c_i \geq 2 c_{i+1} $,
\begin{align*}
\frac{b_i m n}{c_i} &\leq \frac{m^2}{b_1} = n^{4k / (3k + 1)} m^{(2k + 2) / (3k + 1)} \, , \\
\frac{c_i^2 m}{b_{i+1}} &\leq c_k^2 m \leq 2^{2k} n^{4k / (3k + 1)} m^{(2k + 2) / (3k + 1)} \, ,
\end{align*}
and
\begin{align*}
b_i c_i^2 &\leq 2^{2k} n^{(k + 6i - 1) / (3k + 1)} m^{(7k + 5 - 12i) / (6k + 2)} \\
 &= 2^{2k} n^{(k - 1) / (3k + 1)} m^{(3k + 1) / (6k + 2)} m^{(2k + 2) / (3k + 1)} n^{6i/(3k + 1)} m^{-6i / (3k + 1)} \\
 &\leq  2^{2k} n^{4k / (3k + 1)} m^{(2k + 2) / (3k + 1)} \, ,
\end{align*}
where the last inequality holds due to $ n \leq m $ and $ m \leq n^2 $.
Thus, the total update time is
\begin{equation*}
\tilde O \left( k 2^{2k} q^{2k} n^{4k / (3k + 1)} m^{(2k + 2) / (3k + 1)} \right) \, .
\end{equation*}
Now observe that
\begin{equation*}
m^{\frac{2k + 2}{3k + 1}} n^{\frac{4k}{3k + 1}} = m^{\frac{2}{3} + \frac{4}{3} \cdot \frac{1}{3k + 1}} n^{\frac{4}{3} - \frac{4}{3} \cdot\frac{1}{3k + 1}} \leq m^{\frac{2}{3} + \frac{4}{3} \cdot \frac{1}{3k + 1}} n^{\frac{4}{3}} \leq m^{\frac{2}{3} + \frac{1}{k}} n^{\frac{4}{3}} 
\end{equation*}
By our choice of $ k = \lfloor \sqrt{\log{n} / \log{q}} \rfloor $ we have $ 2^k \leq q^k \leq n^{1/k} $.
Since $ k \leq \log{n} $ and $ q = 8 $ we thus obtain a total update time of $ O (m^{2/3} n^{4/3 + O (1 / \sqrt{\log{n}})}) = O(m^{2/3} n^{4/3 + o(1)}) $.

\begin{theorem}\label{thm:st_reachability}
There is a decremental \ssssr algorithm with constant query time and total update time
\begin{equation*}
\tilde O (\min (m^{5/4} n^{1/2}, m^{2/3} n^{4/3 + o(1)})) = O (m n^{6/7 + o(1)})
\end{equation*}
that is correct with high probability against an oblivious adversary.
\end{theorem}

\subsection{Extension to Single-Source Reachability}

The algorithm above maintains reachability from a single source $ s $ to a single sink $ t $.
We can easily modify the algorithm to maintain reachability for a set of source-sink pairs $ (s_i, t_i)_{1 \leq i \leq p} $.
We simply run $ p $ instances of the algorithm, each with a different source-sink pair.
Note however that the algorithm can use the same set of hubs and centers for all instances.
Thus, the cost of maintaining the ES-trees does \emph{not} have to be multiplied by $ p $.
We therefore get a total running time of
\begin{multline*}
\tilde O \left( \sum_{1 \leq i \leq k} b_i c_i^2 + \sum_{1 \leq i \leq k}  \frac{b_i m n}{c_i} + p q^{2k} m \cdot \min \left( \frac{m}{b_1}, \frac{n^2}{b_1^2} \right) 
\right. \\
\left.
+ \sum_{1 \leq i \leq k-1} p q^{2k} c_i^2 \min \left( \frac{m}{b_{i+1}}, \frac{n^2}{b_{i+1}^2} \right) + p q^{2k} c_k^2 m \right) \, .
\end{multline*}
By setting
\begin{align*}
b_i &= \frac{m^{\frac{3 \cdot (2^k - 2^{(i-1)})}{2^{(k+2)}-3}} p^{\frac{2^{(k+1)} - 2^i}{2^{(k+2)}-3}}}{n^{\frac{2^{(k+1)} - 2^i}{2^{(k+2)}-3}}} \\
c_i &= \frac{m^{\frac{2^{(k+1)} - 3 \cdot 2^{(i-1)}}{2^{(k+2)}-3}} n^{\frac{2^i - 1}{2^{(k+2)}-3}}}{p^{\frac{2^i - 1}{2^{(k+2)} - 3}}}
\end{align*}
we get
\begin{equation*}
p c_k^2 m = p m^2 / b_1 = b_i m n / c_i = p c_i^2 m / b_{i+1} = m^{\frac{5 \cdot 2^k-3}{2^{(k+2)}-3}} n^{\frac{2^{(k+1)} - 2}{2^{(k+2)}-3}} p^{\frac{2^{(k+1)} - 1}{2^{(k+2)}-3}} \, .
\end{equation*}
\begin{align*}
\frac{b_i m n}{c_i} &\leq \frac{p m^2}{b_1} = p^{\frac{2^{(k+1)} - 1}{2^{(k+2)}-3}} m^{\frac{5 \cdot 2^k-3}{2^{(k+2)}-3}} n^{\frac{2^{(k+1)} - 2}{2^{(k+2)}-3}} \, , \\
\frac{p c_i^2 m}{b_{i+1}} &\leq p c_k^2 m \leq 2^{2k} p^{\frac{2^{(k+1)} - 1}{2^{(k+2)}-3}} m^{\frac{5 \cdot 2^k-3}{2^{(k+2)}-3}} n^{\frac{2^{(k+1)} - 2}{2^{(k+2)}-3}} \, \text{, and} \\
b_i c_i^2 &\leq 2^{2k} \frac{p^{\frac{2^{(k+1)} - 3 \cdot 2^i + 2}{2^{(k+2)}-3}} m^{\frac{7 \cdot 2^k - 9 \cdot 2^{(i-1)}}{2^{(k+2)}-3}}}{n^{\frac{2^{(k+1)} - 3 \cdot 2^i + 2}{2^{(k+2)}-3}}} \leq 2^{2k} p^{\frac{2^{(k+1)} - 1}{2^{(k+2)}-3}} m^{\frac{5 \cdot 2^k-3}{2^{(k+2)}-3}} n^{\frac{2^{(k+1)} - 2}{2^{(k+2)}-3}} \, ,
\end{align*}
and by setting
\begin{align*}
b_i &= \frac{p^{(k + 1 - i)/ (3k + 1)} n^{(k + 2i - 1) / (3k + 1)}}{m^{(4i - k - 3) / (6k + 2)}} \\
c_i &= \frac{n^{2i / (3k + 1)} m^{(3k + 1 - 4i) / (6k + 2)}}{p^{i / (3k + 1)}} \, .
\end{align*}
we get
\begin{align*}
\frac{b_i m n}{c_i} &\leq \frac{p m^2}{b_1} = p^{(k + 1) / (3k + 1)} n^{4k / (3k + 1)} m^{(2k + 2) / (3k + 1)} \, , \\
\frac{p c_i^2 m}{b_{i+1}} &\leq p c_k^2 m \leq 2^{2k} p^{(k + 1) / (3k + 1)} n^{4k / (3k + 1)} m^{(2k + 2) / (3k + 1)} \, ,
\end{align*}
and
\begin{align*}
b_i c_i^2 &= 2^{2k} p^{(k + 1 - 3i) / (3k + 1)} n^{(k + 6i - 1) / (3k + 1)} m^{(7k + 5 - 12i) / (6k + 2)} \\
 &\leq 2^{2k} p^{(k + 1) / (3k + 1)} n^{4k / (3k + 1)} m^{(2k + 2) / (3k + 1)} \, .
\end{align*}
By the same choices of $ k $ as in the $s$-$t$ reachability algorithm above we get a running times of $ \tilde O (p^{1/2} m^{5/4} n^{1/2}) $ and $ O (p^{1/3} m^{2/3} n^{4/3 + o(1)}) $, respectively, for maintaining reachability between $ p $ source-sink pairs.

\begin{corollary}
There is a decremental algorithm for maintaining reachability of $ p $ source-sink pairs with constant query time and total update time
\begin{equation*}
\tilde O (\min ( p^{1/2} m^{5/4} n^{1/2}, p^{1/3} m^{2/3} n^{4/3 + o(1)}))
\end{equation*}
that is correct with high probability against an oblivious adversary.
\end{corollary}

Using the reduction of \Cref{thm: s-t to single source} this immediately implies single-source reachability algorithm with a total update time of $ \tilde O (m^{7/6} n^{2/3}) $ (we balance the terms $ p^{1/2} m^{5/4} n^{1/2} $ and $ m n / p $ by setting $ p = n^{1/3} / m^{1/6} $) and $ O ( m^{3/4} n^{5/4 + o(1)}) $ (balance $ p^{1/3} m^{2/3} n^{4/3} $ and $ m n / p $ by setting $ p = m^{1/4} / n^{1/4} $).

\begin{corollary}
There is a decremental \ssr algorithm with constant query time and total update time
\begin{equation*}
\tilde O (\min (m^{7/6} n^{2/3}, m^{3/4} n^{5/4 + o(1)})) = O (m n^{9/10 + o(1)})
\end{equation*}
that is correct with high probability against an oblivious adversary.
\end{corollary}

Furthermore, the reduction of \Cref{thm:strongly connected component} gives a decremental algorithm for maintaining strongly connected components.

\begin{corollary}
There is a decremental \scc algorithm with constant query time and expected total update time
\begin{equation*}
\tilde O (\min (m^{7/6} n^{2/3}, m^{3/4} n^{5/4 + o(1)})) = O (m n^{9/10 + o(1)})
\end{equation*}
that is correct with high probability against an oblivious adversary.
\end{corollary}

\section{Approximate Shortest Path}\label{sec:approx_shortest_paths}

In the following we assume that $ G $ is a weighted graph with integer edge weights from $ 1 $ to $ W $ undergoing edge deletions.
We are given some parameter $ 0 < \epsilon \leq 1 $ and our goal is to maintain $ (1 + \epsilon) $-approximate shortest paths.
Similar to the reachability algorithm, we first give an algorithm for maintaining a $ (1 + \epsilon) $-approximate shortest path from a given source $ s $ to a given sink $ t $.
We then extend it to an algorithm for maintaining $ (1 + \epsilon) $-approximate shortest paths from a fixed source node to all other nodes.
We give an algorithm that maintains a $ (1 + \epsilon)^{2k + 2} $-approximate shortest path from $ s $ to $ t $.
Note that by running the whole algorithm with $ \epsilon' = \epsilon / (4k + 2) $ instead of $ \epsilon $ we obtain a $ (1 + \epsilon) $-approximate shortest path.
Throughout this section we assume that $ \epsilon \geq 1 / \log^c{n} $ and $ W \leq 2^{\log^c{n}} $ for some constant $ c $

\subsection{Preliminaries}\label{sec:SSSP_preliminaries}

The first new aspect for approximate shortest paths, compared to reachability, is that we maintain, for each $ 1 \leq i \leq k+1 $ and every pair of $i$-centers $ (x, y) $, an index $ r (x, y, i) $ such that $ \dist_G^h (x, y) \geq (1 + \epsilon)^{r (x, y, i)} $ and, whenever $ (x, y) $ is active, $ \dist_G (x, y) \leq (1 + \epsilon)^{r (x, y, i) + 2i} $.
Thus, $ r (x, y, i) $ indicates the current range of the distance from $ x $ to $ y $.
Roughly speaking, the algorithm will maintain a sequence of centers and for each pair of consecutive centers $ (x, y) $ a path from $ x $ to $ y $ of weight corresponding to the range of the distance from $ x $ to $ y $.
This will be done in a way such that whenever the algorithm cannot find such a path from $ x $ to $ y $, then the range has increased.
By charging this particular increase in the distance from $ x $ to $ y $, it can afford updating the sequence of centers.
As in the case of reachability such paths from $ x $ to $ y $ will either be found via a hub or in a path union graph.

The second new aspect is that now we maintain approximate shortest paths up to $ h_i $ \emph{hops} for certain pairs of $i$-centers, whereas in the reachability algorithm $ h_i $ was an unweighted distance.
This motivates the following modification of the path union that limits allowed paths to a fixed number of hops.
\begin{definition}
For all nodes $ x $ and $ y $ of a graph $ G $ and all integers $ h \geq 1 $ and $ D \geq 1 $, the \emph{$ h $-hop path union} $ \cP^h (x, y, D, G) $ is the set containing every node that lies on some path $ \pi $ from $ x $ to $ y $ in $ G $ that has $ | \pi | \leq h $ edges (hops) and weight $ w (\pi, G) \leq D $.
\end{definition}

Observe that $ \cP (x, y, D, G) = \cP^n (x, y, R, G)  $ by \Cref{def:path_union}.
Dealing with exact bounded-hop paths is usually computationally more expensive than dealing with unbounded paths, even in the static setting.\footnote{For example, one can compute the shortest paths among all paths with at most $ h $ edges from a source node in time $ O (m h) $ by running the first $ h $ iterations of the Bellman-Ford algorithm. This is not efficient enough for our purposes, as we would only like to spend nearly linear time.}
The bounded hop path unions will henceforth only be used in the analysis of the algorithm.
In the algorithm itself we will compute unbounded (i.e., $n$-hop) path unions in graphs with modified edge weights.
For every $ h \geq 1 $ and every $ r \geq 0 $ we define a graph $ \tilde{G}^{h, r} $ that has the same nodes and edges as $ G $ and in which we round the weight of every edge $ (u, v) $ to the next multiple of $ \epsilon (1 + \epsilon)^r / h $ by setting
\begin{equation*}
w_{\tilde{G}^{h, r}} (u, v) = \left\lceil \frac{w_G (u, v) \cdot h}{\epsilon (1 + \epsilon)^r} \right\rceil \cdot \frac{\epsilon (1 + \epsilon)^r}{h} \, .
\end{equation*}
Note that the definition of $ \tilde{G}^{h, r} $ implicitly depends on $ \epsilon $ as well.
To gain some intuition for these weight modifications, observe that for every path $ \pi $ in $ G $ consisting of $ h' $ edges we have
\begin{equation*}
w_{\tilde{G}^{h, r}} (\pi) \leq  w_G (\pi) + h' \cdot \frac{\epsilon (1 + \epsilon)^r}{h} \, .
\end{equation*}
Thus, for every path $ \pi $ of weight $ w_G (\pi) \geq (1 + \epsilon)^r $ consisting of at most $ h $ edges we have
\begin{equation*}
w_{\tilde{G}^{h, r}} (\pi) \leq (1 + \epsilon)^r + \epsilon (1 + \epsilon)^r = (1 + \epsilon)^{r+1} \, .
\end{equation*}
As additionally $ w_G (u, v) \leq w_{\tilde{G}^{h, r}} (u, v) $ for every edge $ (u, v) $ we obtain the following guarantees.
\begin{lemma}
Let $ h \geq 1 $ and $ r \geq 0 $.
For all pairs of nodes $ x $ and $ y $ we have $ \dist_G (x, y) \leq \dist_{\tilde{G}^{h, r}} (x, y) $ and if $ \dist_G^h (x, y) \geq (1 + \epsilon)^r $, then $ \dist_{\tilde{G}^{h, r}} (x, y) \leq (1 + \epsilon)^{r + 1} $.
\end{lemma}

Now observe that $h$-hop path unions can be computed ``approximately'' by computing the (unbounded) path union in $ \tilde{G}^{h, r} $, which can be done in nearly linear time (\Cref{lem:path_union_computation}).
In our case the unbounded path union in $ \tilde{G}^{h, r} $ will contain the $h$-hop path union and we will later argue that all unbounded path unions computed by our algorithm have small size.
In our analysis the following lemma will have the same purpose as \Cref{lem:path_union_in_supergraph} for the $s$-$t$ reachability algorithm.

\begin{lemma}\label{lem:computing_path_union_in_supergraph}
Let $ (x, y) $ be a pair of nodes, let $ h \geq 1 $, $ r \geq 0 $, and $ \alpha \geq 1 $, and let $ \cQ \subseteq V $ be a set of nodes such that $ \cP^h (x, y, \alpha (1 + \epsilon)^r, G) \subseteq \cQ $.
Then $ \cP^h (x, y, \alpha (1 + \epsilon)^r, G) \subseteq \cP (x, y, \alpha (1 + \epsilon)^{r+1}, \tilde{G}^{h, r} [\cQ]) $.
\end{lemma}

\begin{proof}
Let $ v \in \cP^h (x, y, \alpha (1 + \epsilon)^r, G) $, i.e., there is a path $ \pi $ from $ x $ to $ y $ in $ G $ with at most $ h $ edges and weight at most $ \alpha (1 + \epsilon)^r $ containing $ v $.
The weight of $ \pi $ in $ \tilde{G}^{h, r} $ is
\begin{equation*}
w_{\tilde{G}^{h, r}} (\pi) \leq w_G (u, v) + h \cdot \frac{\epsilon (1 + \epsilon)^r}{h} \leq \alpha (1 + \epsilon)^r + \epsilon (1 + \epsilon)^r \leq \alpha (1 + \epsilon)^{r+1} \, .
\end{equation*}
By our assumption all nodes of $ \pi $ are contained in $ \cQ $ and thus $ \pi $ is a path in $ \tilde{G}^{h, r} [\cQ] $ of weight at most $ \alpha (1 + \epsilon)^{r+1} $.
Thus, all nodes of $ \pi $, including $ v $, are contained in $ \cP (x, y, \alpha (1 + \epsilon)^{r+1}, \tilde{G}^{h, r} [\cQ]) $.
\end{proof}

The advantage of the graph $ \tilde{G}^{h, r} $ is that its edge weights are multiples of $ \epsilon (1 + \epsilon)^r / h $.
Thus, after \emph{scaling down} the edge weights by a factor of $ h / (\epsilon (1 + \epsilon)^r) $ we still have integer weights.
This observation can be used to speed up the pseudopolynomial algorithm of Even and Shiloach at the cost of a $ (1 + \epsilon) $-approximation~\cite{Bernstein09,Madry10,Bernstein13}.

\subsection{Algorithm Description}

Our approximate $s$-$t$ shortest path algorithm has a parameter $ k \geq 1 $ and for each $ 1 \leq i \leq k $ parameters $ b_i \leq n $ and $ c_i \leq n $.
We also set $ b_{k+1} = 1 $, $ c_0 = n $, $ c_{k+1} = 1 $, and $ h_i = n / c_i $ for all $ 0 \leq i \leq k + 1 $.
Our algorithm determines sets of nodes $ B_1 \supseteq B_2 \supseteq \dots \supseteq B_k $ and $ C_0 \supseteq C_1 \supseteq \dots \supseteq C_{k+1} $ by random sampling as described in \Cref{sec:st_reach_algo_description} for the $s$-$t$ reachability algorithm, i.e., $i$-hubs are obtained by sampling nodes with probability $ a b_i \ln{(\Delta (G))} / n $ and edges with probability $ a b_i \ln{(\Delta (G))} / m $ and $i$-centers are obtained by sampling nodes with probability $ a c_i \ln{(\Delta (G))} / n $.
Note that, as we may assume $ \Delta (G) \leq m \log_{1 + \epsilon} W $ by \Cref{lem:few_updates} and we further assume that $ \epsilon \geq 1 / \log^c{n} $ and $ W \leq 2^{\log^c{n}} $ for some constant $ c $, the number of $i$-hubs is $ \tilde O (b_i) $ whp and the number of $i$-centers is $ \tilde O (c_i) $ whp by \Cref{lem:hitting_set_argument}.
We can make the algorithm fail with small probability if these sets are larger.
For each $ 1 \leq i \leq k $ we order the set of $i$-hubs in an arbitrary way such that $ B_i = \{ b_1, b_2, \ldots, b_{| B_i |} \} $.
Our algorithm uses the following variables and data structures:
\begin{itemize}
\item An estimate $ \delta (s, t) $ of the distance from $ s $ to $ t $ in $ G $
\item For every pair of $i$-centers $ (x, y) $ (with $ 1 \leq i \leq k $) an index $ r (x, y, i) $ such that $ 0 \leq r (x, y, i) \leq \log_{1 + \epsilon} (nW) $ for the current range of the distance from $ x $ to $ y $.
\item For every $i$-hub $ z $ (with $ 1 \leq i \leq k $) and every $ 0 \leq r \leq \lfloor \log_{1 + \epsilon} (nW) \rfloor $ an ES-tree up to depth $ (1 + \epsilon)^{r + i + 1} $ in $ \tilde{G}^{8 h_i / \epsilon, r} $ (\emph{outgoing} tree of $ z $) and an ES-tree up to depth $ (1 + \epsilon)^{r + i + 1} $ in the reverse graph of $ \tilde{G}^{8 h_i / \epsilon, r} $ (\emph{incoming} tree of~$ z $).
\item For every pair of $i$-centers $ (x, y) $ (with $ 1 \leq i \leq k $) an index $ l (x, y, i) $ such that $ 1 \leq l (x, y, i) \leq | B_i | + 1 $ that either gives the index $ l (x, y, i) $ of the $i$-hub $ b_{l (x, y, i)} \in B_i $ that links $ x $ to $ y $ for the current value of $ r (x, y, i) $ or, if no such $i$-hub exists, is set to $ l (x, y, i) = | B_i | + 1 $.
\item For every pair of $i$-centers $ (x, y) $ (with $ 1 \leq i \leq k $) a set of nodes $ \cQ (x, y, i) $, which is initially empty.
\item For every pair of $i$-centers $ (x, y) $ (with $ 0 \leq i \leq k $) a list of pairs of $(i+1)$-centers called $(i+1)$-parents of $ (x, y) $. 
\item For every pair of $i$-centers $ (x, y) $ (with $ 1 \leq i \leq k+1 $) a list of pairs of $(i-1)$-centers called $(i-1)$-children of $ (x, y) $.
\item For every $ 1 \leq i \leq k $, a list $ A_i $ of $i$-centers called \emph{active $i$-centers}.
\end{itemize}

Besides the generalizations for weighted graphs introduced in \Cref{sec:SSSP_preliminaries}, our approximate $s$-$t$ shortest path algorithm also deviates from the $s$-$t$ reachability algorithm in the way it maintains the hub links.
The main difference to before is that now we maintain the hub links only between pairs of active centers.
In the new algorithm we say, for $ 1 \leq i \leq k $, that an $i$-hub $ z \in B_i $ \emph{links} an $i$-center $ x $ to an $i$-center $ y $ if
\begin{equation*}
\dist_{\tilde{G}^{8 h_i / \epsilon, r (x, y, i)}} (x, z) + \dist_{\tilde{G}^{8 h_i / \epsilon, r (x, y, i)}} (z, y) \leq (1 + \epsilon)^{r (x, y, i) + 2i+1} \, .
\end{equation*}
The algorithm can check whether $ z $ links $ x $ to $ y $ by looking up $ \dist_{\tilde{G}^{8 h_i / \epsilon, r (x, y, i)}} (x, z) $ and $ \dist_{\tilde{G}^{8 h_i / \epsilon, r (x, y, i)}} (z, y) $ in the incoming and outgoing ES-tree of $ z $, respectively.
For every $ 1 \leq i \leq k $, and every pair of active $i$-centers $ x $ and $ y $, the algorithm uses the index $ l (x, y, i) $ to maintain an $i$-hub $ b_{l (x, y, i)} $ that links $ x $ to $ y $.
If the current $i$-hub $ b_{l (x, y, i)} $ does not link $ x $ to $ y $ anymore, the algorithm increases $ l (x, y, i) $ until either such an $i$-hub is found, or $ l (x, y, i) = | B_i | + 1 $.
As soon as $ l (x, y, i) = | B_i | + 1 $, the algorithm computes the path union between $ x $ and $ y $ and maintains an approximate shortest path from $ x $ to $ y $ in the subgraph induced by the path union.
If the algorithm does not find such a path anymore it increases the value of $ r (x, y, i) $ and sets $ l (x, y, i) $ to $ 1 $ again, i.e., making the first $i$-hub the candidate for linking $ x $ to $ y $ for the new value of $ r (x, y, i) $.
The generalization of the $s$-$t$ reachability algorithm is now straightforward and the pseudocode can be found in \Cref{alg:s_t_shortest_path}.
Before the first update the algorithms executes Procedure \Initialize.
After each increase of the weight of an edge $ (u, v) $ or after its deletion the algorithm calls \Update{$u$, $v$}.

\begin{algorithm}
\thisfloatpagestyle{empty} 
\caption{Algorithm for decremental approximate $s$-$t$ shortest path}
\label{alg:s_t_shortest_path}

\Procedure{\Refresh{$x$, $y$, $i$}}{
	\RemoveChildren{$x$, $y$, $i$}\;
	\eIf{$ i = k+1 $}{
		Compute shortest path $ \pi $ from $ s $ to $ t $ in $ G $ and set $ \delta (s, t) \gets (1 + \epsilon)^{2k+1} \dist_G (s, t) $\;
	}{
		Compute shortest path $ \pi $ from $ x $ to $ y $ in $ \tilde{G}^{h_i, r (x, y, i)} [\cQ (x, y, i) (x, y)] $\;
		\If(\tcp*[f]{Range has increased}){$ \dist_{\tilde{G}^{h_i, r (x, y, i)} [\cQ (x, y, i)} (x, y)] > (1 + \epsilon)^{r (x, y, i) + 2} $}{ \label{lin:check_for_radius_increase}
			$ r (x, y, i) \gets r (x, y, i) + 1 $\; \label{lin:set_radius}
			$ l (x, y, i) \gets 1 $\;
			\lForEach{$(i+1)$-parent $ (x', y') $ of $ (x, y) $}{
				\Refresh{$x'$, $y'$, $i+1$}
			}
			\KwBreak\;
		}
	}
	Determine $ (i-1) $-centers $ v_1, \ldots, v_l $ in order of appearance on $ \pi $\;
	\ForEach{$ 1 \leq j \leq l-1 $}{
		Add $ (v_j, v_{j+1}) $ to $ A_{i-1} $\; 
		Make $ (v_j, v_{j+1}) $ an $(i-1)$-child of $ (x, y) $ and $ (x, y) $ an $i$-parent of $ (v_j, v_{j+1}) $\;
		\If{$ i \geq 1 $}{
			\UpdateHubLinks{$v_j$, $v_{j+1}$, $i-1$, $r (v_j, v_{j+1}, i-1)$}\;
			\If{$ l_i (v_j, v_{j+1}) = |B_i| + 1 $}{
				$ \cQ (v_j, v_{j+1}, i-1) \gets \cP (v_j, v_{j+1}, (1 + \epsilon)^{r (v_j, v_{j+1}, i-1) + 2i-1}, \tilde{G}^{8 h_{i-1} / \epsilon, r (v_j, v_{j+1}, i-1)} [\cQ (x, y, i)]) $\; \label{lin:path_union_refresh}
				\Refresh{$v_j$, $v_{j+1}$, $i-1$}\;
			}
		}
	}
}

\Procedure{\UpdateHubLinks{$x$, $y$, $i$, $r$}}{
	\While{$ l (x, y, i) \leq |B_i| $ \KwAnd $ \dist_{\tilde{G}^{8 h_i / \epsilon, r (x, y, i)}} (x, b_{l (x, y, i)}) + \dist_{\tilde{G}^{8 h_i / \epsilon, r (x, y, i)}} (b_{l (x, y, i)}, y) > (1 + \epsilon)^{r (x, y, i) + 2i+1} $}{
		$ l (x, y, i) \gets l (x, y, i) + 1 $\;
		\If{$ l (x, y, i) = |B_i| + 1 $}{
			Let $ (x', y') $ be any $(i+1)$-parent of $ (x, y) $\;
			$ \cQ (x, y, i)\gets \cP (x, y, (1 + \epsilon)^{r (x, y, i) + 2i+1}, \tilde{G}^{8 h_i / \epsilon, r (x, y, i)} [\cQ (x', y', i+1)]) $\;\label{lin:first_path_union_computation}
		}
	}
}

\Procedure{\RemoveChildren{$x$, $y$, $i$}}{
	\ForEach{$(i-1)$-child $ (x', y') $ of $ (x, y) $}{
		Remove $ (x', y') $ from $(i-1)$-children of $ (x, y) $ and $ (x, y) $ from $i$-parents of $ (x', y') $\;
		\If{$ (x', y') $ has no $i$-parents anymore}{
			Remove $ (x', y') $ from $ A_{i-1} $\;
			\lIf{$ i \geq 1 $}{\RemoveChildren{$x'$, $y'$, $i-1$}}
		}
	}
}

\Procedure{\Initialize{}}{
	\lForEach{$ 1 \leq i \leq k $ and all $i$-centers $ (x, y) $}{
		$ r (x, y, i) \gets 0 $,
		$ l (x, y, i) \gets 1 $,
		$ \cQ (x, y, i) \gets \emptyset $
	}
	\lForEach{$ 1 \leq i \leq k $}{
		$ A_i \gets \emptyset $
	}
	\Refresh{$s$, $t$, $k+1$}\; \label{lin:refresh_in_initialization}
}

\Procedure(\tcp*[f]{Deletion or weight increase involving $ (u, v) $}){\Update{$u$, $v$}}{
	\lForEach{$1$-parent $ (x, y) $ of $ (u, v) $}{
		\Refresh{$x$, $y$, $1$}
	}
	\ForEach{$ 1 \leq i \leq k $ and $ (x, y) \in A_i $}{
		\UpdateHubLinks{$x$, $y$, $i$, $r (x, y, i)$}\;
		\lIf{$ l (x, y, i) = |B_i| + 1 $}{
			\Refresh{$x$, $y$, $i$}
		}
	}
}
\end{algorithm}

\subsection{Correctness}

To establish the correctness of the algorithm we will show that
\begin{equation*}
\dist_G (s, t) \leq \delta (s, t) \leq (1 + \epsilon)^{2k + 1} \dist_G (s, t) \, .
\end{equation*}
By running the whole algorithm with $ \epsilon' = \epsilon / (4k + 2) $ (instead of $ \epsilon) $, we then obtain a $ (1 + \epsilon) $-approximation (see \Cref{lem:exponential_inequality} below).

We will achieve this by showing that with the value of $ r (x, y, i) $ of a pair of $i$-centers $ (x, y) $ the algorithm keeps track of the range of the distance from $ x $ to $ y $.
In particular, we will show that $ \dist_G^{h_i} (x, y) \geq (1 + \epsilon)^{r (x, y, i)} $ for every pair of $i$-centers $ (x, y) $ and $ \dist_G (x, y) \leq (1 + \epsilon)^{r (x, y, i) + 2i + 1} $ for every pair of active $i$-centers $ (x, y) \in A_i $.
Both inequalities are obtained by a  ``bottom up'' analysis of the algorithm.
In order to prove the lower bound on $ \dist_G^{h_i} (x, y) $ we have to show that the subgraph $ \cQ (x, y, i) $ computed by the algorithm indeed contains the corresponding path union.
In fact those two invariants rely on each other and we will prove them mutually.

\begin{lemma}\label{lem:lower_bound_on_hop_distance}
Algorithm~\ref{alg:s_t_shortest_path} correctly maintains the following invariants whp:
\begin{itemize}
\item[(I1)] For every $ 1 \leq i \leq k $ and every pair of $i$-centers $ (x, y) $,
\begin{equation*}
\dist_G^{h_i} (x, y) \geq (1 + \epsilon)^{r (x, y, i)} \, .
\end{equation*}
\item[(I2)] For every $ 1 \leq i \leq k $ and every pair of active $i$-centers $ (x, y) \in A_i $,
\begin{equation*}
\cP^{8 h_i / \epsilon} (x, y, (1 + \epsilon)^{r (x, y, i) + 2i}, G) \subseteq \cQ (x, y, i) \, .
\end{equation*}
\end{itemize}
\end{lemma}

\begin{proof}
Both invariants trivially hold directly before the first refresh in \Cref{lin:refresh_in_initialization} of the initialization is called because we set $ r (x, y, i) = 0 $ for all $ 1 \leq i \leq k $ and no pairs of centers are active yet.

We first give a proof of Invariant (I1).
As distances are non-decreasing in $ G $ we only have to argue that the invariant still holds after each time the algorithm changes the value of $ r (x, y, i) $.
The only place where $ r (x, y, i) $ is changed is in Line~\ref{lin:set_radius} (where it is increased by $ 1 $) and if this line is reached the condition
\begin{equation*}
\dist_{\tilde{G}^{h_i, r (x, y, i)} [\cQ (x, y, i)]} (x, y) > (1 + \epsilon)^{r (x, y, i) + 2}
\end{equation*}
holds.

Suppose that in $ G $ there is a path $ \pi $ from $ x $ to $ y $ with at most $ h_i $ edges and weight less than $ (1 + \epsilon)^{r (x, y, i) + 1} $.
The weight of $ \pi $ in $ \tilde{G}^{h_i, r (x, y, i)} $ is
\begin{align*}
w_{\tilde{G}^{h_i, r (x, y, i)}} (\pi) &\leq w_G (\pi) + h_i \cdot \frac{\epsilon (1 + \epsilon)^{r (x, y, i)}}{h_i} \\
 &\leq (1 + \epsilon)^{r (x, y, i) + 1} + \epsilon (1 + \epsilon)^{r (x, y, i)} \\
 &\leq (1 + \epsilon)^{r (x, y, i) + 2} \\ 
 &\leq (1 + \epsilon)^{r (x, y, i) + 2i} \, .
\end{align*}
By Invariant (I2) all nodes of $ \pi $ are contained in $ \cQ (x, y, i) $.
Therefore
\begin{equation*}
\dist_{\tilde{G}^{h_i, r (x, y, i)} [\cQ (x, y, i)]} (x, y) \leq w_{\tilde{G}^{h_i, r (x, y, i)}} (\pi) \leq (1 + \epsilon)^{r (x, y, i) + 2}
\end{equation*}
which contradicts the assumption that the algorithm has reached Line~\ref{lin:set_radius}.
Thus, the path $ \pi $ does not exist and it follows that $ \dist_G^{h_i} (x, y) \geq (1 + \epsilon)^{r (x, y, i) + 1} $ as desired.
Therefore Invariant (I1) still holds when the algorithm increases the value of $ r (x, y, i) $ by~$ 1 $.

We now give a proof of Invariant (I2).
The proof is by induction on $ i $.
There are two places in the algorithm where the value of $ \cQ (x, y, i) $ is changed: Line~\ref{lin:path_union_refresh} and Line~\ref{lin:first_path_union_computation}.
Note that $ r (x, y, i) $ is non-decreasing and $ l (x, y, i) $ decreases only if $ r (x, y, i) $ increases.
Therefore, for a fixed value of $ r (x, y, i) $, the algorithm will ``recompute'' $ \cQ (x, y, i) $ in Line~\ref{lin:path_union_refresh} only if it has ``initialized'' it in Line~\ref{lin:first_path_union_computation} before.

Consider first the case that $ \cQ (x, y, i) $ is ``initialized'' (Line~\ref{lin:first_path_union_computation}) and let $ (x', y') $ be the $(i+1)$-parent of $ (x, y) $ used in this computation.
If $ i \leq k-1 $, then by the induction hypothesis we know that
\begin{equation}
\cP^{8 h_{i+1} / \epsilon} (x', y', (1 + \epsilon)^{r (x', y', i+1) + 2(i+1)}, G) \subseteq \cQ (x', y', i+1) \, . \label{eq:concrete_induction_hypothesis}
\end{equation}
If $ i = k $, then this inclusion also holds because in that case $ x' = s $ and $ y' = t $ (as $ s $ and $ t $ are the only $ (k+1) $-centers) and we have set $ \cQ (s, t, k+1) = V $.
We will show below that, at the last time the algorithm has called \Refresh{$x'$, $y'$, $i+1$} (where it sets $ (x, y) $ as an $i$-child of $ (x', y') $ and $ (x', y') $ as an $(i+1)$-parent of $ (x, y) $), we have
\begin{equation*}
\cP^{8 h_i / \epsilon} (x, y, (1 + \epsilon)^{r (x, y, i) + 2i}, G) \subseteq \cP^{8 h_{i+1} / \epsilon} (x', y', (1 + \epsilon)^{r (x', y', i+1) + 2(i+1)}, G)
\end{equation*}
and thus (together with~\eqref{eq:concrete_induction_hypothesis})
\begin{equation}
\cP^{8 h_i / \epsilon} (x, y, (1 + \epsilon)^{r (x, y, i) + 2i}, G) \subseteq \cQ (x', y', i+1) \, . \label{eq:bounded_hop_path_union_contained_in_Q}
\end{equation}
As distances in $ G $ are non-decreasing, \eqref{eq:bounded_hop_path_union_contained_in_Q} will still hold at the time the algorithm sets
\begin{equation}
\cQ (x, y, i) = \cP (x, y, (1 + \epsilon)^{r (x, y, i) + 2i+1}, \tilde{G}^{8 h_i / \epsilon, r (x, y, i)} [\cQ (x', y', i+1)]) \, .
\end{equation}
according to Line~\ref{lin:first_path_union_computation}.
By \Cref{lem:computing_path_union_in_supergraph}, \eqref{eq:bounded_hop_path_union_contained_in_Q} will imply that
\begin{multline*}
\cP^{8 h_i / \epsilon} (x, y, (1 + \epsilon)^{r (x, y, i) + 2i}, G) \\
\subseteq \cP (x, y, (1 + \epsilon)^{r (x, y, i) + 2i+1}, \tilde{G}^{8 h_i / \epsilon, r (x, y, i)} [\cQ (x', y', i+1)]) = \cQ (x, y, i)
\end{multline*}
as desired.
Note again that, since distances in $ G $ are nondecreasing, the inclusion $ \cP^{8 h_i / \epsilon} (x, y, (1 + \epsilon)^{r (x, y, i) + 2i}, G) \subseteq \cQ (x, y, i) $ will then continue to hold until $ \cQ (x, y, i) $ is ``recomputed'' (Line~\ref{lin:path_union_refresh}) for the first time.
Whenever this happens, it will set new the value of $ \cQ (x, y, i) $ equal to $ \cP (x, y, (1 + \epsilon)^{r (x, y, i) + 2i+1}, \tilde{G}^{8 h_i / \epsilon, r (x, y, i)} [\cQ (x, y, i)]) $.
Thus, by \Cref{lem:computing_path_union_in_supergraph} we will again have $ \cP^{8 h_i / \epsilon} (x, y, (1 + \epsilon)^{r (x, y, i) + 2i}, G) \subseteq \cQ (x, y, i) $ as demanded by Invariant~(I2).

It remains to show that at the last time the algorithm has called \Refresh{$x'$, $y'$, $i+1$} and sets $ (x', y') $ as an $(i+1)$-parent of $ (x, y) $, we have
\begin{equation*}
\cP^{8 h_i / \epsilon} (x, y, (1 + \epsilon)^{r (x, y, i) + 2i}, G) \subseteq \cP^{8 h_{i+1} / \epsilon} (x', y', (1 + \epsilon)^{r (x', y', i+1) + 2(i+1)}, G) \, .
\end{equation*}
Let $ \pi' $ denote the shortest path from $ x' $ to $ y' $ in $ \tilde{G}^{h_{i+1}, r (x', y', i+1)} [\cQ (x', y', i+1)] $ computed at the beginning of the last execution of \Refresh{$x'$, $y'$, $i+1$}.
To enhance the readability of this part of the proof we use the abbreviation $ H = \tilde{G}^{h_{i+1}, r (x', y', i+1)} [\cQ (x', y', i+1)] $.
By the if-condition in \Cref{lin:check_for_radius_increase} we know that $ w (\pi', H) = \dist_H (x', y') \leq (1 + \epsilon)^{r (x', y', i+1)+2} $.

Consider some node $ v \in \cP^{8 h_i / \epsilon} (x, y, (1 + \epsilon)^{r (x, y, i) + 2i}, G) $ which means that $ v $ lies on a path $ \pi $ from $ x $ to $ y $ in $ G $ that has at most $ 8 h_i / \epsilon $ edges and weight at most $ (1 + \epsilon)^{r (x, y, i) + 2i} $.
Remember that $ x $ and $ y $ are consecutive $i$-centers on $ \pi' $.
Let $ \pi_1' $ and $ \pi_2' $ denote the subpaths of $ \pi' $ from $ x' $ to $ x $ and $ y $ to $ y' $, respectively.
The concatenation $ \pi'' = \pi_1' \circ \pi \circ \pi_2' $ is a path from $ x $ to $ y $ in G.
We will show that $ \pi'' $ has at most $ 8 h_{i+1} / \epsilon $ edges and weight at most $ (1 + \epsilon)^{r (x', y', i+1) + 2(i+1)} $.
This then proves that all nodes on $ \pi $ are contained in $ \cP^{8 h_{i+1} / \epsilon} (x, y, (1 + \epsilon)^{r (x', y', i+1) + 2(i+1)}, G) $.

As $ 2 h_i \leq h_{i+1} $, the number of edges of $ \pi $ is at most $ 8 h_i / \epsilon \leq 4 h_{i+1} / \epsilon $.
Remember that the edge weights of $ H $ are multiples of $ \epsilon (1 + \epsilon)^{r (x', y', i+1)} / h_{i+1} $.
Thus, each edge of $ H $ has weight at least $ \epsilon (1 + \epsilon)^{r (x', y', i+1)} / h_{i+1} $.
As the weight of $ \pi' $ in $ H $ is at most $ (1 + \epsilon)^{r (x', y', i+1)+2} $, the number of edges of $ \pi' $ is at most
\begin{equation*}
\frac{(1 + \epsilon)^{r (x', y', i+1)+2} h_{i+1}}{\epsilon (1 + \epsilon)^{r (x', y', i+1)}} = \frac{(1 + \epsilon)^2 h_{i+1}}{\epsilon} \leq \frac{4}{\epsilon} \, .
\end{equation*}
It follows that the number of edges of $ \pi'' $ is at most $ 4 h_{i+1} / \epsilon + 4 h_{i+1} / \epsilon = 8 h_{i+1} / \epsilon $.

By Invariant~(I1) we have $ \dist_G^{h_i} (x, y) \geq (1 + \epsilon)^{r (x, y, i)} $ and thus 
\begin{equation*}
w_G (\pi) \leq (1 + \epsilon)^{r (x, y, i) + 2i} \leq (1 + \epsilon)^{2i} \dist_G^{h_i} (x, y) \, .
\end{equation*}
By the initial random sampling of $i$-centers, every shortest path consisting of $ h_i - 1 $ edges contains an $i$-center whp (\Cref{lem:hitting_set_argument}).
Thus, the subpath of $ \pi' $ from $ x $ to $ y $ has at most $ h_i $ edges.
Since $ \pi' $ is a shortest path in $ H $ we have $ \dist_H (x, y) = \dist_H^{h_i} (x, y) $.
It follows that
\begin{align*}
w_G (\pi) \leq (1 + \epsilon)^{2i} \dist_G^{h_i} (x, y) \leq (1 + \epsilon)^{2i} \dist_{G [\cQ (x', y', i+1)]}^{h_i} (x, y) &\leq (1 + \epsilon)^{2i} \dist_H^{h_i} (x, y) \\
 &= (1 + \epsilon)^{2i} \dist_H (x, y) \, .
\end{align*}
Furthermore we have $ \dist_H (x', y') = \dist_H (x', x) + \dist_H (x, y) + \dist_H $ (and in particular $ \dist_H (x, y) \leq \dist_H (x', y') $).
As $ w_G (\pi_1') \leq \dist_H (x', x) $ and $ w_G (\pi_2') \leq  \dist_H (y, y') $ the weight of the path $ \pi'' $ in $ G $ can be bounded as follows:
\begin{align*}
w_G (\pi'') &= w_G (\pi_1') + w_G (\pi_2') + w_G (\pi) \\
 &\leq \dist_H (x', x) + \dist_H (y, y') + w_G (\pi) \\
 &\leq \dist_H (x', y') - \dist_H (x, y) + w_G (\pi) \\
 &\leq \dist_H (x', y') - \dist_H (x, y) + (1 + \epsilon)^{2i} \dist_H (x, y) \\
 &= \dist_H (x', y') + ((1 + \epsilon)^{2i} - 1) \dist_H (x, y) \\
 &\leq \dist_H (x', y') + ((1 + \epsilon)^{2i} - 1) \dist_H (x', y') \\
 &= (1 + \epsilon)^{2i} \dist_H (x', y') \\
 &\leq (1 + \epsilon)^{2i} (1 + \epsilon)^{r (x', y', i+1) + 2} \\
 &= (1 + \epsilon)^{r (x', y', i+1) + 2 (i + 1)} \, . \qedhere
\end{align*}
\end{proof}

\begin{lemma}\label{lem:upper_bound_distances}
For every $ 1 \leq i \leq k $ and every active pair of $i$-centers $ (x, y) \in A_i $ we have $ \dist_G (x, y) \leq (1 + \epsilon)^{r (x, y, i) + 2i+1} $ whp after the algorithm has finished its updates.
\end{lemma}

\begin{proof}
The proof is by induction on $ i $.
If $ l (x, y, i) \leq | B_i | $, then there is some $i$-hub $ z $ that links $ x $ to~$ y $.
By the triangle inequality we then have
\begin{align*}
\dist_G (x, y) &\leq \dist_G (x, z) + \dist_G (z, y) \\
 &\leq \dist_{\tilde{G}^{8 h_i / \epsilon, r (x, y, i)}} (x, z) + \dist_{\tilde{G}^{8 h_i / \epsilon, r (x, y, i)}} (z, y) \\
 &\leq (1+\epsilon)^{r (x, y, i) + 2i + 1} \, .
\end{align*}
If $ l (x, y, i) = | B_i | + 1 $, then consider the last time the algorithm has called \Refresh{$x$, $y$, $i$}.
Let $ G' $ and $ \cQ' (x, y, i) $ denote the versions of $ G $ and $ \cQ (x, y, i) $ at the beginning of the refresh operation and let $ \pi $ be the shortest path from $ x $ to $ y $ in $ (\tilde{G}')^{h_i, r (x, y, i)} [\cQ' (x, y, i)] $ computed by the algorithm.
To enhance readability we set $ H = (\tilde{G}')^{h_i, r (x, y, i)} [\cQ' (x, y, i)] $  in this proof.
The algorithm ensures that $ \dist_{H} (x, y) \leq (1 + \epsilon)^{r (x, y, i) + 2} $ as otherwise, by the if-condition in Line~\ref{lin:check_for_radius_increase}, $ (x, y) $ would have either become inactive (i.e., removed from $ A_i $) or the algorithm would have called \Refresh{$x$, $y$, $i$} again.
Let $ v_1, \ldots, v_l $ denote the $(i-1)$-centers in order of appearance on $ \pi $, i.e., $ \dist_{H} (x, y) = \sum_{1 \leq j \leq l-1} \dist_{H} (v_j, v_{j+1}) $.

If $ i = 1 $, then remember that the set of $0$-centers is equal to $ V $.
Note that no edge $ (u, v) $ of $ \pi $ has been deleted from $ G $ since the last time the algorithm has called \Refresh{$x$, $y$, $i$} because $ (x, y) $ is a $1$-parent of $ (u, v) $.
Therefore all edges of $ \pi $ still exist in $ G $.
Since $ r (x, y, i) $ is non-decreasing this means that $ \dist_G (x, y) \leq \dist_H (x, y) \leq (1 + \epsilon)^{r (x, y, i) + 2} $ as desired. 

If $ i \geq 2 $, then let $ 1 \leq j \leq l - 1 $.
First of all, note that the value of $ r (v_j, v_{j+1}, i) $ has not changed since the last time the algorithm has called \Refresh{$x$, $y$, $i$} because otherwise after Line~\ref{lin:set_radius}, the algorithm would call \Refresh{$x$, $y$, $i$} as $ (x, y) $ is an $i$-parent of $ (v_j, v_{j+1}) $.
Therefore, by Invariant~(I1) of \Cref{lem:lower_bound_on_hop_distance}, we have $ \dist_{G'}^{h_{i-1}} (v_j, v_{j+1}) \geq (1 + \epsilon)^{r (v_j, v_{j+1}, i-1)} $.
Note that $ H $ has the same nodes and edges as $ G' $ and the weight of each edge in $ H $ is at least its weight in $ G' $.
Therefore $ \dist_{H}^{h_{i-1}} (v_j, v_{j+1}) \geq \dist_{G'}^{h_{i-1}} (v_j, v_{j+1}) $ and thus $ \dist_{H}^{h_{i-1}} (v_j, v_{j+1}) \geq (1 + \epsilon)^{r (v_j, v_{j+1}, i-1)} $.
By the initial random sampling of $(i-1)$-centers, every shortest path consisting of $ h_{i-1} - 1 $ edges contains an $(i-1)$-center whp (\Cref{lem:hitting_set_argument}).
Thus, the subpath of $ \pi $ from $ v_j $ to $ v_{j+1} $ has at most $ h_{i-1} $ edges which means that $ \dist_{H} (v_j, v_{j+1}) = \dist_{H}^{h_{i-1}} (v_j, v_{j+1}) \geq (1 + \epsilon)^{r (v_j, v_{j+1}, i-1)} $.

By the induction hypothesis we have $ \dist_G (v_j, v_{j+1}) \leq (1 + \epsilon)^{r (v_j, v_{j+1}, i-1) + 2 (i-1) + 1} $ in the current version of $ G $, which by the argument above implies that $ \dist_G (v_j, v_{j+1}) \leq (1 + \epsilon)^{2(i-1) + 1} \dist_{H} (v_j, v_{j+1}) $.
Thus, by the triangle inequality we get 
\begin{align*}
\dist_G (x, y) &\leq \sum_{1 \leq j \leq l-1} \dist_G (v_j, v_{j+1}) \\
 &\leq \sum_{1 \leq j \leq l-1} (1 + \epsilon)^{2 (i-1) + 1} \dist_{H} (v_j, v_{j+1}) \\
 &= (1 + \epsilon)^{2i - 1} \dist_{H} (x, y)  \\
 &\leq (1 + \epsilon)^{2i - 1} (1 + \epsilon)^{r (x, y, i) + 2} \\
 &= (1 + \epsilon)^{2i + 1}
\end{align*}
as desired.
\end{proof}

\begin{lemma}\label{lem:approximation_guarantee_s_t}
After the algorithm has finished its updates we have $ \dist_G (s, t) \leq \delta (s, t) \leq (1 + \epsilon)^{k+1} \dist_G (s, t) $ whp.
\end{lemma}

\begin{proof}
Consider the last time the algorithm has called \Refresh{$s$, $t$, $k+1$}.
Let $ G' $ denote the versions of $ G $ at the beginning of the refresh operation and let $ \pi $ be the shortest path from $ x $ to $ y $ in $ G $ computed by the algorithm.
Note that $ \delta (s, t) = (1 + \epsilon)^{2k+1} \dist_{G'} (s, t) $.
As distances are non-decreasing under deletions in $ G $ we trivially have $ \delta (s, t) \leq (1 + \epsilon)^{k+1} \dist_G (s, t) $.

Let $ v_1, \ldots, v_l $ denote the $k$-centers in order of appearance on $ \pi $, which means that $ \dist_{G'} (x, y) = \sum_{1 \leq j \leq l-1} \dist_{G'} (v_j, v_{j+1}) $, and let $ 1 \leq j \leq l - 1 $.
The value of $ r (v_j, v_{j+1}, k) $ has not changed since the last time the algorithm has called \Refresh{$x$, $y$, $i$} because otherwise after Line~\ref{lin:set_radius}, the algorithm would call \Refresh{$s$, $t$, $k$} as $ (s, t) $ is a $(k+1)$-parent of $ (v_j, v_{j+1}) $.
By \Cref{lem:lower_bound_on_hop_distance}, we have $ \dist_{G'}^{h_k} (v_j, v_{j+1}) \geq (1 + \epsilon)^{r (v_j, v_{j+1}, k)} $ and by \Cref{lem:upper_bound_distances} we have $ \dist_G (v_j, v_{j+1}) \leq (1 + \epsilon)^{r (v_j, v_{j+1}, k) + 2k + 1} $.
It follows that $ \dist_G (v_j, v_{j+1}) \leq (1 + \epsilon)^{2k + 1} \dist_{G'} (v_j, v_{j+1}) $.

By the initial random sampling of $k$-centers, every shortest path consisting of $ h_k - 1 $ edges contains a $k$-center whp (\Cref{lem:hitting_set_argument}).
Thus, the subpath of $ \pi $ from $ v_j $ to $ v_{j+1} $ has at most $ h_k $ edges which means that $ \dist_{G'} (v_j, v_{j+1}) = \dist_{G'}^{h_k} (v_j, v_{j+1}) $.
Thus, by the triangle inequality we get 
\begin{multline*}
\dist_G (s, t) \leq \sum_{1 \leq j \leq l-1} \dist_G (v_j, v_{j+1}) \leq \sum_{1 \leq j \leq l-1} (1 + \epsilon)^{2 k + 1} \dist_{G'} (v_j, v_{j+1}) \\
= (1 + \epsilon)^{2 k + 1} \dist_{G'} (s, t) = \delta (s, t) \, . \qedhere
\end{multline*}
\end{proof}

\begin{lemma}\label{lem:exponential_inequality}
For all $ 0 \leq x \leq 1 $ and all $ y > 0 $,
\begin{equation*}
\left( 1 + \frac{x}{2 y} \right)^y \leq 1 + x \, .
\end{equation*}
\end{lemma}

\begin{proof}
Let $ e $ denote Euler's constant.
We will use the following well-known inequalities: $ (1 + 1/z)^z \leq e $ (for all $ z > 0 $), $ e^z \leq 1 / (1 - z) $ (for all $ z < 1 $), and $ 1 / (1 - z) \leq 1 + 2 z $ (for all $ 0 \leq z \leq 1/2 $).
We then get:
\begin{equation*}
\left( 1 + \frac{x}{2 y} \right)^y = \left( \left( 1 + \frac{x}{2 y} \right)^{\frac{2 y}{x}} \right)^{\frac{x}{2}} \leq e^{\frac{x}{2}} \leq \frac{1}{1 - \frac{x}{2}} \leq 1 + x \, . \qedhere
\end{equation*}
\end{proof}

\subsection{Running Time}

The running time analysis follows similar arguments as for the $s$-$t$ reachability algorithm in \Cref{sec:running_time_st_reach}.
For every pair of $i$-centers $ (x, y) $, the algorithm never decreases $ r (x, y, i) $ and $ (x, y) $ causes a refresh only if $ r (x, y, i) $ increases, which happens $ O (\log W / \epsilon) $ times.
Furthermore it can be argued in a straightforward way that for every $ 1 \leq i \leq k+1 $ and every pair of $i$-centers $ (x, y) $ the graph $ G [\cQ (x, y, i)] $ has at most $ \min(m / b_i, n^2 / b_i^2) $ edges whp.
Now the only differences compared to the running time analysis of the $s$-$t$ reachability algorithm are the number of path unions of active pairs of centers each node is contained in and the time needed for maintaining the hub links, both of which are analyzed below.

\begin{lemma}\label{lem:number_of_children_containing_node_weighted}
Let $ 1 \leq i \leq k $ and consider a pair of active $(i+1)$-centers $ (x', y') $ and their $i$-children $ (x_j, y_j)_{1 \leq j \leq l} $ (which are active $i$-centers).
Then, for every node $ v $, there are at most $ q = 2^{2i+2} \lceil \log_{1+\epsilon} (n W) \rceil $ pairs of $i$-children $ (x_j, y_j) $ of $ (x', y') $ such that $ v \in \cQ (x_j, y_j, i) $.
\end{lemma}

\begin{proof}
We show that at the last time the algorithm has called {\Refresh{$x'$, $y'$, $i+1$}} (where it determined the current $i$-children of $ (x', y') $) there are at most $ q = 2^{2i+2} \lceil \log_{1+\epsilon} (n W) \rceil $ pairs of $i$-children $ (x_j, y_j) $ of $ (x', y') $ such that
\begin{equation*}
v \in \cP (x_j, y_j, (1 + \epsilon)^{r (x_j, y_j, i) + 2i + 1}, \tilde{G}^{8 h_i / \epsilon, r (x_j, y_j, i)} [\cQ (x', y', i+1)]) \, .
\end{equation*}
Note that, for each $ 1 \leq i \leq l $, $ r (x_j, y_j, i) $ has not changed since the last refresh (as such a change implies refreshing $ (x', y') $).
Thus, for each $ 1 \leq i \leq l $, $ \cQ (x_j, y_j, i) $ is a subset of the path union above, which means that the lemma will be implied by this claim.
Now we actually prove the following slightly stronger claim:
for every every node~$ v $ and every $ 0 \leq r \leq \lfloor \log_{1 + \epsilon} (nW) \rfloor $ , there are at most $ q' = 2^{2i+2} $ pairs of $i$-children $ (x_j, y_j) $ of $ (x', y') $ such that
\begin{equation*}
v \in \cP (x_j, y_j, (1 + \epsilon)^{r + 2i + 1}, \tilde{G}^{8 h_i / \epsilon, r} [\cQ (x', y', i+1)]) \text{ and } \dist_G^{h_i} (x_j, y_j) \geq (1 + \epsilon)^r \, .
\end{equation*}

In this proof we use the abbreviation $ H = \tilde{G}^{8 h_i / \epsilon, r} [\cQ (x', y', i+1)] $.
Suppose that $ v $ is contained in $ q' > 2^{2i+2} $ path unions $ \cP (x_j, y_j, (1 + \epsilon)^{r + 2i + 1}, H) $ of $i$-children $ (x_j, y_j) $ such that $ \dist_G^{h_i} (x_j, y_j) \geq (1 + \epsilon)^r $.
Let $ j_1, \ldots, j_{q'} $ be the corresponding indices and assume without loss of generality that $ j_1 < j_2 \ldots < j_{q'} $.
By the initial random sampling of $i$-centers, every shortest path consisting of $ h_i - 1 $ edges contains an $i$-center whp (\Cref{lem:hitting_set_argument}).
Therefore, for every $ 1 \leq j \leq l $, there is a shortest path between $ x_j $ and $ y_j $ in $ H $ with at most $ h_i $ edges, i.e., $ \dist_H (x_j, y_j) = \dist_H^{h_i} (x_j, y_j) $.
Furthermore, for every $ j \in \{ j_1, \ldots, j_{q'} \} $, we have $ \dist_G^{h_i} (x_j, y_j) \geq (1 + \epsilon)^r $ by our assumption and thus
\begin{multline}
\dist_H (x_{j_1}, y_{j_{q'}}) \geq \sum_{j \in \{ j_1, \ldots, j_{q'} \}} \dist_H (x_j, y_j) 
 = \sum_{j \in \{ j_1, \ldots, j_{q'} \}} \dist_H^{h_i} (x_j, y_j) \\
 \geq \sum_{j \in \{ j_1, \ldots, j_{q'} \}} \dist_G^{h_i} (x_j, y_j)
 \geq \sum_{j \in \{ j_1, \ldots, j_{q'} \}} (1 + \epsilon)^r
 = {q'} (1 + \epsilon)^r \label{eq:distance_outer_centers_lower_bound}
\end{multline}

We now derive an upper bound on $ \dist_H (x_{j_1}, y_{j_{q'}}) $ contradicting this lower bound.
Since $ v $ is contained in the path union $ \cP (x_j, y_j, (1 + \epsilon)^{r + 2i + 1}, H) $, we know by the definition of the path union that $ v $ lies on a shortest path from $ x_{j_1} $ to $ y_{j_1} $ in $ H $ of weight at most $ (1 + \epsilon)^{r + 2i + 1} $ and thus $ \dist_H (x_{j_1}, v) \leq (1 + \epsilon)^{r + 2i + 1} $.
The same argument shows that $ \dist_H (v, y_{j_{q'}}) \leq (1 + \epsilon)^{r + 2i + 1} $.
By the triangle inequality we therefore have
\begin{equation}
\dist_H (x_{j_1}, y_{j_{q'}}) \leq \dist_H (x_{j_1}, v) + \dist_H (v, y_{j_{q'}}) \leq 2 (1 + \epsilon)^{r + 2i + 1} \label{eq:distance_outer_centers_upper_bound}
\end{equation}
By combining Inequalities~\eqref{eq:distance_outer_centers_lower_bound} and~\eqref{eq:distance_outer_centers_upper_bound} we get $ q' \leq 2 (1 + \epsilon)^{2i + 1} $, which is a contradictory statement for $ q' > 2^{2i+2} $ because $ \epsilon \leq 1 $.
Thus, $ v $ is contained in at most $ q' = 2^{2i+2} $ path unions $ \cP (x_j, y_j, (1 + \epsilon)^r, H) $ of $i$-children $ (x_j, y_j) $ such that $ \dist_G^{h_i} (x_j, y_j) \geq (1 + \epsilon)^r $.
\end{proof}

We now bound the time needed for maintaining the hub links as follows.
For every $i$-hub $ z $ (with $ 1 \leq i \leq k $) and every $ 0 \leq r \leq \log_{1 + \epsilon} (nW) $ we maintain both an incoming and an outgoing ES-tree up to depth $ (1 + \epsilon)^{r + 2i + 1} $ in $ \tilde{G}^{8 h_i / \epsilon, r} $.
In $ \tilde{G}^{8 h_i / \epsilon, r} $ every edge weight is a multiple of $ \rho = \epsilon^2 (1 + \epsilon)^r / (8 h_i) $ and thus we can scale down the edge weights by the factor $ 1 / \rho $ and maintain the ES-tree in the resulting integer-weighted graph up to depth
\begin{equation*}
\frac{(1 + \epsilon)^{r + 2i + 1}}{\rho} \leq \frac{2^{2k+1} (1 + \epsilon)^r}{\rho} = \frac{2^{2k+1} (1 + \epsilon)^r 8 h_i}{\epsilon^2 (1 + \epsilon)^r} = O (2^{2k} h_i / \epsilon^2) \, .
\end{equation*}
Thus, maintaining these two trees takes time $ O (2^{2k} m h_i / \epsilon^2) $ which is $ O (2^{2k} m n / (c_i \epsilon^2)) $ as $ h_i = n / c_i $.
As we have $ O (\log_{1+\epsilon} (nW)) = \tilde O (\log{W} / \epsilon) $ such trees for every $i$-hub and there are $ \tilde O (b_i) $ $i$-hubs in, maintaining all these trees takes time $ \tilde O (\sum_{1 \leq i \leq k} 2^{2 k} b_i m n \log{W} / (c_i \epsilon^3)) $.

We now bound the time needed for maintaining the index $ l (x, y, i) $ for every pair of $i$-centers $ (x, y) $.
First, observe that $ l (x, y, i) $ assumes integer values from $ 1 $ to $ |B_i| $ (the number of $i$-hubs) and it only decreases (to $ 0 $) if $ r (x, y, i) $ increases.
The index $ r (x, y, i) $ on the other hand is non-decreasing and assumes integer values from $ 0 $ to $ \lceil \log_{1+\epsilon} (nW) \rceil = \tilde O (\log{W} / \epsilon) $.
As $ |B_i| = \tilde O (b_i) $, the value of $ l (x, y, i) $ therefore changes $ \tilde O (b_i \log{W} / \epsilon) $ times.
As there are $ \tilde O (2^k c_i) $ $i$-centers, the indices of all pairs of centers together change at most $ \tilde O (\sum_{1 \leq i \leq k} 2^{2k} b_i c_i^2 \log{W} / \epsilon) $ times.
It remains to bound the time spent in total for calls of the form \UpdateHubLinks{$x$, $y$, $i$, $r$} in which the index $ l (x, y, i) $ does not increase but the algorithm still spends constant time for checking whether $ l (x, y, i) $ should increase.
In such a case we charge the running time to the pair of $i$-centers $ (x, y) $.
Note that the call \UpdateHubLinks{$x$, $y$, $i$, $r$} has either happened because $ (x, y) $ is made an $i$-child of some pair of $i+1$-centers or because $ (x, y) $ is an active pair of $i$-centers (i.e., in $ A_i $) after the deletion of a node.
We only have to focus on the second case because in the first case the charge of $ O(1) $ on $ (x, y) $ can be neglected.
As argued above, at any time, there are at most $ O (q^k \log_{1 + \epsilon} (nW)) $ active $i$-centers.
Furthermore, there are at most $ m $ deletions in $ G $.
Therefore the total time needed for maintaining all indices $ l (x, y, i) $ is $ \tilde O (\sum_{1 \leq i \leq k} 2^{2k} b_i c_i^2 \log{W} / \epsilon + q^k m \log{W} + \epsilon) $.

Putting everything together, the total running time of our algorithm using $ k $ layers is
\begin{multline*}
\tilde O \left( \sum_{1 \leq i \leq k} \frac{b_i c_i^2 \log W}{\epsilon} + \sum_{1 \leq i \leq k} \frac{b_i m n \log{W}}{c_i \epsilon^3} + q^{2k} m \cdot \min \left( \frac{m}{b_1}, \frac{n^2}{b_1^2} \right)
 \right. \\
\left.
+ \sum_{1 \leq i \leq k-1} q^{2k} c_i^2 \min \left( \frac{m}{b_{i+1}}, \frac{n^2}{b_{i+1}^2} \right) + 2^{2k} q^{2k} c_k^2 m \right)
\end{multline*}
where $ q = 2^{2k+2} \lfloor \log_{1+\epsilon}{W} \rceil $.
Assuming that $ W \leq 2^{\log^c{n}} $ and $ \epsilon \geq 1 / \log^c{n} $ and setting $ k = \lfloor \log^{1/4} {n} \rfloor $ we get $ 2^k q = O (n^{O(\log{\log{n}} / \log^{1/4} {n})}) = O (n^{o(1)}) $.
We now simply set the parameters $ b_i \leq n $ and $ c_i \leq n $ for each $ 1 \leq i \leq k $ in the same way as in \Cref{sec:running_time_st_reach} to obtain the same asymptotic running time (in terms of polynomial factors) as for the $s$-$t$ reachability algorithm.

\begin{theorem}\label{thm:st_shortest_path}
For every $ W \geq 1 $ and every $ 0 < \epsilon \leq 1 $ such that $ W \leq 2^{\log^c{n}} $ and $ \epsilon \geq 1 / \log^c{n} $ for some constant $ c $, there is a decremental $ (1 + \epsilon) $-approximate \sssssp algorithm with constant query time and total update time
\begin{equation*}
O (\min (m^{5/4} n^{1/2 + o(1)}, m^{2/3} n^{4/3 + o(1)})) = O (m n^{6/7 + o(1)})
\end{equation*}
that is correct with high probability against an oblivious adversary.
\end{theorem}

Similar to the $s$-$t$ reachability algorithm we can extend the $ (1 + \epsilon) $-approximate \sssssp algorithm in the following ways.

\begin{corollary}
For every $ W \geq 1 $ and every $ 0 < \epsilon \leq 1 $ such that $ W \leq 2^{\log^c{n}} $ and  $ \epsilon \geq 1 / \log^c{n} $ for some constant $ c $, there is a decremental $ (1 + \epsilon) $-approximate algorithm for maintaining shortest paths between $ p $ source-sink pairs with constant query time and total update time
\begin{equation*}
O (\min ( p^{1/2} m^{5/4} n^{1/2 + o(1)}, p^{1/3} m^{2/3} n^{4/3 + o(1)}) )
\end{equation*}
that is correct with high probability against an oblivious adversary.
\end{corollary}

\begin{corollary}
For every $ W \geq 1 $ and every $ 0 < \epsilon \leq 1 $ such that $ W \leq 2^{\log^c{n}} $ and  $ \epsilon \geq 1 / \log^c{n} $ for some constant $ c $, there is a decremental $ (1 + \epsilon) $-approximate \sssp algorithm with constant query time and total update time 
\begin{equation*}
O (\min (m^{7/6} n^{2/3 + o(1)}, m^{3/4} n^{5/4 + o(1)})) = O (m n^{9/10 + o(1)})
\end{equation*}
that is correct with high probability against an oblivious adversary.
\end{corollary}

\section{Faster Single-Source Reachability in Dense Graphs}\label{sec:SSR_dense}

In this section we first introduce a path union data structure that is more efficient than the naive approach for repeatedly computing path unions between one fixed node and other variable nodes.
We then show how to combine it with a multi-layer approach to obtain a faster decremental single-source reachability algorithm for dense graphs.

\subsection{Approximate Path Union Data Structure}\label{sec:approx_path_union_ds}

In the following we present a data structure for a graph $ G $ undergoing edge deletions, a fixed node~$ x $, and a parameter~$ h $.
Given a node $ y $, it computes an ``approximation'' of the path union $ \cP (x, y, h, G) $.
Using a simple static algorithm the path union can be computed in time $ O(m) $ for each pair $ (x, y) $.
We give an (almost) output-sensitive data structure for this problem, i.e., using our data structure the time will be proportional to the size of the approximate path union which might be $ o(m) $.
Additionally, we have to pay a global cost of $ O (m) $ that is amortized over \emph{all} approximate path union computations for the node $ x $ and \emph{all} nodes $ y $.
This will be useful because in our reachability algorithm we can use probabilistic arguments to bound the size of the approximate path unions.

\begin{proposition}\label{prop:approximate path union}
There is a data structure that, given a graph $ G $ undergoing edge deletions, a fixed node~$ x $, and a parameter~$ h $, provides a procedure \ApproximatePathUnion such that, given sequence of nodes $ y_1, \ldots, y_k $, this procedure computes sets $ F_1, \ldots F_k $ guaranteeing $ \cP (x, y, h, G) \subseteq F_i \subseteq \cP (x, y, (\log{m} + 3) h, G) $ for all $ 1 \leq i \leq k $.
The total running time is $ O(\sum_{1 \leq i \leq k} |F_i| + m) $.
\end{proposition}

\subsubsection{Algorithm Description}

Internally, the data structure maintains a set $ R (x) $ of nodes, initialized with $ R (x) = V $, such that the following invariant is fulfilled at any time:
all nodes that can be reached from $ x $ by a path of length at most~$ h $ are contained in $ R (x) $ (but $ R (x) $ might contain other nodes as well).
Observe that thus $ R (x) $ contains the path union $ \cP (x, y, h, G) $ for every node $ y $.

To gain some intuition for our approach consider the following way of computing an approximation of the path union $ \cP (x, y, h, G) $ for some node $ y $.
First, compute $ B_1 = \{ v \in R(x) \mid \dist_{G[R (x)]} (v, y) \leq h \} $ using a backward breadth-first search (BFS) to $ y $ in  $ G[R(x)] $, the subgraph of~$ G $ induced by $ R (x) $.
Second, compute $ F = \{ v \in R(x) \mid \dist_{G[B_1]} (x, v) \leq h \} $ using a forward BFS from $ x $ in~$ G[B_1] $.
It can be shown that $ \cP (x, y, h, G) \subseteq F \subseteq \cP (x, y, 2 h, G) $.\footnote{Indeed, $ F $ might contain some node $ v $ with $ \dist_G (x, v) = h $ and $ \dist_G (v, y) = h $, but it will not contain any node $ w $ with either $ \dist_G (x, w) > h $ or $ \dist_G (w, y) > h $.}
Given $ B_1 $, we could charge the time for computing $ F $ to the set $ F $ itself, but we do not know how to pay for computing $ B_1 $ as $ B_1 \setminus F $ might be much larger than $ F $.

Our idea is to additionally identify a set of nodes $ X \subseteq \{ v \in V \mid \dist_G (x, v) > h \} $ and remove it from $ R (x) $.
Consider a second approach where we first compute $ B_1 $ as above and then compute $ B_2 = \{ v \in R(x) \mid \dist_{G[R(x)]} (v, y) \leq 2 h \} $ and $ F = \{ v \in R(x) \mid \dist_{G[B_2]} (x, v) \leq h \} $.
It can be shown that $ \cP (x, y, h, G) \subseteq F \subseteq \cP (x, y, 3 h, G) $.
Additionally, all nodes in $ X = B_1 \setminus F $ are at distance more than~$ h $ from $ x $ and therefore we can remove $ X $ from $ R (x) $.
Thus, we can charge the work for computing $ B_1 $ and $ F $ to $ X $ and $ F $, respectively.\footnote{Note that in our first approach removing $ B_1 \setminus F $ would not have been correct as $ F $ was computed w.r.t\ to $ G[B_1] $ and not w.r.t.\ $ G[B_2] $.}
However, we now have a similar problem as before as we do not know whom to charge for computing $ B_2 $.

We resolve this issue by simply computing $ B_i = \{ v \in R(x) \mid \dist_{G[R (x)]} (v, y) \leq i h \} $ for increasing values of $ i $ until we arrive at some $ i^* $ such that the size of $ B_{i^*} $ is at most double the size of $ B_{i^*-1} $.
We then return $ F = \{ v \in R(x) \mid \dist_{G[B_i]} (x, v) \leq h \} $ and charge the time for computing $ B_i $ to $ X = B_{i-1} \setminus F $ and~$ F $, respectively.
As the size of $ B_i $ can double at most $ O (\log{n}) $ times we have $ \cP (x, y, h, G) \subseteq F \subseteq \cP (x, y, O (h \log{n}), G) $, as we show below.
Procedure~\ref{alg:path_unions} shows the pseudocode of this algorithm.
Note that in the special case that $ x $ cannot reach $ y $ the algorithm returns the empty set.
In the analysis below, let $ i^* $ denote the final value of $ i $ before Procedure~\ref{alg:path_unions} terminates.

\begin{procedure}

\caption{ApproximatePathUnion($y$)}
\label{alg:path_unions}

\tcp{All calls of \ApproximatePathUnion{$y$} use fixed $ x $ and $ h $.}

Compute $ B_1 = \{ v \in R(x) \mid \dist_{G[R(x)]} (v, y) \leq h \} $ \tcp{backward BFS \emph{\textbf{to}} $ y $ in subgraph induced by $ R(x) $}
\For{$ i = 2 $ \KwTo $ \lceil \log{m} \rceil + 1 $}{
	Compute $ B_i = \{ v \in R(x) \mid \dist_{G[R(x)]} (v, y) \leq i h \} $ \tcp{backward BFS \emph{\textbf{to}} $ y $ in subgraph induced by $ R(x) $}
	\If{$ | E [B_i] | \leq 2 | E [B_{i-1}] | $}{\label{lin:check_for_size}
		Compute $ F = \{ v \in B_i \mid \dist_{G[B_i]} (x, v) \leq h \} $ \tcp{forward BFS \emph{\textbf{from}} $ x $ in subgraph induced by $ B_i $}
		$ X \gets B_{i-1} \setminus F $, $ R (x) \gets R (x) \setminus X $\;
		\Return{$ F $}\;
	}
}
\end{procedure}

\subsubsection{Correctness}
We first prove Invariant~(I): the set $ R (x) $ always contains all nodes that are at distance at most $ h $ from $ x $ in $ G $.
This is true initially as we initialize $ R (x) $ to be $ V $ and we now show that it continues to hold because we only remove nodes at distance more than $ h $ from $ x $.
\begin{lemma}
If $ R (x) \subseteq \{ v \in V \mid \dist_G (x, v) \leq h \} $, then for every node $ v \in X $ removed from $ R (x) $, we have $ \dist_G (x, v) > h $.
\end{lemma}

\begin{proof}
Let $ v \in X = B_{i^*-1} \setminus F $ and assume by contradiction that $ \dist_G (x, v) \leq h $.
Since $ v \in B_{i^*-1} $ we have $ \dist_{G[R(x)]} (v, y) \leq (i^*-1) h $.
Now consider the shortest path $ \pi $ from $ x $ to~$ v $ in~$ G $, which has length at most $ h $.
By the assumption, every node on $ \pi $ is contained in $ G[R (x)] $.
Therefore, for every node $ v' $ on $ \pi $, we have $ \dist_{G[R(x)]} (v', v) \leq h $ and thus
\begin{equation*}
\dist_{G[R(x)]} (v', y) \leq \dist_{G[R(x)]} (v', v) + \dist_{G[R(x)]} (v, y) \leq h + (i^*-1) h \leq i^* h
\end{equation*}
which implies that $ v' \in B_{i^*} $.
Thus, every node on $ \pi $ is contained in $ B_{i^*} $.
As $ \pi $ is a path from $ x $ to $ v $ of length at most $ h $ it follows that $ \dist_{G[B_{i^*}]} (x, v) \leq h $.
Therefore $ v \in F $, which contradicts the assumption $ v \in X $.
\end{proof}

We now complete the correctness proof by showing that the set of nodes returned by the algorithm approximates the path union.

\begin{lemma}
Procedure~\ref{alg:path_unions} returns a set of nodes $ F $ such that $ \cP (x, y, h, G) \subseteq F \subseteq \cP (x, y, (\log{m} + 3) h, G) $.
\end{lemma}

\begin{proof}
We first argue that the algorithm actually returns some set of nodes $ F $.
Note that in \Cref{lin:check_for_size} of the algorithm we always have $ | E [B_i] | \geq | E [B_{i-1}] | $ as $ B_{i-1} \subseteq B_i $.
As $ E [B_i] $ is a set of edges and the total number of edges is at most $ m $, the condition $ | E [B_i] | \leq | E [B_{i-1}] | $ therefore must eventually be fulfilled for some $ 2 \leq i \leq \lceil \log{m} \rceil + 1 $.

We now show that $ \cP (x, y, h, G) \subseteq F $.
Let $ v \in \cP (x, y, h, G) $, which implies that $ v $ lies on a path $ \pi $ from $ x $ to $ y $ of length at most $ h $.
For every node $ v' $ on $ \pi $ we have $ \dist_G (x, v') \leq h $, which by Invariant~(I) implies $ v' \in R (x) $.
Thus, the whole path $ \pi $ is contained in $ G[R (x)] $.
Therefore $ \dist_{G[R (x)]} (v', y) \leq h $ for every node $ v' $ on $ \pi $ which implies that $ \pi $ is contained in $ G[B_{i^*}] $.
Then clearly we also have $ \dist_{G[B_{i^*}]} (x, v) \leq h $ which implies $ v \in F $.

Finally we show that $ F \subseteq \cP (x, y, (\log{m} + 3) h, G) $ by proving that $ \dist_G (x, v) + \dist_G (v, y) \leq (\log{m} + 3) h $ for every node $ v \in F $. 
As $ G [B_{i^*}] $ is a subgraph of~$ G $, we have $ \dist_G (x, v) \leq \dist_{G[B_{i^*}]} (x, v) $ and $ \dist_G (v, y) \leq \dist_{G[B_{i^*}]} (v, y) $.
By the definition of $ F $ we have $ \dist_{G[B_{i^*}]} (x, v) \leq h $.
As $ F \subseteq B_{i^*} $ we also have $ \dist_{G[B_{i^*}]} (v, y) \leq i^* h \leq (\lceil \log{m} \rceil + 1) h \leq (\log{m} + 2) h $.
It follows that $ \dist_G (x, v) + \dist_G (v, y) \leq {h + (\log{m} + 2) h} = {(\log{m} + 3) h} $.
\end{proof}

\subsubsection{Running Time Analysis}
To bound the total running time we prove that each call of Procedure~\ref{alg:path_unions} takes time proportional to the number of edges in the returned approximation of the path union plus the number of edges incident to the nodes removed from $ R (x) $.
As each node is removed from $ R(x) $ at most once, the time spent on \emph{all} calls of Procedure~\ref{alg:path_unions} is then $ O (m) $ plus the sizes of the subgraphs induced by the approximate path unions returned in each call.

\begin{lemma}\label{lem:running_time_approximate_path_union}
The running time of Procedure~\ref{alg:path_unions} is $ O (| E [F] | + | E [X, R(x)] | + | E [R(x), X] |) $ where $ F $ is the set of nodes returned by the algorithm, and $ X $ is the set of nodes the algorithm removes from $ R(x) $.
\end{lemma}

\begin{proof}
The running time in iteration $ 2 \leq j \leq i^*-1 $ is $ O(| E [B_j] |) $ as this is the cost of the breadth-first-search performed to compute $ B_j $.
In the last iteration $ i^* $, the algorithm additionally has to compute $ F $ and $ X $ and remove $ X $ from $ R (x) $.
As $ F $ is computed by a BFS in $ G[B_{i^*}] $ and $ X \subseteq B_{i^*-1} \subseteq B_{i^*} $, these steps take time $ O(| E [B_{i^*}] |) $.
Thus the total running time is $ O (\sum_{1 \leq j \leq i^*} | E [B_j] |) $.

By checking the size bound in Line~\ref{lin:check_for_size} of Procedure~\ref{alg:path_unions} we have $ |E [B_j] | > 2 | E [B_{j-1}] | $ for all $ 1 \leq j \leq i^*-1 $ and $ | E [B_{i^*}] | \leq 2 | E [B_{i^*-1}] | $.
By repeatedly applying the first inequality it follows that $ \sum_{1 \leq j \leq i^*-1} | E [B_j] | \leq 2 | E [B_{i^*-1}] | $.
Therefore we get
\begin{equation*}
\sum_{1 \leq j \leq i^*} | E [B_j] | = \sum_{1 \leq j \leq i^*-1} | E [B_j] | + | E [B_{i^*}] | \leq 2 | E [B_{i^*-1}] | + 2 | E [B_{i^*-1}] | = 4 | E [B_{i^*-1}] |
\end{equation*}
and thus the running time is $ O (| E [B_{i^*-1}] |) $.
Now observe that by $ X = B_{i^*-1} \setminus F $ we have $ B_{i^*-1} \subseteq X \cup F $ and thus
\begin{align*}
E [B_{i^*-1}] \subseteq E [F] \cup E [X] \cup E [X, F] \cup E [F, X] \subseteq E [F] \cup E [X, R(x)] \cup E [R(x), X] \, .
\end{align*}
Therefore the running time is $ O (| E [F]| + | E [X, R(x)] | + | E [R(x), X] |) $.
\end{proof}

\subsection{Reachability via Center Graph}

We now show how to combine the approximate path union data structure with a hierarchical approach to get an improved decremental reachability algorithm for dense graphs.
The algorithm has a parameter $ 1 \leq k \leq \log{n} $ and for each $ 1 \leq i \leq k $ a parameter $ c_i \leq n $.
We determine suitable choices of these parameters later in this section.
For each $ 1 \leq i \leq k-1 $, our choice will satisfy $ c_i \geq c_{i+1} $ and $ c_i = \hat{O} (c_{i+1}) $.
Furthermore, we set $ h_i = (3 + \log{m})^{i-1} n / c_1 $ for $ 1 \leq i \leq k $.
At the initialization, the algorithm determines sets of nodes $ C_1 \supseteq C_2 \supseteq \dots \supseteq C_k $ such that $ s, t \in C_1 $ as follows.
For each $ 1 \leq i \leq k $, we sample each node of the graph with probability $ a c_i \ln{n} / n $ (for a large enough constant $ a $), where the value of~$ c_i $ will be determined later.
The set $ C_i $ then consists of the sampled nodes, and if $ i \leq k-1 $, it additionally contains the nodes in $ C_{i+1} $.
For every $ 1 \leq i \leq k $ we call the nodes in $ C_i $ $i$-centers.
In the following we describe an algorithm for maintaining pairwise reachability between all $1$-centers.

\subsubsection{Algorithm Description}

\paragraph*{Data Structures.}
The algorithm uses the following data structures:
\begin{itemize}
\item For every $i$-center $ x $ and every $ i \leq j \leq k $ an approximate path union data structure (see \Cref{prop:approximate path union}) with parameter $ h_j $.
\item For every $k$-center $ x $ an incoming and an outgoing ES-tree of depth $ h_k $ in $ G $.
\item For every pair of an $i$-center $ x $ and a $j$-center $ y $ such that $ l := \max(i, j) \leq k-1 $, a set of nodes $ \cQ (x, y, l) \subseteq V $.
Initially, $ \cQ (x, y, l) $ is empty and at some point the algorithm might compute $ \cQ (x, y, l) $ using the approximate path union data structure of $ x $.
\item For every pair of an $i$-center $ x $ and a $j$-center $ y $ such that $ l := \max(i, j) \leq k-1 $ an ES-tree of depth $ h_l $ from $ x $ in $ \cQ (x, y, l) $.
\item For every pair of an $i$-center $ x $ and a $j$-center $ y $ such that $ l := \max(i, j) \leq k-1 $ a set of $(l+1)$-centers certifying that $ x $ can reach $ y $.

\end{itemize}

\paragraph*{Certified Reachability Between Centers (Links).}
The algorithm maintains the following limited path information between centers, called \emph{links}, in a top-down fashion.
Let $ x $ be a $k$-center and let $ y $ be an $i$-center for some $ 1 \leq i \leq k - 1 $.
The algorithm links $ x $ to $ y $ if and only if $ y $ is contained in the outgoing ES-tree of depth $ h_k $ of $ x $.
Similarly the algorithm links $ y $ to $ x $ if and only if $ y $ is contained in the incoming ES-tree of depth $ h_k $ of $ x $.
Let $ x $ be an $i$-center and let $ y $ be a $j$-center such that $ l := \max(i, j) \leq k-1 $.
If there is an $(l+1)$-center $ z $ such that $ x $ is linked to $ z $ and $ z $ is linked to $ y $, the algorithm links $ x $ to $ y $ (we also say that $ z $ links $ x $ to $ y $).
Otherwise, the algorithm computes $ \cQ (x, y, l) $ using the approximate path union data structure of $ x $ and starts to maintain an ES-tree from $ x $ up to depth $ h_l $ in $ G[\cQ (x, y, l)] $.
It links $ x $ to $ y $ if and only if $ y $ is contained in the ES-tree of $ x $.
Using a list of centers $ z $ certifying that $ x $ can reach $ y $, maintaining the links between centers is straightforward.

\paragraph*{Center Graph.}
The algorithm maintains a graph called \emph{center graph}.
Its nodes are the $1$-centers and it contains the edge $ (x, y) $ if and only if $ x $ is linked to $ y $.
The algorithm maintains the transitive closure of the center graph.
A query asking whether a center $ y $ is reachable from a center $ x $ in $ G $ is answered by checking the reachability in the center graph.
As $ s $ and $ t $ are $1$-centers this answers $s$-$t$ reachability queries.

\subsubsection{Correctness}

For the algorithm to be correct we have to show that there is a path from $ s $ to $ t $ in the center graph if and only if there is a path from $ s $ to $ t $ in $ G $.
We will in fact show in more generality that this is the case for any pair of $1$-centers.

\begin{lemma}\label{lem:links_imply_paths}
For every pair of an $i$-center $ x $ and a $j$-center $ y $, if $ x $ is linked to $ y $, then there is a path from $ x $ to $ y $ in $ G $.
\end{lemma}

\begin{proof}
The proof is by induction on $ l = \max (i, j) $.
As $ x $ is linked to $ y $, one of the following three cases applies:
\begin{enumerate}
\item $ j = k $ and $ x $ is contained in the incoming ES-tree of depth $ h_k $ of $ y $ in $ G $.
\item $ i = k $ and $ y $ is contained in the outgoing ES-tree of depth $ h_k $ of $ x $ in $ G $.
\item There is an $ (l+1) $-center $ z $ such that $ x $ is linked to $ z $ and $ z $ is linked to $ y $.
\item $ y $ is contained in the ES-tree of depth $ h_l $ of $ x $ in $ G [\cQ (x, y, l)] $.
\end{enumerate}
In the first two cases there obviously is a path from $ x $ to $ y $ in $ G $ by the correctness of the ES-tree.
In the third case we may apply the induction hypothesis and find a path from $ x $ to $ z $ as well as a path from $ z $ to $ y $ in $ G $.
The concatenation of these paths is a path from $ x $ to $ y $ in $ G $.
In the fourth case we know by the correctness of the ES-tree that there is a path from $ x $ to $ y $ in $ G [\cQ (x, y, l)] $ and therefore also in $ G $.
\end{proof}

\begin{lemma}\label{lem:bounded_hop_paths_imply links}
For every pair of an $i$-center $ x $ and a $j$-center $ y $, if $ \dist_G (x, y) \leq h_l $, then $ x $ is linked to $ y $.
\end{lemma}

\begin{proof}
Set $ l = \max (i, j) $.
If $ l = k $, then assume that $ l = i $ (the proof for $ l = j $ is symmetric).
Since $ \dist_G (x, y) \leq h_k $, $ y $ is contained in the outgoing ES-tree of depth $ h_k $ of $ x $ by the correctness of the ES-tree.
Thus, $ x $ is linked to $ y $.

If $ l \leq k-1 $ and there is an $ (l+1) $-center $ z $ such that $ x $ is linked to $ z $ and $ z $ is linked to $ y $, then also $ x $ is linked to $ y $.
If this is not the case, then the algorithm has computed $ \cQ (x, y, l) $ using the approximate path union data structure and at that time we have $ \cP (x, y, h_l, G) \subseteq \cQ (x, y, l) $ by \Cref{prop:approximate path union}.
By \Cref{pro:path_union_property_under_deletions} the set $ \cQ (x, y, l) $ contains $ \cP (x, y, h_l, G) $, also in the current version of $ G $.
As $ \dist_G (x, y) \leq h_l $, all nodes on the shortest path from $ x $ to $ y $ in $ G $ are contained in $ \cP (x, y, h_l, G) $ and thus in $ \cQ (x, y, l) $.
Therefore the ES-tree from $ x $ up to depth $ h_l $ in $ G [\cQ (x, y, l)] $ contains $ y $ which means that $ x $ is linked to $ y $.
\end{proof}

\begin{lemma}
For every pair of $1$-centers $ x $ and $ y $, there is a path from $ x $ to $ y $ in the center graph if and only if there is a path from $ x $ to $ y $ in $ G $.
\end{lemma}

\begin{proof}
Assume that there is a path from $ x $ to $ y $ in the center graph.
Every edge $ (x', y') $ in the center graph (where $ x' $ and $ y' $ are $1$-centers) can only exist if the algorithm has linked $ x' $ to $ y' $.
By \Cref{lem:links_imply_paths} this implies that there is a path from $ x' $ to $ y' $ in $ G $.
By concatenating all these paths we obtain a path from $ x $ to $ y $ in $ G $.

Now assume that there is a path from $ x $ to $ y $ in $ G $.
Consider the shortest path $ \pi $ from $ x $ to $ y $ in $ G $ and any two consecutive $1$-centers $ x' $ and $ y' $ on $ \pi $.
By the initial random sampling of $1$-centers, every shortest path consisting of $ h_1 - 1 $ edges contains a $1$-center whp (\Cref{lem:hitting_set_argument}) and thus $ \dist_G (x', y') \leq h_1 $.
By \Cref{lem:bounded_hop_paths_imply links} the algorithm has linked $ x' $ to $ y' $ and thus there is an edge from $ x' $ to $ y' $ in the center graph.
As such an edge exists for all consecutive $1$-centers on $ \pi $, there is a path from $ x $ to $ y $ in the center graph.
\end{proof}

\subsubsection{Running Time Analysis}\label{sec:running_time_dense}

The key to the efficiency of the algorithm is to bound the size of the graphs $ \cQ (x, y, l) $.

\begin{lemma}\label{lem:bound_on_size_of_path_union}
Let $ x $ be an $i$-center and let $ y $ be a $j$-center such that $ l := \max (i, j) \leq k-1 $.
If $ x $ is not linked to $ y $ by an $(l+1)$-center, then $ \cQ (x, y, l) $ contains at most $ n / c_{l+1} $ nodes with high probability.
\end{lemma}

\begin{proof}
Suppose that $ \cQ (x, y, l) $ contains more than $ n / c_{l+1} $ nodes.
Then by the random sampling of centers, $ \cQ (x, y, l) $ contains an $(l+1)$-center $ z $ with high probability by \Cref{lem:hitting_set_argument}.
By \Cref{prop:approximate path union} we have $ \dist_G (x, z) + \dist_G (z, y) \leq (3 + \log{m}) h_l = h_{l+1} $ and thus $ \dist_G (x, z) \leq h_{l+1} $ and $ \dist_G (z, y) \leq h_{l+1} $.
It follows that $ x $ is linked to $ z $ and $ z $ is linked to $ y $ by \Cref{lem:bounded_hop_paths_imply links}.
But then $ x $ is linked to $ y $, contradicting our assumption.
\end{proof}

With the help of this lemma we first analyze the running time of each part of the algorithm and argue that our choice of parameters gives the desired total update time.

\paragraph*{Parameter Choice.}
We carry out the running time analysis with regard to two parameters $ 1 \leq b \leq c \leq n $ which we will set at the end of the analysis.
We set
\begin{equation*}
k = \left\lceil \frac{\log{(c/b)}}{\sqrt{\log{n} \cdot \log{\log{n}}}} \right\rceil + 1 \, ,
\end{equation*}
$ c_k = b $ and $ c_{i} = 2^{\sqrt{\log{n} \cdot \log{\log{n}}}} c_{i+1} = \hat{O} (c_{i+1}) $ for $ 1 \leq i \leq k-1 $.
Note that the number of $i$-centers is $ \tilde O (c_i) $ whp.
Observe that
\begin{multline*}
(3 + \log{m})^{k-1} = O ((\log{n})^k) \leq O ((\log{n})^{\sqrt{\log{n} / \log{\log{n}}}}) \\ = O (2^{\sqrt{\log{n} \cdot \log{\log{n}}}}) = O (n^{\sqrt{\log{\log{n}} / \log{n}}}) = O (n^{o(1)}) \, .
\end{multline*}
Furthermore we have
\begin{equation*}
c_1 = \left( 2^{\sqrt{\log{n} \cdot \log{\log{n}}}} \right)^{k-1} c_k \geq 2^{\log{(c / b)}} b = \frac{c}{b} \cdot b = c
\end{equation*}
and by setting $ k' = (\log{(c/b)}) / (\sqrt{\log{n} \cdot \log{\log{n}}}) $ we have $ k \leq k' + 2 $ and thus
\begin{equation*}
c_1 = \left( 2^{\sqrt{\log{n} \cdot \log{\log{n}}}} \right)^{k-1} c_k \leq \left( 2^{\sqrt{\log{n} \cdot \log{\log{n}}}} \right)^{k'+1} c_k = 2^{\sqrt{\log{n} \cdot \log{\log{n}}}} c = \hat O (c) \, .
\end{equation*}
Remember that $ h_i = (3 + \log{m})^{i-1} n / c_1 $ for $ 1 \leq i \leq k $.
Therefore we have $ h_i = \hat{O} (n / c_1) = \hat{O} (n / c) $.

\paragraph*{Maintaining ES-Trees.}
For every $k$-center we maintain an incoming and an outgoing ES-tree of depth $ h_k $, which takes time $ O (m h_k) $.
As there are $ \tilde O (c_k) $ $k$-centers, maintaining all these trees takes time $ \tilde O (c_k m h_k) = \hat O (b m n / c) $.

For every $i$-center $ x $ and every $j$-center $ y $ such that $ l := \max (i, j) \leq k-1 $, we maintain an ES-tree up to depth $ h_l $ in $ G[\cQ (x, y, l)] $.
By \Cref{lem:bound_on_size_of_path_union} $ \cQ (x, y, l) $ has at most $ n / c_{l+1} $ nodes and thus $ G[\cQ (x, y, l)] $ has at most $ n^2 / c_{l+1}^2 $ edges.
Maintaining this ES-tree therefore takes time $ O ((n^2 / c_{l+1}^2) \cdot h_l) = \hat O (n^2 / c_{l+1}^2 (n / c_1)) = \hat O (n^3 / (c_1 c_{l+1}^2)) $.
In total, maintaining all these trees takes time
\begin{multline*}
\hat{O} \left( \sum_{1 \leq i \leq k-1} \sum_{1 \leq j \leq i} c_i c_j \frac{n^3}{c_1 c_{i+1}^2} \right) =
\hat{O} \left( \sum_{1 \leq i \leq k-1} \sum_{1 \leq j \leq i} \frac{c_i c_1 n^3}{c_{i+1} c_1 c_k} \right) \\
= \hat{O} \left( \sum_{1 \leq i \leq k-1} \sum_{1 \leq j \leq i} \frac{n^3}{c_k} \right)
= \hat{O} \left( k^2 \frac{n^3}{c_k} \right) = \hat{O} \left(\frac{n^3}{b} \right) \, .
\end{multline*}

\paragraph*{Computing Approximate Path Unions.}
For every $i$-center $ x $ and every $ i \leq j \leq k $ we maintain an approximate path union data structure with parameter $ h_j $.
By \Cref{prop:approximate path union} this data structures has a total running time of $ O (m) $ and an additional cost of $ O (|E [\cQ (x, y, j]) |) $ each time the approximate path union $ \cQ (x, y, j) $ is computed for some $j$-center $ y $.
By \Cref{lem:bound_on_size_of_path_union} the number of nodes of $ \cQ (x, y, j) $ is $ n / c_{j+1} $ with high probability and thus its number of edges is $ n^2 / c_{j+1}^2 $.
Therefore, computing all approximate path unions takes time
\begin{multline*}
\tilde O \left( \sum_{1 \leq i \leq k-1} \sum_{i \leq j \leq k} \left( c_i m + c_i c_j \frac{n^2}{c_{j+1}^2} \right) \right) =
\tilde O \left( \sum_{1 \leq i \leq k-1} \sum_{i \leq j \leq k} \left( c_1 m + \frac{c_1 c_j n^2}{c_{j+1} c_k} \right) \right) \\
= \hat O \left( \sum_{1 \leq i \leq k-1} \sum_{i \leq j \leq k} \left( c_1 m + \frac{c_1 n^2}{c_k} \right) \right)
= \hat O ( k^2 c_1 m + k^2 c_1 n^2 / c_k )
= \hat O ( c m + c n^2 / b ) \, .
\end{multline*}

\paragraph*{Maintaining Links Between Centers.}
For each pair of an $i$-center $ x $ and a $j$-center $ y $ there are at most $ \tilde O (c_{l+1}) $ $(l+1)$-centers that can possibly link $ x $ to $ y $.
Each such $(l+1)$-center is added to and removed from the list of $(l+1)$-centers linking $ x $ to $ y $ at most once.
Thus, the total time needed for maintaining all these links is $ \tilde O (\sum_{1 \leq i \leq k-1} \sum_{1 \leq j \leq i} c_i c_j c_{i+1} ) = \tilde O (k^2 c_1^3) = \tilde O (c^3) $.

\paragraph*{Maintaining Transitive Closure in Center Graph.}
The center graph has $ \tilde O (c_1) $ nodes and thus $ \tilde O (c_1^2) $ edges.
During the algorithm edges are only deleted from the center graph and never inserted.
Thus we can use known $ O(mn) $-time decremental algorithms for maintaining the transitive closure~\cite{RodittyZ08,Lacki13} in the center graph in time $ \tilde O (c_1^3) = \tilde O (c^3) $.

\paragraph*{Total Running Time.}
Since the term $ c n^2 / b$ is dominated by the term $ n^3 / b $, we obtain a total running time of $\hat{O} \left( b m n / c + n^3 / b + c m + c^3 \right) $.
By setting $ b = n^{5/3} / m^{2/3} $ and $ c = n^{4/3} / m^{1/3} $ the running time is $ \hat{O} (m^{2/3} n^{4/3} + n^4 / m) $ and by setting $ b = n^{9/7} / m^{3/7} $ and $ c = m^{1/7} n^{4/7} $ the running time is $ \hat{O} (m^{3/7} n^{12/7} + m^{8/7} n^{4/7}) $.

\subsubsection{Decremental Single-Source Reachability}

The algorithm above works for a set of randomly chosen centers.
Note that the algorithm stays correct if we add any number of nodes to $ C_1 $, thus increasing the number of $1$-centers for which the algorithm maintains pairwise reachability.
If the number of additional centers does not exceed the number of randomly chosen centers, then the same running time bounds still apply.
Thus, from the analysis above, we obtain the following result.

\begin{theorem}\label{pro:decremental_bounded_all_pairs_reachability}
Let $ S \subseteq V $ be a set of nodes.
If $ |S| \leq n^{4/3} / m^{1/3} $, then there is a decremental algorithm for maintaining pairwise reachability between all nodes in $ S $ with constant query time and a total update time of $ \hat{O} (m^{2/3} n^{4/3} + n^4 / m) $ that is correct with high probability against an oblivious adversary.
If $ |S| \leq m^{1/7} n^{4/7} $, then there is a decremental algorithm for maintaining pairwise reachability between all nodes in $ S $ with constant query time and total update time
$ \hat{O} (m^{3/7} n^{12/7} + m^{8/7} n^{4/7}) $
that is correct with high probability against an oblivious adversary.
\end{theorem}


Using the reduction of \Cref{thm: s-t to single source} this also gives us a  single-source reachability algorithm.

\begin{corollary}
There are decremental \ssr and \scc algorithms with constant query time and total update time
\begin{equation*}
\hat{O} (m^{2/3} n^{4/3} + m^{3/7} n^{12/7})
\end{equation*}
that are correct with high probability against an oblivious adversary.
The total update time of the \scc algorithm is only in expectation.
\end{corollary}

\begin{corollary}
There is a decremental \ssr algorithm with constant query time and total update time
\begin{equation*}
\hat{O} (m^{2/3} n^{4/3} + m^{3/7} n^{12/7})
\end{equation*}
that is correct with high probability against an oblivious adversary.
\end{corollary}

\begin{proof}
First, set $ c_1 = n^{4/3} / m^{1/3} $ and observe that by \Cref{pro:decremental_bounded_all_pairs_reachability} we have an algorithm that can maintain reachability from a source to $ c_1 $ sinks with a total update time of $ \hat{O} (m^{2/3} n^{4/3} + n^4 / m) $.
By \Cref{thm: s-t to single source} this implies a decremental single-source reachability algorithm with a total update time of $ \hat{O} (m^{2/3} n^{4/3} + n^4 / m + m n / c_1) $.
The same argument gives a decremental single-source reachability algorithm with a total update time of $ \hat{O} (m^{3/7} n^{12/7} + m^{8/7} n^{4/7} + m n / c_2) $ where $ c_2 = m^{1/7} n^{4/7} $.

If $ m \geq n^{8/5} $ we have $ n^4 / m \leq m^{2/3} n^{4/3} $ and $ m n / c_1 = m^{4/3} / n^{1/3} \leq m^{2/3} n^{4/3} $.
If $ m \leq n^{8/5} $ we have $ m^{8/7} n^{4/7} \leq m^{3/7} n^{12/7} $ and $ m n / c_2 = m^{6/7} / n^{3/7} \leq m^{3/7} n^{12/7} $.
Thus, by running the first algorithm if $ m \geq n^{8/5} $ and the second one if $ m < n^{8/5} $ we obtain a total update time of $ \hat{O} (m^{2/3} n^{4/3} + m^{3/7} n^{12/7}) $ (Note that $ m^{3/7} n^{12/7} \leq m^{2/3} n^{4/3} $ if and only if $ m \geq n^{8/5} $).
\end{proof}

Furthermore, the reduction of \Cref{thm:strongly connected component} gives a decremental algorithm for maintaining strongly connected components.

\begin{corollary}
There is a decremental \scc algorithm with constant query time and expected total update time
\begin{equation*}
\hat{O} (m^{2/3} n^{4/3} + m^{3/7} n^{12/7})
\end{equation*}
that is correct with high probability against an oblivious adversary.
\end{corollary}

\section{Conclusion}

In this paper we have presented decremental algorithms for maintaining approximate single-source shortest paths (and thus also single-source reachability) with constant query time and a total update time of $ o (m n) $. This result implies that the thirty-year-old algorithm with $O(mn)$ total update time is not the best possible.
Since it is hard to believe that the running time obtained in this paper is tight, it is an important open problem to search for faster (and perhaps simpler) algorithms. 
Note that decremental approximate \sssp in \emph{undirected} graphs can be solved with almost linear total update time~\cite{HenzingerKNFOCS14}. Thus, it is interesting to see if the same total update time can be obtained for the directed case. Getting this result for $s$-$t$ reachability will already be a major breakthrough. 

Given recent progress in lower bounds for dynamic graph problems~\cite{Patrascu10,AbboudW14,HenzingerKNS15}, it might also be possible that an almost linear update time is not possible for decremental approximate \sssp and \ssr in directed graphs.
However, it might be challenging to prove this as all known constructions give the same lower bounds for both the incremental and the decremental version of a problem.
As incremental \ssr has a total update time of $ O (m) $ we cannot hope for lower bounds for decremental \ssr using these known techniques.
This only leaves the possibility of finding a lower bound for decremental approximate \ssr or of developing stronger techniques.
It also motivates studying the incremental approximate \sssp problem which is intuitively easier than its decremental counterpart, but no as easy as incremental \ssr.

Finally, we ask whether it is possible to remove the assumption that the adversary is oblivious, i.e., to allow an \emph{adaptive adversary} that may choose each new update or query based on the algorithm's previous answers.
This would immediately be guaranteed by a deterministic algorithm.

\printbibliography[heading=bibintoc] 

\end{document}